\documentclass[11pt]{article}
\usepackage[utf8]{inputenc}
\usepackage[T1]{fontenc}
\usepackage[margin=1in]{geometry}
\usepackage{microtype}
\usepackage{amsmath,amssymb,amsfonts,amsthm,mathtools,bm}
\usepackage{graphicx,booktabs,tabularx}
\usepackage{xcolor}
\usepackage{url}
\usepackage[numbers,sort&compress]{natbib}
\definecolor{linkblue}{HTML}{1B4F72}
\usepackage[
  unicode=true,
  pdfencoding=auto,
  psdextra,
  hypertexnames=false,
  bookmarks=true,
  bookmarksnumbered=true,
  bookmarksopen=true,
  bookmarksopenlevel=2,
  pdfpagemode=UseOutlines,
  pdfpagelayout=OneColumn,
  linktoc=all,
  breaklinks=true,
  colorlinks=true,
  linkcolor=linkblue,
  citecolor=linkblue,
  urlcolor=linkblue
]{hyperref}
\usepackage{enumitem}
\usepackage{float}
\usepackage{algorithm}
\usepackage{algorithmic}

\newif\iffastfigures
\fastfiguresfalse
\DeclareGraphicsExtensions{.pdf,.png,.jpg,.jpeg}
\iffastfigures
  \graphicspath{{figures_overleaf/}{./figures_overleaf/}{figures_neurips/}{./figures_neurips/}{figures_ordered/}{./figures_ordered/}{figures_supplementary/}{./figures_supplementary/}}
\else
  \graphicspath{{figures_neurips/}{./figures_neurips/}{figures_ordered/}{./figures_ordered/}{figures_supplementary/}{./figures_supplementary/}{figures_overleaf/}{./figures_overleaf/}}
\fi
\newcommand{\figwidth}{\textwidth}

\newcommand{\E}{\mathbb{E}}
\newcommand{\Var}{\mathrm{Var}}
\newcommand{\Cov}{\mathrm{Cov}}
\newcommand{\Corr}{\mathrm{Corr}}
\newcommand{\MI}{\mathrm{MI}}
\newcommand{\tr}{\mathrm{tr}}
\newcommand{\diag}{\mathrm{diag}}
\newcommand{\1}{\mathbf{1}}

\newcommand{\xe}{X_e}
\newcommand{\xI}{X_i}
\newcommand{\we}{\bm w_e}
\newcommand{\wi}{\bm w_i}

\newcommand{\dprime}{\ensuremath{{d'}}}
\newcommand{\add}{\mathrm{add}}

\newcommand{\sh}{\mathrm{sh}}

\newenvironment{ack}{\section*{Acknowledgments}}{}

\newtheorem{theorem}{Theorem}
\newtheorem{lemma}{Lemma}
\newtheorem{proposition}{Proposition}
\newtheorem{corollary}{Corollary}

\newif\ifneuripsanonymous
\neuripsanonymousfalse

\newif\ifneuripscompact
\neuripscompactfalse

\newif\ifneuripscameraready
\neuripscamerareadyfalse

\newif\ifpapersynthesis
\papersynthesistrue

\renewcommand{\NoHyper}{}
\renewcommand{\endNoHyper}{}

\hypersetup{
  pdftitle={When Branch-Local Shunting Helps: A Gain-Load-Alignment Principle for Dendritic E/I Networks},
  pdfauthor={Houman Safaai; Maceo Richards; Naeem Khoshnevis; Bernardo L. Sabatini},
  pdfsubject={Dendritic computation, shunting inhibition, morphology, and population coding},
  pdfkeywords={dendrites, shunting inhibition, divisive normalization, population coding, morphology, gain control, neural networks},
  pdfdisplaydoctitle=true
}

\title{When Branch-Local Shunting Helps: A Gain-Load-Alignment Principle for
Dendritic E/I Networks}

\author{%
  Houman Safaai\,$^{1,\ast}$ \enskip Maceo Richards\,$^{1}$ \enskip Naeem Khoshnevis\,$^{1}$ \enskip Bernardo L. Sabatini\,$^{1,2,\ast}$ \\[0.5em]
  $^{1}$Kempner Institute for the Study of Natural and Artificial Intelligence \\
  at Harvard University \\
  $^{2}$Department of Neurobiology, Howard Hughes Medical Institute, \\
  Harvard Medical School, Boston, MA 02115, USA \\[0.3em]
  {\normalfont\footnotesize $^{\ast}$Correspondence:
  \texttt{houman\_safaai@harvard.edu},
  \texttt{bernardo\_sabatini@hms.harvard.edu}}
}
\date{}

\begin{document}


\maketitle

\begin{abstract}
Biological neurons combine excitatory and inhibitory synaptic activity on
branched dendrites through shunting-based interactions such that inhibition effectively
divisively attenuates
excitation. Whether this improves population readout over additive excitatory and
inhibitory (E/I) interactions with the
same nonnegative inputs remains unclear. We introduce DendriNet, a trainable
framework that varies the synaptic integration rule, dendritic morphology, synapse allocation,
divisor locality, and dendritic nonlinearities. For population codes corrupted by
multiplicative gain, linearizing any realizable shunting readout around its
operating point yields a decision direction within the positive additive E/I
cone; matching the additive optimum requires a positive self-consistent
shunting realization. Every scalar shunting threshold also has an exact affine
additive realization.
Beyond this local limit, performance follows a gain--load--alignment
principle. Branch-local shunting helps when a reliable divisor acts on the
nuisance-affected signal pathway and suppresses signal-aligned gain more than
it attenuates signal or introduces denominator variability. Passive additive
trees flatten to linear readouts, whereas shunting trees compose local
divisors. In a designed passive hierarchy, deep shunting exceeds tangent and
fitted-linear controls. Flexible nonlinear predictors overtake it with enough
labels, while misspecified divisor support weakens its small-sample advantage.
Support shuffling reverses the linear comparisons; sensor corruption
reverses only the fitted-linear comparison. Contact- and parameter-matched activated
DendriNet training shows no consistent benefit from greater depth. The same
support and reliability interaction appears in a frozen-feature normalization
test. Modeling downstream readout of primary visual cortex (V1) recordings from three
mice shows that the
shunting-over-additive decoder gap is
largest for narrow readouts, reverses under strong private noise at the widest
readout, and varies across running states. Morphology can therefore determine
where reliable nuisance estimates meet task-relevant signals, but neither depth
nor shunting is intrinsically advantageous.
\end{abstract}

\setlength{\textfloatsep}{5pt plus 2pt minus 2pt}
\setlength{\floatsep}{5pt plus 2pt minus 2pt}
\setlength{\abovecaptionskip}{3pt}
\setlength{\belowcaptionskip}{0pt}
\renewcommand{\bottomfraction}{0.65}

\NoHyper
\section{Introduction}
\label{sec:intro}

Most artificial neurons compute a weighted sum of their inputs with signed coefficients,
followed by a pointwise nonlinearity. Biological neurons instead integrate
nonnegative excitatory and inhibitory conductances across branched dendrites,
allowing inhibition to divide excitation locally before signals reach the soma.
This difference suggests a possible computational advantage but does not
guarantee one: division can suppress shared variability but can also
attenuate signal and import noise through its denominator. We therefore ask a
matched question: when does branch-local shunting outperform an additive E/I
readout given the same nonnegative inputs, synaptic resources, and morphology?

Divisive normalization is a canonical strategy for gain control and invariance
~\citep{Heeger1992,ReynoldsHeeger2009,CarandiniHeeger2012}, and conductance-based
inhibition provides a plausible dendritic implementation
~\citep{ChanceAbbottReyes2002,MitchellSilver2003,LondonHausser2005,Silver2010}.
At the same time, whether dendritic shunting produces a truly \emph{divisive}
somatic effect rather than a largely subtractive one has long been debated, and
its firing-rate gain effect can depend on synaptic noise and dendritic saturation
~\citep{HoltKoch1997,PrescottDeKoninck2003}. We therefore do not assume a
shunting advantage. Instead, we examine the conditions in which branch-local division helps under
contact- and parameter-matched optimization.

The question is narrower than the established result that dendritic
nonlinearities and morphology can expand single-neuron or network computation
~\citep{Poirazi2003,PolskyEtAl2004,Beniaguev2021,Chavlis2025,
LyoSavin2024,AgrawalBuice2025}. Here, morphology controls local access and determines whether a
nuisance estimate can interact with the affected signal before somatic pooling.
A dendritic tree can thus act as a within-neuron hierarchy that transforms
stimulus-by-state activity before it reaches the downstream decoder. The placement of inhibitory synapses
and the operating point determine which signals and divisors interact
~\citep{GidonSegev2012,JarvisEtAl2018,ZhangEtAl2024DynamicGain,
StokesTeeterIsaacson2014,BlossEtAl2016,YangMurrayWang2016}.

We use noisy population-code readout as a controlled evaluation. Shared
contrast, gain, and brain-state fluctuations are widely observed
~\citep{NiellStryker2010,McGinleyEtAl2015,GorisMovshonSimoncelli2014,
RabinowitzEtAl2015}; when they align with the stimulus direction in Fisher
geometry, they can limit a constrained downstream decoder
~\citep{Averbeck2006,CohenKohn2011,MorenoBote2014,Kafashan2021,
DeneveMachens2016}. A branch-local divisor may suppress multiplicative gain but
can also erase signal or import variability. We call
non-target pool activity \emph{denominator load}: its mean $L_0$ moves the
operating point, while residual fluctuation $\delta L$ contributes $\Lambda$.
Analogues include background conductance and mask, surround, or distractor drive
~\citep{ChanceAbbottReyes2002,FellousEtAl2003,
CarandiniHeegerMovshon1997,BusseWadeCarandini2009,CoenCagliSolomon2019}.

Prior work shows that division can decorrelate activity or improve information
transfer, whereas a stochastic pool can add covariance and reverse that benefit
~\citep{SchwartzSimoncelli2001,BalleLaparraSimoncelli2016,Tripp2012,
WeissEtAl2023}. Learned pools and conductance-based compartments can implement
adaptive or Bayes-optimal weighting~\citep{BurgEtAl2021,NiBurge2024,
JordanEtAl2024}, whereas flexible additive networks can learn related weights
~\citep{FarahmandiEtAl2025}. These strands do not identify which input
statistics let local division outperform a matched additive readout. The
gain--load--alignment principle addresses that gap by connecting Fisher
geometry, anatomical branch support, and signal fidelity.

Our contributions are threefold. First, we give exact scalar and first-order
representational boundaries for additive and shunting E/I readouts. Second, we
derive a finite-amplitude gain--load--alignment criterion and test its predicted
support, reliability, and morphology reversals. Third, we provide a trainable
research framework that separates E/I positivity, synaptic integration rules, dendritic morphology,
synaptic allocation, divisor support, dendritic nonlinearities, and population access
(Fig.~\ref{fig:framework-main}A--B). We use this framework to model and analyze designed tasks,
conventional learned features, and a bounded downstream V1 decoding problem.

\begin{figure}[t]
\centering
\includegraphics[width=\figwidth]{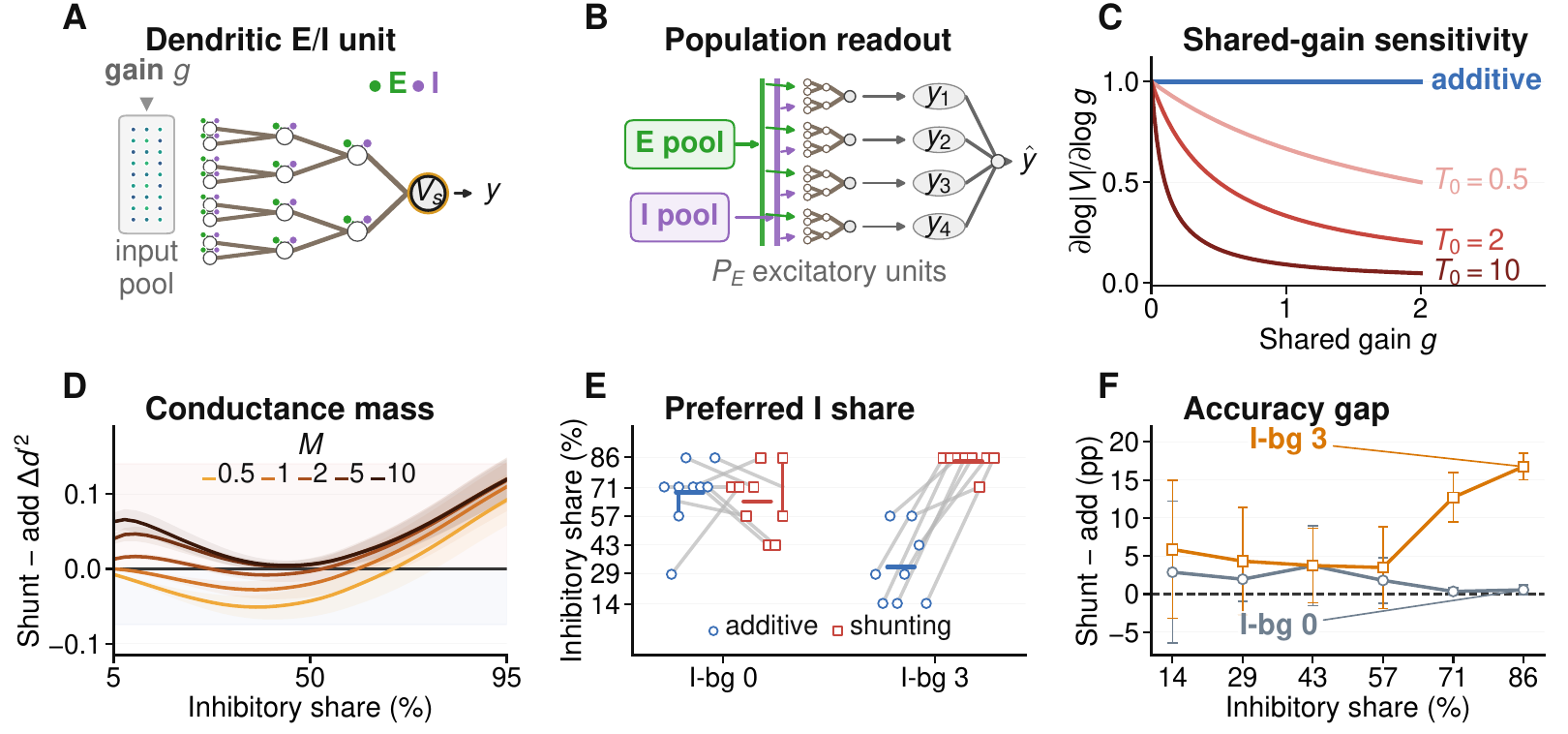}
\caption{\textbf{DendriNet and resource controls.}
\textbf{(A)}~One unit combines nonnegative E/I drive additively or through a
local shunting divisor. \textbf{(B)}~Branched units form a population decoder.
\textbf{(C)}~Passive additive log-gain elasticity is 1; shunting elasticity is
$[1+gT_0]^{-1}$. \textbf{(D)}~A continuous mass constraint
$M=\lVert w\rVert_1$ maps the local mechanism gap across total mass and
inhibitory share (360 generated libraries). \textbf{(E--F)}~A separate
fixed-contact budget gives validation-selected inhibitory shares (E) and
held-out paired gaps (F); it is computational, not a biological E/I ratio.
Passive voltage; E--F use paired seeds ($n=8$, 95\% $t$ intervals).}
\label{fig:framework-main}
\end{figure}

Each comparison names what is held fixed. Mechanism-isolating controls share
input samples, nuisance draws, and morphology while matching a stated quantity
such as conductance mass, contact count, topology and parameters, or operating
point. Complete trained-model comparisons instead give each mechanism the same
validation-only search budget---the same grid and tuning seeds---and evaluate
the frozen choices on an untouched test set. Detailed optimization controls
are reported in Methods and the Supplement. The comparator ladder runs from fixed $E-I$ through
tangent-additive and optimized positive-E/I readouts to fitted
$E-\rho I-b$ and sign-unconstrained linear decoders (Methods).

The central claim is therefore conditional: depth helps only when routing and
nuisance suppression outweigh signal-fidelity and denominator-variability
costs. The analysis of activity in V1 examines whether this computational effect transfers to
observed population activity; it does not identify a cortical circuit.
\endNoHyper
\section{Results}

The argument proceeds from a local limit to branch composition and data. We
first analyze one shunting branch, then test finite-amplitude and morphological
effects in simulations, trained DendriNet models, conventional frozen
features, and mouse V1 recordings.
Throughout, the controlled input model for the upstream population activity
$X$ is
\begin{equation}
  X=g f(\theta)+\varepsilon,
  \label{eq:latent-gain-model}
\end{equation}
for nonnegative $X$, stimulus $\theta$, mean tuning $f(\theta)$, and shared
gain $g$ with $\E[g]=1$ and $\Var(g)=\sigma_g^2$; the residual noise
$\varepsilon$ is independent of $g$. Its stimulus-conditional covariance is
\begin{equation}
  \Sigma(\theta)=\Sigma_{\mathrm{ind}}
  +\sigma_g^2 f(\theta)f(\theta)^\top.
  \label{eq:latent-covariance}
\end{equation}
Low-rank variability is limiting only when it overlaps the signal and remains
inaccessible to the constrained decoder~\citep{WakhlooSlattonChung2026}.
Private noise is represented by $\Sigma_{\rm ind}$, whereas shared gain or state
contributes the low-rank mode
~\citep{GorisMovshonSimoncelli2014,RabinowitzEtAl2015}. Inaccessible differential
noise can impose an information ceiling~\citep{MorenoBote2014,Kafashan2021}.
Normalization-pool fluctuations contribute the residual denominator-variance
term $\Lambda$. Observation noise instead corrupts the measured inputs without
specifying a circuit source. We analyze these sources separately because
division need not affect them in the same way.

\subsection{Dendritic model and local matched-comparison boundary}
\label{subsec:population-problem}
\label{subsec:model}

We first examine whether a single shunting branch can represent a decision direction
unavailable to an additive readout of the same nonnegative inputs.
For excitatory and inhibitory conductances $(E_n,I_n)$ at node $n$ and voltages
$V_b$ of its children $b\in\mathcal D_n$, normalizing conductances by the
implemented leak gives the exact unit-leak passive-node equation
\begin{equation}
  V_n=\frac{E_n+\sum_{b\in\mathcal D_n}g_{n\leftarrow b}V_b}
  {1+E_n+I_n+\sum_{b\in\mathcal D_n}g_{n\leftarrow b}},
  \qquad g_{n\leftarrow b}\ge0.
  \label{eq:node-voltage}
\end{equation}
A terminal node gives $E/(1+E+I)$; the additive comparator is $E-I$ at a leaf
and a nonnegative linear current-mode recursion in a passive tree. Rates,
pathway maps, weights, and couplings are nonnegative, although an additive
hidden voltage may be signed (raw-unit details are in Methods).

The above equation gives the steady-state voltage produced by synaptic conductances and,
in biological terms, reflects passive synaptic integration. An optional
secondary nonlinearity (termed a ``reactivation'') would model an active dendrite
and can further shape the signal before it is transmitted downstream. Theory
and designed controls use raw passive voltages unless stated otherwise;
activated trained panels apply the same shifted-tanh family after each branch
voltage in both mechanisms. Captions and Methods give each panel's setting.

Define the numerator, total conductance, child contribution, and child coupling
by
\begin{equation}
  N_n=E_n+C_n,\quad T_n=E_n+I_n+G_n,\quad
  C_n=\sum_b g_bV_b,\quad G_n=\sum_b g_b,\quad
  V_n^{\sh}=\frac{N_n}{1+T_n}.
  \label{eq:node-NT}
\end{equation}
At any positive anchor $(N_0,T_0)$, with $D_0=1+T_0$ and $V_0=N_0/D_0$,
direct algebra gives
\begin{equation}
 D_0(V^{\sh}-V_0)=\frac{D_0}{1+T}\left[\delta N-V_0\delta T\right].
 \label{eq:node-gated-contrast}
\end{equation}
The bracket is the tangent-matched additive contrast
\begin{equation}
 A_{\rm tan}=(1-V_0)\delta E-V_0\delta I+\delta C-V_0\delta G.
 \label{eq:tangent-matched-node}
\end{equation}
Thus, the finite-amplitude distinction is exactly a sample-dependent total-drive
gate, not a new local contrast.

\begin{proposition}[Every scalar shunting threshold is an additive comparator]
\label{prop:scalar-threshold}
For $V^{\sh}=E/(1+E+I)$ and $0<\tau<1$,
\begin{equation}
 V^{\sh}>\tau\quad\Longleftrightarrow\quad
 E-\rho_\tau I>b_\tau,\qquad
 \rho_\tau=b_\tau=\frac{\tau}{1-\tau}.
 \label{eq:scalar-threshold-equivalence}
\end{equation}
The identity is distribution free. It does not imply identical ranking or
calibration across all thresholds, nor does it cover compositions of several
branch ratios.
\end{proposition}

For the scalar comparisons below, we quantify discriminability by
\begin{equation}
 \dprime^2(V)=
 \frac{(\E[V\mid1]-\E[V\mid0])^2}
 {\tfrac12\Var(V\mid0)+\tfrac12\Var(V\mid1)}.
 \label{eq:dprime}
\end{equation}
Because this is a moment-based readout metric, the non-Gaussian experiments also
report held-out area under the receiver-operating-characteristic curve (AUC),
train-fitted threshold error, and categorical log loss. We do not
interpret $d'^2$ as mutual information.

\label{subsec:local}
Linearizing a terminal shunt at its induced operating point gives
\begin{equation}
 \delta V^{\sh}\approx
 \frac{1+I_0}{D_0^2}\delta E-
 \frac{E_0}{D_0^2}\delta I,\qquad D_0=1+E_0+I_0.
 \label{eq:shunting-linearization}
\end{equation}
The effective coordinates are a positive diagonal image of the additive E/I
cone. The image must, however, be mapped back to weights that induce the same
positive operating point.
Cone statements use signed input coordinates
$x=(\delta E,-\delta I)$: hence $q=(q_E,q_I)\ge0$ denotes
$q_E^\top\delta E-q_I^\top\delta I$.

\begin{theorem}[Terminal local containment and conditional equivalence]
\label{thm:local-equivalence}
For one terminal pooled E/I compartment, let $\mathcal C_{\sh}^{\rm loc}$ be
the effective Jacobian rays generated by nonnegative shunting weights at
positive self-consistent operating points, and let
$\mathcal C_{\add}=\mathbb R_+^{d_E+d_I}$. Under a valid first-order
linearization, $\mathcal C_{\sh}^{\rm loc}\subseteq\mathcal C_{\add}$. With
nondegenerate covariance on feasible active subspaces,
$\dprime^{2,*}_{\sh,\mathrm{loc}}\le\dprime^{2,*}_{\add}$. Equality holds when
an additive-optimal ray, optimized over both signal orientations, has a positive
self-consistent shunting realization. Equal effective rays generally use
different weights and operating points.
\end{theorem}
The Methods section provides a bound on the curvature remainder. The theorem applies to the local
Jacobian; finite-amplitude $d'^2$ can differ through curvature and branch
composition.

\begin{corollary}[Monotone output activation]
\label{cor:activation-local-tie}
A differentiable scalar activation applied only after the terminal voltage,
with positive slope at the operating point, preserves this first-order result.
\end{corollary}

Internal activation instead multiplies each pathwise passive derivative by
the intervening branch slopes $\gamma_n^q=(\phi_n^q)'(\bar V_n^q)$,
$q\in\{\add,\sh\}$; a terminal slope cancels from local $d'^2$, whereas
internal slopes reweight paths (Methods). For $Y=\phi(V(g))>0$, the exact gain
elasticity is $A_\phi(V)\,\partial\log|V|/\partial\log g$, where
$A_\phi=V\phi'/\phi$. For shifted-tanh slope $\kappa_\phi$,
$A_\phi(V)=\kappa_\phi V[1-\tanh(\kappa_\phi(V-b))]$, so the activated additive
elasticity need not be flat. A strictly increasing terminal activation preserves
sample ordering, AUC, and optimal scalar-threshold error, but midpoint mismatch
can change moment-based effects (Supplementary
Fig.~\ref{fig:activation-sensitivity-supp}).

\begin{proposition}[Constrained local tie]
\label{prop:constraint-equivariant}
For feasible set $\mathcal F$ and positive diagonal map $D$, an added constraint
preserves the fixed-operating-point comparison when the generated ray sets
agree,
$\mathcal R(D\mathcal F)=\mathcal R(\mathcal F)$; a tie still requires a feasible,
self-consistent realization of an additive-optimal ray.
\end{proposition}

The local boundary leaves four possible sources of separation: finite-amplitude
gating, ranking or calibration, additional constraints, and composition across
branches. In a fixed 160-case numerical audit, 90 additive optima admit
self-consistent shunting realizations and tie exactly; the remaining 70 fail
that realizability test
(Supplementary Fig.~\ref{fig:local-equivalence-supp}).
Continuous conductance constraints and discrete contact allocation answer
different questions. Fixing total conductance preserves rays under a fixed
diagonal map, whereas fixing E/I share may exclude them.
Figure~\ref{fig:framework-main}D shows the continuous conductance-mass control;
Fig.~\ref{fig:framework-main}E--F separately shows how a fixed contact budget
changes the preferred allocation and held-out mechanism gap.

\subsection{Beyond the local limit: gain suppression versus denominator variability}
\label{subsec:beyond-local}

Beyond the local limit, shunting helps only when it removes signal-aligned gain
more strongly than it attenuates signal or imports denominator variability.
Figure~\ref{fig:regime-map}A--B organizes this competition using gain alignment
$\alpha$, surviving gain $M_g$, baseline fidelity $B$, and residual denominator
variability $\Lambda$; formal definitions follow below.

\ifpapersynthesis
\begin{figure}[!t]
\else
\begin{figure}[tbp]
\fi
\centering
\includegraphics[width=\figwidth]{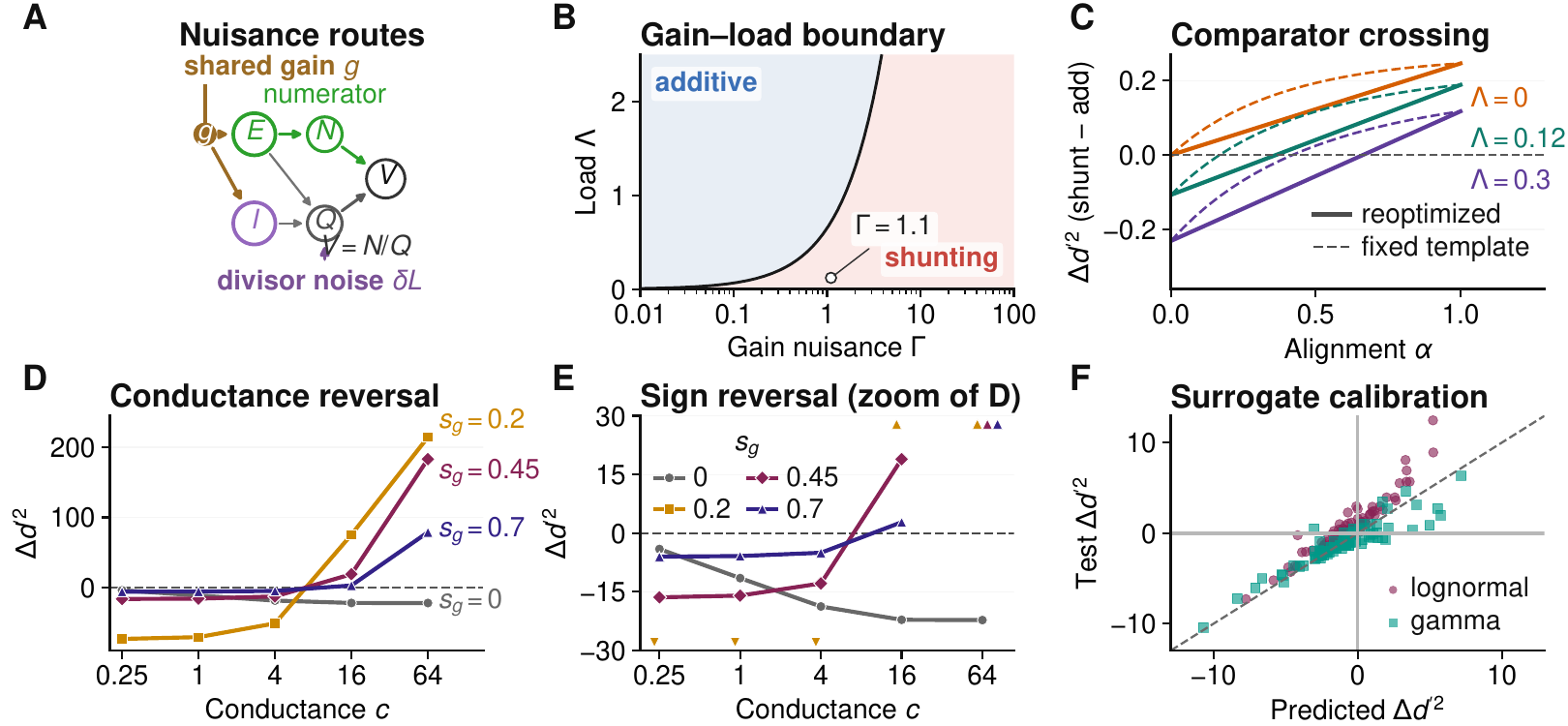}
\caption{\textbf{Gain suppression competes with load and signal loss.}
\textbf{(A--B)}~Shared gain reaches numerator and divisor; independent pool
noise reaches only the divisor. At $B=\alpha=1$ and $M_g=0.35$, the boundary is
$(1-M_g)\Gamma=\Lambda$. \textbf{(C)}~Across $\alpha\in[0,1]$ at
$\Gamma=1.1$,
reoptimization (solid) shifts the
fixed-template crossing (dashed). \textbf{(D)}~Conductance and gain change
the mechanism ordering on identical E/I observations. \textbf{(E)}~The
near-zero band of D, where the ordering actually flips; carets mark samples
that leave the band. \textbf{(F)}~A split-sample surrogate predicts 133/166
held-out signs. Passive voltage is used throughout.}
\label{fig:regime-map}
\end{figure}

If common gain multiplies both pooled drives, the shunt
$gE/[1+g(E+I)]$ becomes progressively less sensitive to gain as conductance
increases. In the low-drive limit, by contrast, its log-gain elasticity returns
to one (Fig.~\ref{fig:framework-main}C). More generally, cancellation requires matched fractional
changes in the nuisance $z$:
$\partial_z\operatorname{logit}V^{\sh}=\partial_z\log E-\partial_z\log(1+I)$.
Common scaling of $E$ and $I$ approaches exact cancellation only when leak is
negligible relative to $I$. A divisor that removes gain can also attenuate task
signal, so distribution shape, comparator choice, and evaluation metric can
reverse the ordering. For the scalar stimulus coordinate $\theta$, the gain--signal
alignment is
\begin{equation}
 \alpha(\theta)=
 \frac{(f'(\theta)^\top\Sigma_{\rm ind}^{-1}f(\theta))^2}
 {(f'(\theta)^\top\Sigma_{\rm ind}^{-1}f'(\theta))
  (f(\theta)^\top\Sigma_{\rm ind}^{-1}f(\theta))}\in[0,1],
 \label{eq:alignment-alpha}
\end{equation}
with dimensionless gain strength
$\Gamma(\theta)=\sigma_g^2 f(\theta)^\top\Sigma_{\rm ind}^{-1}f(\theta)$; we write
$G_{\rm eff}=\alpha\Gamma$. Denominator fluctuations instead add a
penalty proportional to $(\bar E/D_0^2)^2\sigma_\xi^2$, where $\xi=\delta L$ is
the residual denominator fluctuation; positive
numerator--denominator covariance can partially cancel it.

We next separate three effects of the divisor: preservation of the baseline
signal, suppression of aligned gain, and propagation of residual denominator
variability. Let $d_0^{\prime2}(r)$ be the nuisance-free discriminability of the
optimized passive additive comparator in input-statistics regime $r$. Under
the direction-preserving surrogate, define the surviving gain factor
$M_g(\mathcal T,r)$ for tree $\mathcal T$, whose branch-voltage vector is
$h_{\sh}$ with linear readout weights $w_{\sh}$, by
$\Var_g(w_{\sh}^{\top}h_{\sh})/
\Var_{\rm ind}(w_{\sh}^{\top}h_{\sh})=\alpha M_g(\mathcal T,r)\Gamma$ whenever
$\alpha\Gamma>0$; the product is zero when $\alpha\Gamma=0$. Let
$\Lambda_{\mathcal T}(r)$ denote the accumulated residual denominator-variance penalty,
and let $B(\mathcal T,r)=d_{0,\sh}^{\prime2}(\mathcal T,r)/d_{0,\add}^{\prime2}(r)$ measure
baseline signal fidelity. Holding the baseline decoder fixed after adding
nuisance gives the first-order morphology-to-readout approximation
\begin{equation}
 d_{\add}^{\prime2}\approx\frac{d_0^{\prime2}}{1+G_{\rm eff}},\qquad
 d_{\sh}^{\prime2}\approx
 \frac{B(\mathcal T,r)d_0^{\prime2}}
 {1+M_g(\mathcal T,r)G_{\rm eff}+\Lambda_{\mathcal T}(r)}.
 \label{eq:regime-caricature}
\end{equation}
The corresponding criterion for a fixed-template shunting advantage is
\begin{equation}
 B(\mathcal T,r)(1+G_{\rm eff}(r))>
 1+M_g(\mathcal T,r)G_{\rm eff}(r)+\Lambda_{\mathcal T}(r).
 \label{eq:tree-criterion}
\end{equation}
Equivalently, define the signed \emph{shunting-advantage margin}
\begin{equation}
 \mathcal A(\mathcal T,r)
 :=B(1+G_{\rm eff})-[1+M_gG_{\rm eff}+\Lambda_{\mathcal T}]
 =(B-1)+(B-M_g)G_{\rm eff}-\Lambda_{\mathcal T}.
 \label{eq:advantage-margin}
\end{equation}
The leading-order prediction is $\mathcal A>0$. If the nuisance is orthogonal to
the signal ($\alpha=0$), suppressing it cannot improve this readout. With the
other quantities fixed, stronger aligned gain favors shunting only when
$B>M_g$. Independent denominator variability lowers the margin, and a divisor
that suppresses nuisance can still lose if it also reduces $B$. For a deeper
topology $\mathcal T'$ relative to $\mathcal T$,
\begin{equation}
 \Delta\mathcal A=(1+G_{\rm eff})\Delta B
 -G_{\rm eff}\Delta M_g-\Delta\Lambda_{\mathcal T}.
 \label{eq:depth-advantage-margin}
\end{equation}
Thus depth helps shunting when improved multiscale routing and gain suppression
outweigh the added attenuation and load.
For a tree with internal activation, Eqs.~\eqref{eq:regime-caricature}--
\eqref{eq:depth-advantage-margin} retain this bookkeeping form only after
$B$, $M_g$, and $\Lambda_{\mathcal T}$ are recomputed from the activated representations
(or their Jacobians); passive values cannot be carried over.  Saturated internal
slopes can suppress nuisance, reduce signal fidelity, or do both.

The comparison must expose both mechanisms to the same nuisance inputs. For any
fixed additive
comparator $a$, define its isolated normalized gain and pool terms
$G_a=d_{0,a}^{\prime2}/d_{g,a}^{\prime2}-1$ and
$L_a=d_{0,a}^{\prime2}/d_{l,a}^{\prime2}-1$, and define $G_{\sh},L_{\sh}$
analogously. If the reciprocal-discriminability penalties from gain and pool
noise are additive, the comparator-aware margin is
\begin{equation}
 \mathcal A_a^{\rm match}
 =B_a(1+G_a+L_a)-(1+G_{\sh}+L_{\sh}),
 \qquad B_a=\frac{d_{0,\sh}^{\prime2}}{d_{0,a}^{\prime2}}.
 \label{eq:matched-input-margin}
\end{equation}
This expression reduces to Eq.~\eqref{eq:advantage-margin} when $B_a=B$,
$G_a=G_{\rm eff}$, $L_a=0$, $G_{\sh}=M_gG_a$, and
$L_{\sh}=\Lambda_{\mathcal T}$. Because the tangent-additive control also
receives the corrupted sensor (the inhibitory stream reporting the nuisance),
its signed load term $L_a$ must be included in the comparison.

A reoptimized decoder can recover signal components orthogonal to the nuisance,
so we distinguish it from the fixed-template calculation. For rank-one Gaussian
covariance, the additive identity below holds for any positive-definite
independent covariance. The shunting identity additionally assumes an
isotropic, direction-preserving surrogate with $c_{\mathcal T}=1+\Lambda_{\mathcal T}$.
Writing $R^{\rm opt}$ for the fraction of clean discriminability an optimized
comparator retains,
\begin{equation}
 R_{\add}^{\rm opt}=1-\frac{\alpha\Gamma}{1+\Gamma},\qquad
 R_{\sh}^{\rm opt}=\frac{B}{c_{\mathcal T}}\left(1-\frac{\alpha M_g\Gamma}{c_{\mathcal T}+M_g\Gamma}\right).
 \label{eq:noise-resilience-reoptimized}
\end{equation}
Within each mechanism, its fixed-template and reoptimized expressions coincide
at $\alpha=1$ but can differ strongly otherwise
(Fig.~\ref{fig:regime-map}C); positive-E/I cone constraints
require the corresponding numerical optimization.
The Appendix gives an exact Cox/shared-gain additive penalty and states the
additional direction-preserving assumptions used for the shunting surrogate;
positive non-Gaussian families are tested numerically rather than inferred from
Gaussian moments.

A split-sample second-order surrogate tracks the held-out nonlinear effects,
especially when those effects are large (Fig.~\ref{fig:regime-map}F). One
sample fits the direction, a second calibrates its moments, and an independent
third tests the nonlinear effect. Positive non-Gaussian controls further show that the outcome
depends on the input distribution, evaluation metric, and comparator
(Supplementary Fig.~\ref{fig:nongaussian-supp}).

Increasing conductance can reverse the mechanism ordering when shared gain is
accessible to the divisor (Fig.~\ref{fig:regime-map}D,E). The benefit grows
with population access because shunting more strongly reduces the effective
shared covariance. Zero gain or mismatched sensing removes the advantage,
whereas sensor load weakens it
(Supplementary Fig.~\ref{fig:identified-ceiling-supp}). A contact-capped
control also approaches, but cannot exceed, the information ceiling set by its
inputs.

\subsection{Morphology changes locality and resource allocation}
\label{subsec:trees}

The criterion predicts that morphology matters not because depth creates
capacity by itself, but because branch support determines where a divisor meets
the nuisance it must suppress (Fig.~\ref{fig:locality}A). We test this
prediction using balanced,
level-regular trees that isolate routing effects without ranking biological
morphologies. Prior work establishes morphology-dependent computation
~\citep{Poirazi2003,Beniaguev2021,Chavlis2025,LyoSavin2024,
AgrawalBuice2025}; here morphology determines how nuisance information reaches
the branches that carry the affected signal. The per-branch caps $N_E$ and
$N_I$ count contacts rather than E/I cells. Somatic width is $P_E$, and the main
experiments use no learned inhibitory population ($P_I=0$).

\begin{figure}[tbp]
\centering
\includegraphics[width=\figwidth]{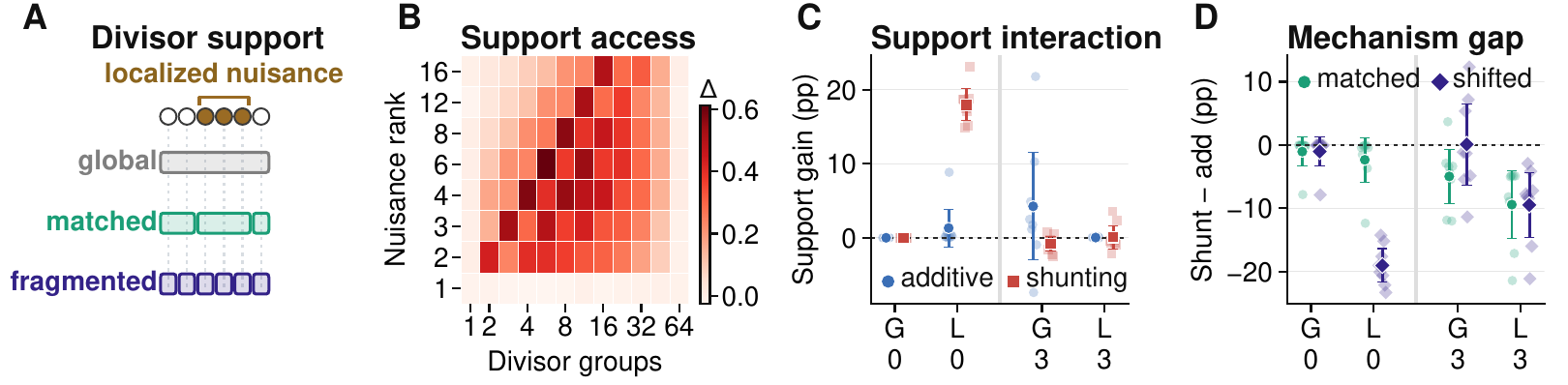}
\caption{\textbf{Divisor support must match the nuisance.}
\textbf{(A)}~Global, matched, and fragmented divisors repartition the same
nuisance-bearing inputs. \textbf{(B)}~Matched-minus-shuffled linear
accessibility across 96 computed oracle conditions; the benefit peaks on the
diagonal, where divisor granularity equals nuisance rank.
\textbf{(C--D)}~Matching inhibitory
support preferentially improves shunting at zero I-stream background (I-bg),
without creating a robust absolute shunting win; G/L denote global/local
nuisance support, C distinguishes mechanisms, and D support placements.
Passive voltage; paired computational
seeds ($n=8$, 95\% $t$ intervals).}
\label{fig:locality}
\end{figure}

Accessibility peaks where divisor granularity equals nuisance rank and vanishes
for one global or 64 fragmented divisors (Fig.~\ref{fig:locality}B). The
locality factorial confirms the same requirement in trained models: inhibitory
support must match the nuisance-bearing pathway (Fig.~\ref{fig:locality}C--D).
All non-rule settings are held fixed, and the support-by-mechanism interaction
is not an overall shunting advantage.

To understand how routing then changes with depth, note that a passive additive
tree has no nonlinear composition: it flattens exactly to
\begin{equation}
 V_s^{\add}=\sum_{n\in\mathcal T}a_n^{\add}(E_n-I_n).
 \label{eq:add-tree-flatten}
\end{equation}
For fixed gains and couplings, $a_n^{\add}$ depends on path depth, so morphology
changes the realized additive score. If local linear weights are freely
reoptimized over the same observations, nonzero path gains can instead be
absorbed; under those assumptions additive capacity is topology-equivalent.
Linearized shunting instead has operating-point-dependent path gains and local
E/I sensitivity ratios,
\begin{equation}
 \delta V_s^{\sh}\approx
 \sum_{n\in\mathcal T}a_n^{\sh}(\delta E_n-\rho_n\delta I_n).
 \label{eq:sh-tree-linearize}
\end{equation}
Along each path, local divisor sensitivities determine the surviving gain, while
somatic sensitivities to node-specific loads determine $\Lambda_{\mathcal T}$. Signal
fidelity $B$ must be measured separately because a signal-bearing divisor can
perform poorly even when $M_g$ is small.

\begin{figure}[!t]
\centering
\includegraphics[width=\figwidth]{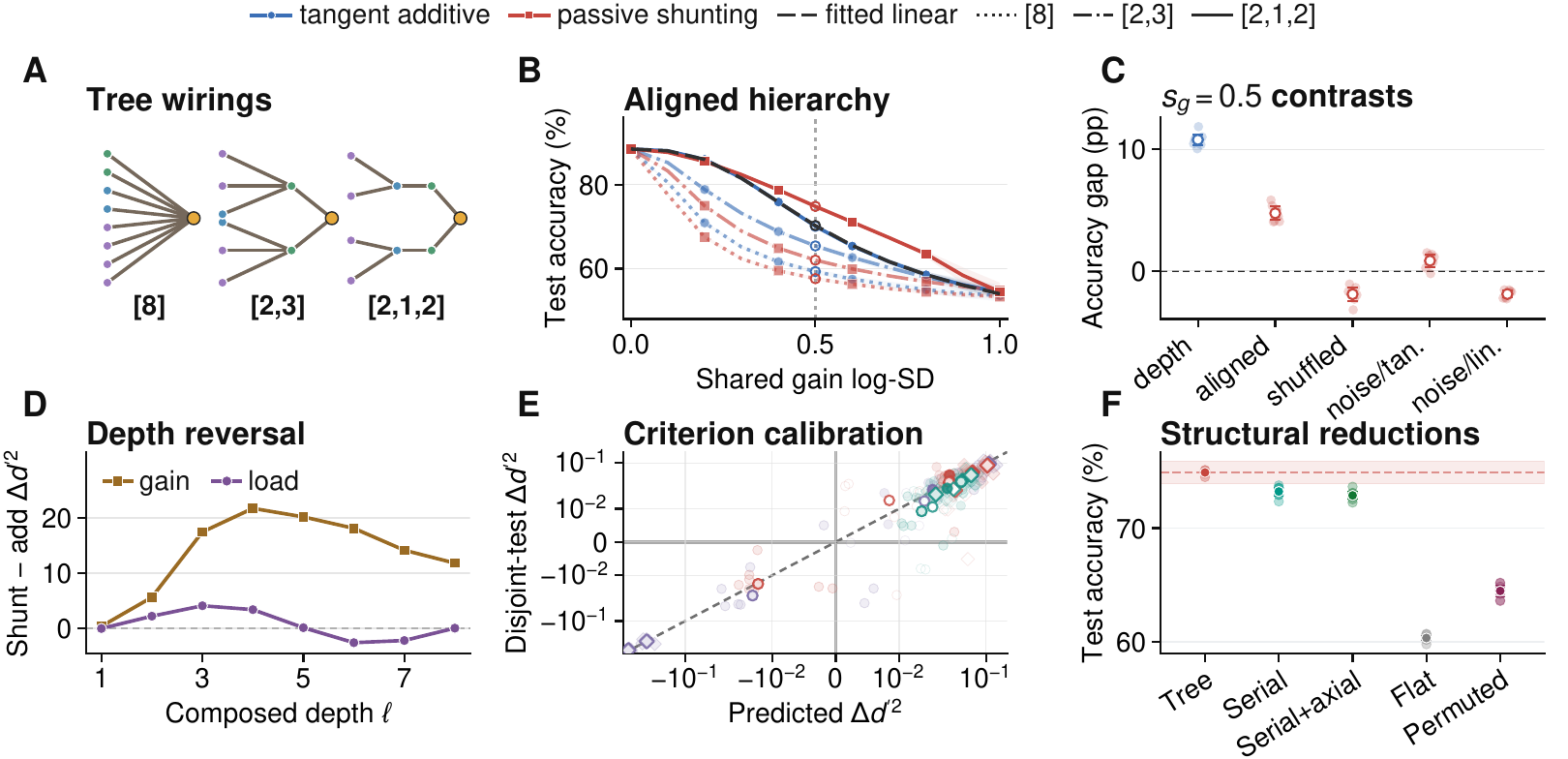}
\caption{\textbf{Aligned hierarchical routing can suppress shared gain.}
\textbf{(A)}~The same eight observations are routed through flat $[8]$, shallow
$[2,3]$, and deep $[2,1,2]$ trees; purple/blue/green mark fine/coarse/global
inputs. \textbf{(B)}~Only the deep shunting
composition crosses above its tangent and the fitted-linear decoder.
\textbf{(C)}~Paired contrasts show the aligned benefit and its support/noise
reversals. \textbf{(D)}~A controlled sweep shows that depth helps under gain
but reverses under load. \textbf{(E)}~A calibrated margin predicts held-out ordering;
purple/red/teal, gain log-SD $0.2/0.5/0.8$; circles/diamonds,
tangent/fitted-linear; filled/open, aligned/shuffled support; faint/large,
paired seeds/condition means.
\textbf{(F)}~Structural reductions of the same construction: the tree is not
reproduced by the tested serial, serial-plus-axial, flat, or level-permuted
reductions, which share its observations, divisor multiset and decoder. This
isolates structure, not biological branch alignment, which the
global-versus-local gain experiment addresses; shading marks the prespecified
$\pm1$ percentage-point tree-equivalence margin. Every dendritic arm in this
figure uses passive branch voltage, so no panel tests an activated tree.
B--C and E use paired computational seeds
($n=8$, paired 95\% $t$ intervals); D is a standalone simulation with
10,000 trials per class.}
\label{fig:criterion-morphology}
\end{figure}

Aligned hierarchical routing strongly suppresses within-class gain while node,
observation, and edge counts remain fixed
(Fig.~\ref{fig:criterion-morphology}A--C). Three comparators address distinct
questions: fixed recursion tests the implemented computation, the tangent at the
clean operating point tests the local approximation, and the fitted
sign-unconstrained readout provides a stronger linear baseline. From flat to
deep, both shunting and its tangent suppress within-class gain. Shunting is
below the tangent for the flat and shallow trees and exceeds it only in the
deep composition; the deep shunting readout also exceeds the fitted-linear
decoder by 4.57 percentage points at gain log-SD $0.5$. Shuffling support
reverses that comparison to $-2.08$ points
(Fig.~\ref{fig:criterion-morphology}B--C). The fixed additive cascade
shows the implemented raw $E-I$ rule under the same inventory, but it is
diagnostic rather than the primary capacity baseline: its subtractive current
drives performance near chance here. The fitted-linear decoder is the stronger
capacity comparison
(Supplementary Fig.~\ref{fig:morphology-noise-diagnostics-supp}). This ordering
is stable across the tested operating points and stronger linear comparators
(Supplementary Fig.~\ref{fig:exact-inventory-operating-sensitivity-supp}).

\ifpapersynthesis
Passive controls further bound this designed example. The
structure-matched shunting readout shows a support-dependent small-sample
accuracy advantage, but validation-selected nonlinear predictors overtake it
with more labels (Supplementary Fig.~\ref{fig:capacity-reduction-supp}A,B), and
log loss crosses in the same bracket
(Supplementary Fig.~\ref{fig:capacity-reduction-supp}C). The full tree also exceeds the tested
serial, flat, and level-permuted reductions
(Fig.~\ref{fig:criterion-morphology}F). DendriNet post-voltage activation is
off in the dendritic arms, while ordinary nonlinear predictors act on the raw
observations; these controls do not establish an ordering for activated
trained trees.
\fi

Serial shunting helps when gain removal outruns accumulated load, but becomes
neutral or harmful when load dominates (Fig.~\ref{fig:criterion-morphology}D).
This controlled sweep does not match parameters or contacts across depth;
complete diagnostics are in Supplementary Fig.~\ref{fig:tree-depth}.

For panel E, we fit templates on clean training data, estimate
comparator-specific terms on calibration data, and test combined corruption on
fresh samples. The margin correctly predicts which mechanism performs better for
the tested conditions (Fig.~\ref{fig:criterion-morphology}E). Omitting the
tangent load term substantially weakens the prediction. Across morphologies,
the shunting gain penalty falls while signed load rises. These proxies are
motivated by, but do not estimate, $(B,M_g,\Lambda_{\mathcal T})$.

Under contact- and parameter-matched training, greater depth shows no consistent
held-out
benefit (Supplementary Fig.~\ref{fig:resource-matched-morph-supp}). This
complete-model comparison reinforces that support and the balance among $M_g$,
$\Lambda_{\mathcal T}$, and $B$, not depth alone, determine performance.

\subsection{Training reveals conditional rule effects}
\label{subsec:ei-circuit}

The designed hierarchy establishes a regime in which aligned shunting helps,
but it does not show that training will discover it.
In the contextual task, a cue selects one of two streams while gain perturbs
the irrelevant stream (Fig.~\ref{fig:ei-circuit}A); the ordered-E/I task shifts
E and I oppositely by class and varies I-stream background. We train both mechanisms with shifted-tanh branches, but
optimization can still place them at responsive or saturated operating points.

\begin{figure}[t]
\centering
\includegraphics[width=\figwidth]{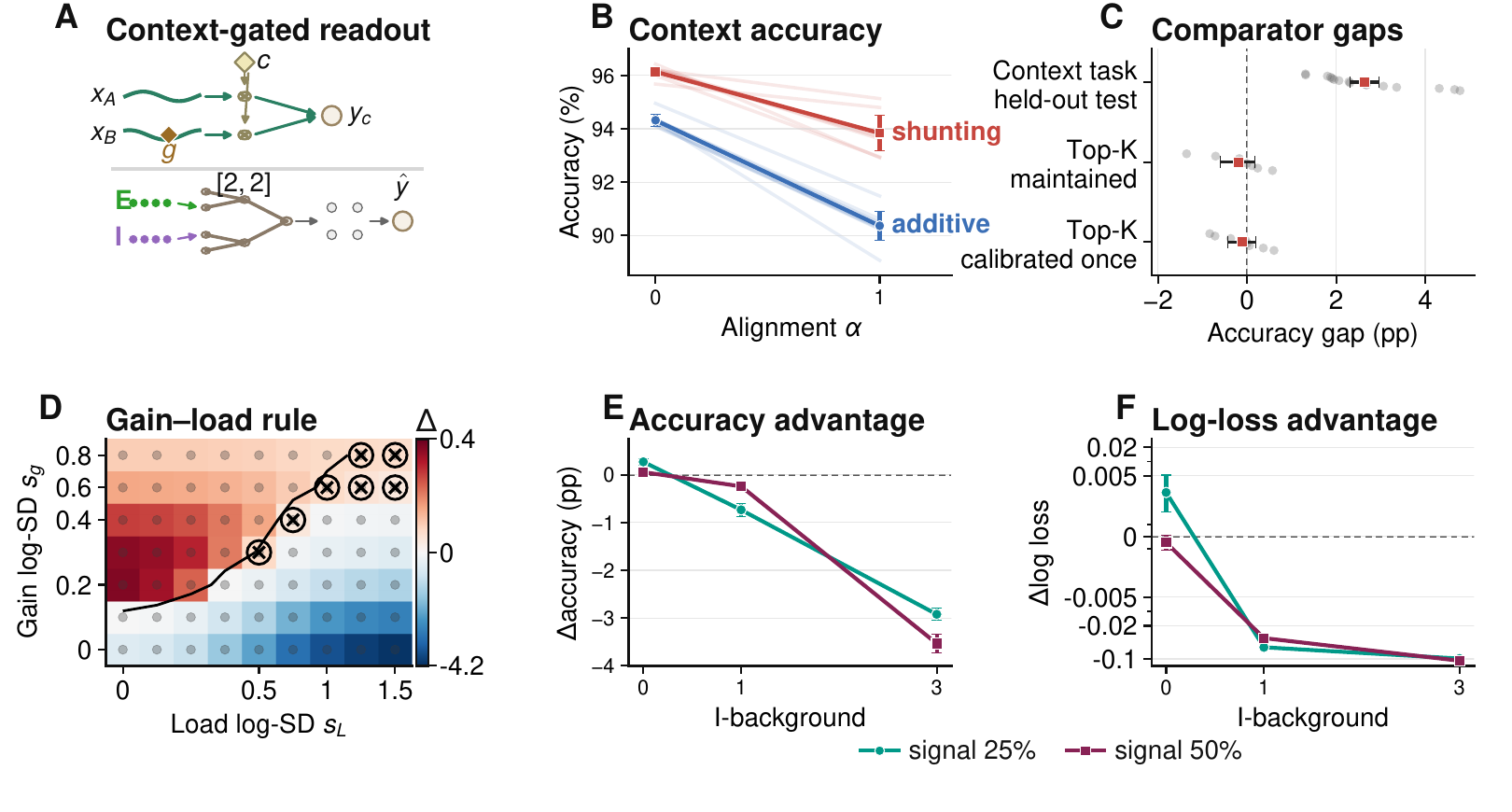}
\caption{\textbf{Training gives conditional effects and a load reversal.}
\textbf{(A)}~Cue $c$ selects target $y_c$ from two streams under gain $g$;
below, a $[2,2]$ DendriNet predicts $\hat y$. \textbf{(B)}~Selected shunting
and additive accuracies at the prespecified alignment endpoints
$\alpha=0,1$ show the absolute contextual-task effect.
\textbf{(C)}~Paired complete-model and matched-operating-point gaps distinguish
trained performance from rule-isolating controls. \textbf{(D)}~Passive Monte Carlo shows
shunting-minus-additive $\Delta d'^2$ over gain and load; grey dots, evaluated
cells; black curve, delta-method boundary; circled crosses, sign
disagreements (56/63 agreement). \textbf{(E--F)}~Accuracy and log-loss
advantages at 25/50\% signal support are oriented so positive favors shunting.
Shifted-tanh is used in B--C,E--F; D is passive. Paired computational seeds ($n=8$;
95\% $t$ intervals, except percentile bootstrap in C).}
\label{fig:ei-circuit}
\end{figure}

For the contextual comparison, we selected one setup per mechanism by validation
accuracy under the same search budget and froze both choices before held-out
testing. Shunting improves held-out accuracy in this selected comparison
(by 2.64 percentage points; Fig.~\ref{fig:ei-circuit}B--C), whereas the
advantage disappears under Top-$K$ policies matched in contacts and activation
calibration. Full model-selection controls are reported
in Supplementary Figs.~\ref{fig:contextual-factorial-supp}
and~\ref{fig:topk-comparator-ladder-supp}.

We next increased I-stream background in the prespecified ordered-E/I task to
test the predicted load-driven reversal. Increasing background changes the
ordering from approximately neutral to additive-favored, and the delta-method
boundary captures most of the Monte Carlo regime map
(Fig.~\ref{fig:ei-circuit}D--F). Denominator-only controls and progressively
stronger additive comparators further delimit this effect (Supplementary
Figs.~\ref{fig:denominator-only-load-supp}
and~\ref{fig:topk-comparator-ladder-supp}).

To test whether the same nuisance logic extends beyond dendritic trees, we
freeze conventional ResNet-18 feature maps on CIFAR-10
~\citep{HeEtAl2016,Krizhevsky2009CIFAR}. Divisor-aligned group gain separates
the mean-only divisive readout from an unnormalized additive readout, whereas
GroupNorm~\citep{WuHe2018GroupNorm} is comparably robust. Shuffling gain
support weakens the separation, and corrupting the divisor eventually reverses
it (Fig.~\ref{fig:cifar-normalization-main}A--D).

\begin{figure}[t]
\centering
\includegraphics[width=\figwidth]{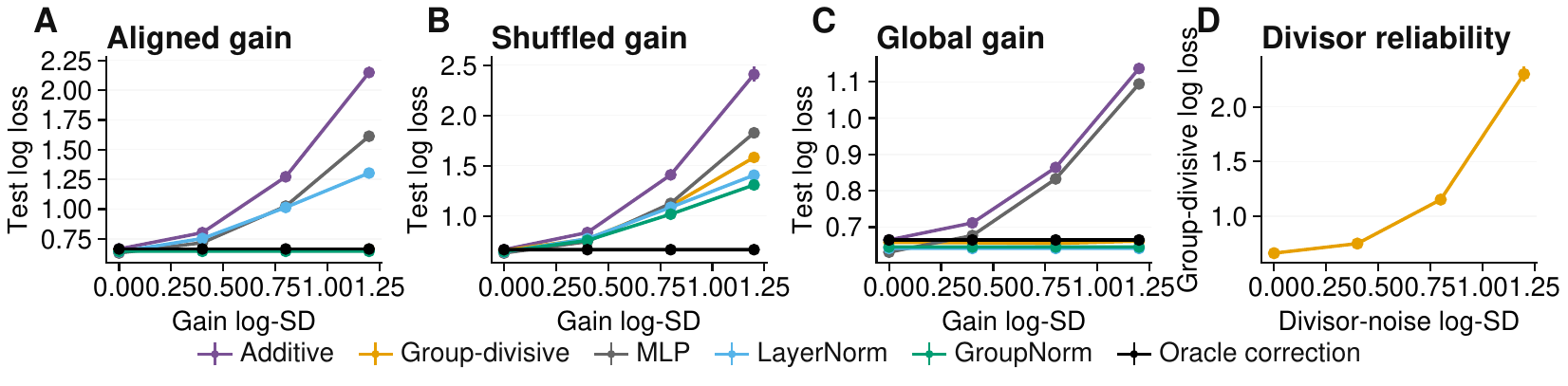}
\caption{\textbf{The gain--load interaction transfers to conventional feature
normalization.}
A frozen ImageNet-pretrained ResNet-18 supplies CIFAR-10 features; all trained
arms use the full 45,000/5,000 training/validation split and the same
clean-validation search budget. \textbf{(A)}~Aligned group gain separates the
mean-only divisive readout from an unnormalized additive readout, while
GroupNorm is comparably robust. \textbf{(B)}~Support shuffling weakens the
divisive benefit. \textbf{(C)}~Global gain is removed by divisive, LayerNorm,
and GroupNorm readouts. \textbf{(D)}~Divisor noise reverses the benefit. Mean
$\pm$ standard error of the mean (SEM) over five paired computational seeds;
this is an ML normalization
test, not biological evidence.}
\label{fig:cifar-normalization-main}
\end{figure}

\subsection{Shared variability in mouse primary visual cortex}
\label{subsec:v1-readout}

Shared trial-to-trial fluctuations can limit neural population codes when they
align with stimulus-informative dimensions
~\citep{GorisMovshonSimoncelli2014,CohenKohn2011,MorenoBote2014}. As locomotion
broadly modulates visually evoked V1 activity, it may therefore contribute
state-related shared fluctuations that affect stimulus decoding
~\citep{NiellStryker2010}. Here, we test such a downstream-decoder implication
and do not make a claim about biological mechanisms. We analyzed
nonnegative activity from three mouse V1 sessions
~\citep{Stringer2019,StringerOrientedStimuliFigshare2019}. Training-only
response pools feed parameter-matched direct-E/I readouts with excitatory
readout width $P_E=8,16,64$ for eight-way drifting-grating direction decoding.
After the same
validation budget selected one setup per mechanism, each condition was refit
with four new seeds; shifted-tanh branches feed a linear decoder
(Fig.~\ref{fig:v1-main}A).

\begin{figure}[t]
\centering
\includegraphics[width=\figwidth]{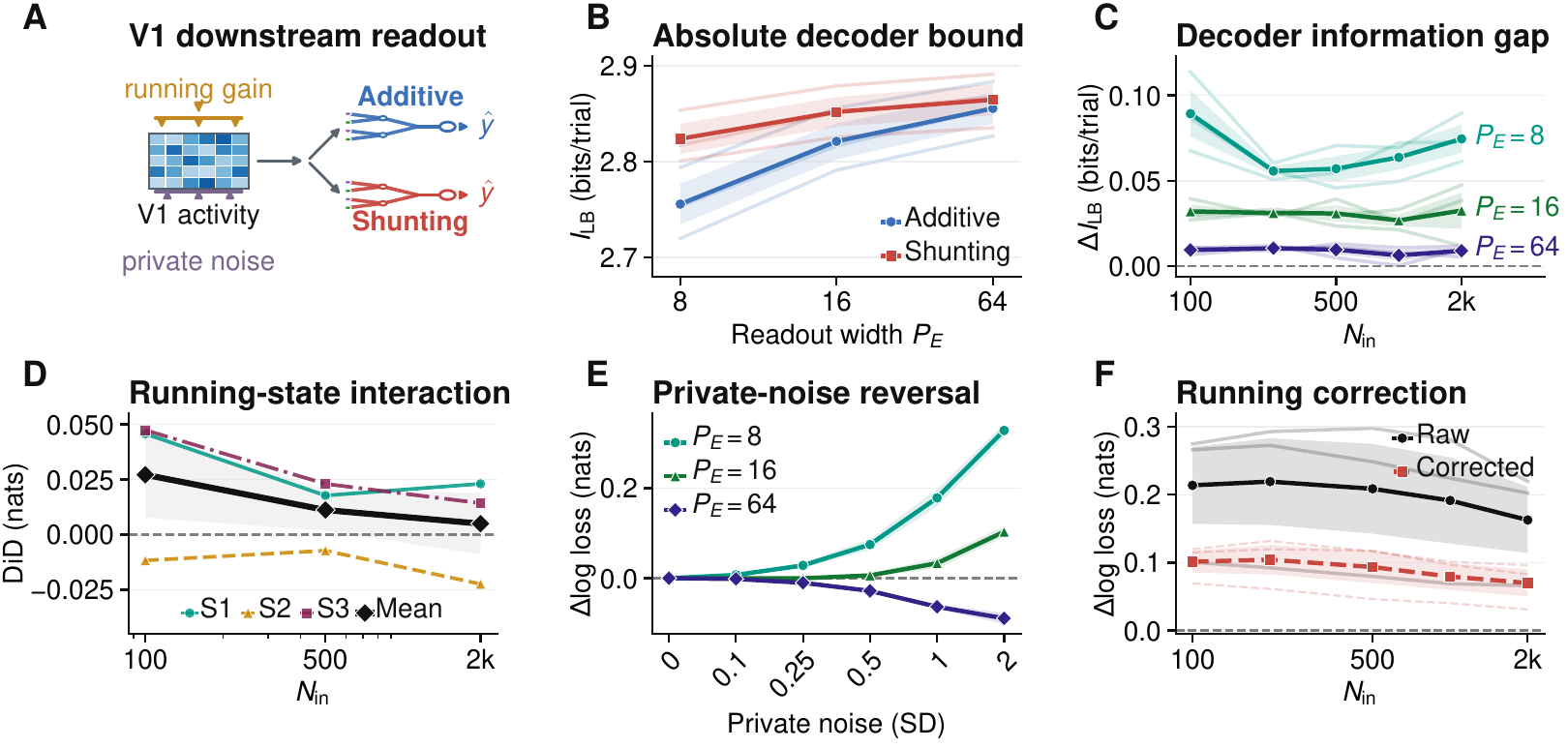}
\caption{\textbf{V1 decoder effects depend on readout width.}
\textbf{(A)}~Training-selected V1 pools feed matched additive and shunting
dendritic readouts (schematic). \textbf{(B--C)}~Absolute bounds stay high, but
their gap is largest at $P_E=8$. \textbf{(D)}~The running-state difference-in-differences (DiD)
varies by session. \textbf{(E)}~Private noise favors shunting at $P_E=8,16$ but
reverses at 64. \textbf{(F)}~A training-estimated multiplicative correction
reduces the fixed-feature shunting advantage. Positive favors shunting in C--F.
B--E use shifted-tanh branches; F is passive. Session/mouse
($n=3$; mean $\pm$ SEM).}
\label{fig:v1-main}
\end{figure}

We summarize held-out decoding with $I_{\rm LB}$, a variational
decoder-information lower bound rather than total population information
~\citep{PooleEtAl2019}. The shunting-minus-additive gap is largest for narrow
readouts and decreases as access widens despite high absolute bounds
(Fig.~\ref{fig:v1-main}B--C). Private noise favors shunting for narrow readouts
but reverses at the widest (Fig.~\ref{fig:v1-main}E); depth lowers the secondary
clean-accuracy endpoint for both mechanisms
(Supplementary Fig.~\ref{fig:v1-exact-resource-depth-noise-supp}A).

The prespecified running-state interaction is heterogeneous across sessions and
decreases with upstream access (Fig.~\ref{fig:v1-main}D). A training-estimated
multiplicative correction reduces the shunting log-loss advantage in every
session (Fig.~\ref{fig:v1-main}F), although the supplementary $d'$ endpoint
moves oppositely. These are conditional decoder effects, not increased
population information, a direct $\Lambda$ estimate, or circuit identification
(Supplementary Figs.~\ref{fig:v1-fixed-template-supp}
and~\ref{fig:v1-running-recovery-supp}).

\section{Discussion}
\label{sec:discussion}
Nonlinearity alone does not make division beneficial. Branch-local shunting
helps only when a reliable divisor meets the nuisance-affected signal and
suppresses more shared variability than its signal attenuation and denominator
noise. A terminal shunt adds no first-order decision direction; within a tree,
morphology and activation can reweight paths. Misplaced or noisy division can
therefore erase or reverse the benefit.

This refines classic accounts of divisive normalization and conductance-based
gain control~\citep{Heeger1992,ReynoldsHeeger2009,CarandiniHeeger2012,
ChanceAbbottReyes2002,MitchellSilver2003}, and complements work on dendritic
nonlinearities and morphology
~\citep{Poirazi2003,PolskyEtAl2004,LondonHausser2005,Beniaguev2021}. The matched-resource
comparison adds a boundary: local division is not automatically more expressive
than additive E/I, and depends on divisor placement and reliability.

Designed simulations recover this regime and its reversals, but training
recovers only part: I-stream background reverses the rule ordering, while V1
effects peak at narrow access and cross over with width and noise. These are
conditional decoder effects, not universal benefits, increased population
information, or circuit identification.

\ifpapersynthesis
The passive hierarchy further illustrates an inductive-bias rather than an
asymptotic-capacity advantage: flexible nonlinear predictors eventually
outperform the fixed shunting readout, while its small-sample accuracy benefit
is strongest with the correct support prior. Internal post-voltage activation
changes path gains and is therefore evaluated in the trained-model analyses
rather than inferred from this passive construction.
\fi

The gain--load--alignment principle separates nuisance tracking, fidelity,
denominator variability, and pathway placement.

For machine learning, the principle bounds local normalization under
multiplicative shift: division helps when reliable group-specific side
information precedes aggregation; global, noisy, or misaligned normalization
may be preferable. DendriNet tests this boundary, not generic benchmark
superiority.

\paragraph{Scope and limitations.}
The model omits temporal and regenerative mechanisms
~\citep{YangMurrayWang2016,Larkum2013,StuartSpruston2015,Gidon2020}. Terminal
monotone activation preserves the local boundary; internal activation can
reweight paths. The E/I streams are computational pathways, not identified cell
types, and model seeds are computational rather than biological replicates. The
hierarchy is designed; three-mouse V1 decoding cannot establish independent
replication or circuit identity.
\ifpapersynthesis
A five-mouse Allen scope control gave heterogeneous state interactions and did
not pass its prespecified replication rule
(Methods~\ref{methods:allen-v1-scope}).
\fi

\ifpapersynthesis
\section{Scope, predictions, and limitations}
\label{sec:camera-synthesis}

Taken together, the results show that depth changes a routing graph rather than
a scalar capacity knob. Mechanism-by-depth interactions and absolute depth
effects answer different questions; Table~\ref{tab:depth-synthesis}
distinguishes them across nuisance regimes.

\begin{table}[htbp]
\centering
\small
\setlength{\tabcolsep}{6pt}
\renewcommand{\arraystretch}{1.08}
\caption{Predictions by nuisance regime. Evidence is shown in parentheses.}
\label{tab:depth-synthesis}
\begin{tabularx}{\textwidth}{@{}>{\raggedright\arraybackslash}p{0.24\textwidth}>{\raggedright\arraybackslash}X@{}}
\toprule
Regime & Outcome (evidence) \\
\midrule
Terminal local & \textbf{Neutral.} No new decision direction
  \emph{(threshold and cone results)}. \\ \addlinespace[2pt]
Aligned reliable divisor & \textbf{Helps conditionally.} Gain suppression must
  exceed fidelity and load costs \emph{(gain sweep; exact hierarchy)}. \\ \addlinespace[2pt]
Recoverable nuisance & \textbf{Signal retained.} Refitting recovers orthogonal
  signal \emph{(rank-one boundaries)}. \\ \addlinespace[2pt]
Noisy divisor & \textbf{Can hurt.} Load can dominate
  \emph{(load and sensor-noise sweeps)}. \\ \addlinespace[2pt]
Mismatched support & \textbf{Fails.} The divisor misses the nuisance
  \emph{(support-shuffle reversal)}. \\ \addlinespace[2pt]
Aligned hierarchy & \textbf{Can help.} Gain reduction can beat composition
  costs \emph{(aligned-routing construction)}. \\ \addlinespace[2pt]
Resource-matched depth & \textbf{No generic depth benefit}
  \emph{(synthetic and V1 controls)}. \\ \addlinespace[2pt]
Decoder bottleneck & \textbf{Access only.} Access can improve; the information
  ceiling cannot \emph{(scaling, contact-cap, and V1 controls)}. \\
\bottomrule
\end{tabularx}
\end{table}

\paragraph{Predictions.}
With contacts held fixed, moving a reliable nuisance sensor onto the branch
carrying its affected excitation should preferentially improve shunting, while
sensor corruption or support shuffling should erase that increment. The
shunting increment should be largest behind a downstream bottleneck and need not grow with the
number of recorded inputs. Branch-resolved inhibitory-placement perturbations
are needed to test the biological mechanism; somatic V1 fits alone cannot do so.

\fi


\ifpapersynthesis
  \par\medskip
\fi

\ifneuripsanonymous\else
\begin{ack}
This work was supported in part by a gift from the Chan Zuckerberg Initiative
Foundation to establish the Kempner Institute for the Study of Natural and
Artificial Intelligence at Harvard University.
\end{ack}
\fi

\bibliographystyle{unsrtnat}
\bibliography{dendrit_theory}

\clearpage
\appendix

\renewcommand{\topfraction}{0.92}
\renewcommand{\bottomfraction}{0.75}
\renewcommand{\textfraction}{0.06}
\renewcommand{\floatpagefraction}{0.60}
\setcounter{topnumber}{3}
\setcounter{bottomnumber}{2}
\setcounter{totalnumber}{4}

\section*{Appendix Overview}
This appendix provides the theoretical derivations, experimental methods, and
additional analyses underlying the main results.

\begin{itemize}[leftmargin=*]
\item \textbf{Section~\ref{sec:app-methods}: Methods.} Biophysical derivations and
proofs; nuisance models and tree linearization; dendritic-network simulations;
representation geometry; V1 preprocessing, model fitting, and evaluation;
scope and assumptions; and reproducibility details covering statistical units,
code entry points, random seeds, compute, and asset licenses.
\item \textbf{Section~\ref{sec:app-supp-results}: Supplementary results.}
Additional analyses of local and comparator boundaries, population access,
morphology and divisor locality, learned-network optimization, and transfer to
V1 recordings.
\end{itemize}

The formal material states the assumptions behind the scalar boundary,
local-containment theorem, additive-tree flattening result, and depth-wise
decorrelation lemma. The implementation sections connect the abstract E/I
variables to the simulations and data analyses, specifying which quantities
are fitted and which remain fixed. Table~\ref{tab:comparator-ladder} lists the
comparator ladder used throughout the figures and the question each comparator
answers. Supplementary figures follow the order of
the main argument, from local theory through morphology and trained networks to
the V1 analyses.

\begin{table}[htbp]
\centering
\small
\setlength{\tabcolsep}{3.5pt}
\renewcommand{\arraystretch}{1.08}
\caption{\textbf{Comparator ladder.} Each row answers a different question; no
single baseline is used for every claim.}
\label{tab:comparator-ladder}
\begin{tabularx}{\textwidth}{@{}p{0.22\textwidth}p{0.25\textwidth}X@{}}
\toprule
Comparator & Fixed element & Purpose \\
\midrule
Fixed raw $E-I$ & Weights and recursion & Implemented-rule diagnostic \\
Local tangent & Shunting anchor and Jacobian & Finite-amplitude gating \\
Optimized positive E/I & Inputs and positive cone & Constrained additive capacity \\
Fitted $E-\rho I-b$ & Pooled E/I features & Ranking and calibration \\
Fitted linear & Raw observations & Unconstrained linear capacity \\
Validation-selected nonlinear & Raw observations and search budget & Passive nonlinear capacity and sample efficiency \\
\bottomrule
\end{tabularx}
\end{table}

\begin{table}[htbp]
\centering
\small
\setlength{\tabcolsep}{3.5pt}
\renewcommand{\arraystretch}{1.06}
\caption{\textbf{Operational meanings of matched comparisons.} Wherever
``matched'' names an experimental control, the main text states the fixed
quantity in the sentence itself. The remaining main-text uses are mathematical,
not control terms: ``tangent-matched'' labels the contrast defined by its own
equation, ``matched fractional changes'' means numerically equal
log-derivatives, and Fig.~\ref{fig:locality}'s ``matched'' denotes divisor
support that matches the nuisance.}
\label{tab:matching-controls}
\begin{tabularx}{\textwidth}{@{}p{0.26\textwidth}X@{}}
\toprule
Label & Quantity held fixed \\
\midrule
Conductance-mass matched & Total nonnegative E/I weight mass \\
Contact matched & Realized E/I mask counts; the E/I share may vary \\
Topology/resource matched & Branch graph, population width, active contacts, and parameter count \\
Same-policy mechanism control & Initialization, activation policy, optimizer, data, and nuisance draws \\
Equal-search complete models & Validation information and search budget; selected settings may differ \\
\bottomrule
\end{tabularx}
\end{table}

\begin{table}[htbp]
\centering
\caption{\textbf{Notation quick-reference (partial).}}
\label{tab:notation}
\footnotesize
\begin{tabularx}{\textwidth}{@{}lX@{}}
\toprule
Symbol & Meaning \\
\midrule
$X$ & upstream population activity vector; strict conductance analyses use $X\ge 0$ \\
$\xe,\xI$ & nonnegative excitatory/inhibitory pathway activities routed from $X$ \\
$\we,\wi$ & nonnegative excitatory/inhibitory synaptic weights \\
$E,I$ & nonnegative pooled excitatory/inhibitory drives ($E=\we^\top\xe$, $I=\wi^\top\xI$) \\
$N,T$ & implemented branch numerator and total conductance, Eq.~\eqref{eq:node-NT} \\
$V^{\add}, V^{\sh}$ & additive and shunting readouts \\
$\dprime^2$ & discriminability index (Eq.~\ref{eq:dprime}) \\
$M$ & total conductance/weight mass, $M=\lVert w\rVert_1$ \\
$P_E,P_I$ & learned excitatory/inhibitory somatic population widths; the trained configurations analyzed here use $P_I=0$ \\
$D_E,D_I$ & candidate feature dimensions available to E/I Top-$K$ projections \\
$N_E,N_I$ & configured upper caps on active E/I contacts per directly innervated branch (implementation fields \texttt{k\_e}, \texttt{k\_i}) \\
$N_{E,b}^{\rm real},N_{I,b}^{\rm real}$ & exact row-mask sums; respectively at most $\min(N_E,D_{E,b})$ and $\min(N_I,D_{I,b})$, with equality in the unmasked reported sweeps \\
$N_+,N_-$ & training-assigned positive-/negative-evidence counts in the fixed-template V1 diagnostic; not biological cell types \\
$B_{\rm dend}$ & non-somatic branches per somatic unit, Eq.~\eqref{eq:branch-resource-count} \\
$r$ & input-statistics regime index (Gaussian, nonlinear/saturating, Cox/Poisson, heavy-tailed, or V1-derived)  (never a response; cf.\ Eq.~\eqref{eq:latent-gain-model})\\
$\alpha(r)$ & signal--gain alignment in Fisher geometry (0 = orthogonal nuisance; 1 = fully aligned) \\
$\Gamma(r)$ & shared-gain magnitude before alignment projection (dimensionless; Methods~\ref{methods:regime-map}) \\
$G_{\mathrm{eff}}(r)$ & harmful aligned gain strength, $G_{\mathrm{eff}}(r)=\alpha(r)\Gamma(r)$ \\
$M_g(\mathcal T,r)$ & dimensionless surviving gain factor defined by
$\Var_g(w_{\sh}^{\top}h_{\sh})/\Var_{\rm ind}(w_{\sh}^{\top}h_{\sh})
=\alpha M_g\Gamma$ under the direction-preserving surrogate \\
$\mathcal T$ & branch tree (topology); $T$ always denotes total conductance \\
$\Lambda,\Lambda_{\mathcal T}$ & single-stage and tree-accumulated residual denominator-variance penalty \\
$B(\mathcal T,r)$ & optimized-comparator baseline fidelity $d_{0,\sh}^{\prime\,2}(\mathcal T,r)/d_{0,\add}^{\prime\,2}(r)$ of the divisive operating point \\
$B_{\mathrm{cas}}(\mathcal T,r)$ & fixed-cascade baseline fidelity, where both additive and shunting cascades can depend on morphology \\
$J_{q,\phi}$ & local Jacobian of the activated additive or shunting tree, $q\in\{\add,\sh\}$ \\
$\gamma_n$ & local branch-activation slope $\phi'_n(\bar V_n)$ \\
$\gamma_g$ & geometric mean per-stage attenuation of shared gain in the chain special case \\
$\rho,\rho_n$ & local shunting I/E sensitivity ratios ($\rho=E_0/(1+I_0)$ in one compartment; $\rho_n=\beta_n/\alpha_n$ in trees) \\
$\rho_{EE},\rho_{II},\rho_{EI}$; $\rho$ in Sec.~\ref{methods:exchangeable} & exchangeable pairwise correlations in saturation models \\
$\eta_n,\tau_{n\leftarrow b}$ & nonnegative self-gain and child-transfer coefficients of the passive additive tree recursion \\
\bottomrule
\end{tabularx}
\end{table}

\clearpage
\section{Methods}
\label{sec:app-methods}
\renewcommand{\figurename}{Figure}
\renewcommand{\thefigure}{\arabic{figure}}
\setcounter{secnumdepth}{3}

This section gives the derivations, assumptions, simulation details, dataset
processing, and reproducibility information supporting the main text.
Table~\ref{tab:notation} summarizes the notation used throughout.

\ifneuripsanonymous
\else
\citet{SafaaiRichardsSabatini2026} use the same architecture to study local
credit assignment; the present study concerns forward population readout rather
than feedback-field compression.
\fi

\subsection{Biophysical and local comparator theory}
\label{methods:grp-theory}

\subsubsection{Biophysical derivation of the node equation}
\label{methods:node-derivation}

We derive the nondimensional node rule~\eqref{eq:node-voltage} from
Kirchhoff's current law at steady state.
Let $V_n$ denote the membrane potential at node $n$,
$g_L$ the leak conductance (reversal $E_L$),
$g_{E,n}$ the total excitatory conductance (reversal $E_E$),
$g_{I,n}$ the total inhibitory conductance (reversal $E_I$),
and $g_{n\leftarrow b}$ the axial conductance coupling child $b$ to node $n$.
Current balance requires
\[
  g_L(V_n - E_L)
  + g_{E,n}(V_n - E_E)
  + g_{I,n}(V_n - E_I)
  + \sum_b g_{n\leftarrow b}(V_n - V_b) = 0.
\]
Solving for $V_n$:
\[
  V_n = \frac{g_L E_L + g_{E,n} E_E + g_{I,n} E_I + \sum_b g_{n\leftarrow b} V_b}
             {g_L + g_{E,n} + g_{I,n} + \sum_b g_{n\leftarrow b}}.
\]
Define nondimensional variables:
$E_n := g_{E,n}(E_E - E_L)/(g_L(E_E - E_L)) = g_{E,n}/g_L$,
$I_n := g_{I,n}/g_L$,
$g_{n\leftarrow b}:=g_{n\leftarrow b}/g_L$,
and the inhibitory-reversal factor
$\chi_{\mathrm{rev}}:=(E_I-E_L)/(E_E-E_L)$.
Rescaling potentials by $(E_E - E_L)$ so that leak maps to 0 and excitatory reversal
to 1, and normalizing conductances by $g_L$, gives
\[
V_n
=
\frac{E_n+\chi_{\mathrm{rev}}I_n+\sum_b g_{n\leftarrow b}V_b}
     {1+E_n+I_n+\sum_b g_{n\leftarrow b}}.
\]
For pure shunting ($E_I=E_L$), $\chi_{\mathrm{rev}}=0$ and this reduces to
Eq.~\eqref{eq:node-voltage}.

An external denominator conductance has the same normalization: if
$g_{\mathrm{bg}}(t)$ has reversal $E_{\mathrm{bg}}=E_L$, then
$L(t)=g_{\mathrm{bg}}(t)/g_L$ adds to the denominator and not the numerator.
Its constant component changes input resistance and the operating point; its
fluctuating component changes the trial distribution. Stochastic E/I
conductances and balanced background barrages are established ways to produce
such high-conductance states and gain changes
~\citep{DestexhePare1999,ChanceAbbottReyes2002,FellousEtAl2003}. This mapping is
exact only near the reference reversal. For $E_{\mathrm{bg}}\ne E_L$, the same
physical input contributes both numerator current and denominator conductance.
Moreover, a phenomenological normalization pool need not be implemented by a
local inhibitory shunt; recurrent withdrawal of excitation can generate
divisive suppression as well~\citep{OzekiEtAl2009}. The denominator-only
DendriNet intervention should therefore be read as mechanism isolation, not as
a same-physical-input comparison of two biological neurons.

When the inhibitory reversal is slightly below leak
($E_I = E_L + \chi_{\mathrm{rev}}(E_E - E_L)$ with $\chi_{\mathrm{rev}}\approx -0.07$ to $-0.14$ for cortical
parameters), a small subtractive term $\chi_{\mathrm{rev}} I_n$ appears in the numerator,
giving the conductance-based subtractive comparator
$V_n^{\mathrm{cond}} = (E_n + \chi_{\mathrm{rev}} I_n + \sum_b g_{n\leftarrow b} V_b)/
(1 + E_n + I_n + \sum_b g_{n\leftarrow b})$.
We use this as a matched-mechanism control to separate denominator effects from
current-mode baseline choices.

\subsubsection{From the cable equation to the node-voltage model}
\label{methods:cable-derivation}

The single-equation node rule used throughout the paper is a conductance-based
feed-forward compartmental abstraction motivated by steady-state cable balance
in the high-conductance regime. Here we make that connection explicit and
document what the simplification keeps versus drops.

\paragraph{Steady-state cable equation.}
The continuous passive cable equation for a dendritic segment is
\[
  c_m\,\partial_t V(x,t) = \tfrac{1}{r_a}\partial_x^2 V(x,t)
                            - g_L (V - E_L)
                            - g_E(x,t)(V - E_E)
                            - g_I(x,t)(V - E_I),
\]
with $c_m$ the membrane capacitance, $r_a$ the axial resistance per unit length,
$g_L$ the leak, and $g_{E},g_I$ the excitatory and inhibitory synaptic conductance
densities. Spatially discretizing into compartments and going to steady state
($\partial_t = 0$) gives Kirchhoff's law at each node
\citep{Rall1962, Koch1999}. A full passive cable tree contains bidirectional
parent--child axial coupling. Eq.~\eqref{eq:node-voltage} is the reduced
directional pooling model used here, in which child subunit voltages are passed
proximally with fixed effective axial gains.

\paragraph{Validity of the steady-state approximation.}
The steady-state approximation is accurate when input fluctuations are slow
relative to the effective membrane time constant
$\tau_{\mathrm{eff}} = c_m / (g_L + g_E + g_I)$. In the high-conductance state
characteristic of cortical neurons in vivo \citep{DestexhePare1999,
DestexheRudolphPare2003}, $\tau_{\mathrm{eff}}$ falls to $\sim$5--15~ms, well
below typical perceptual integration windows. We additionally restrict the
model to moderate subthreshold inputs and omit active dendritic conductances;
under this passive modeling approximation, conductance currents sum linearly in
the numerator of Eq.~\eqref{eq:node-voltage}. The observed location dependence
of somatic EPSP efficacy \citep{MageeCook2000,WilliamsStuart2002} motivates the
morphology-specific effective axial gains; it does not establish exact
linearity of active dendrites.

\paragraph{What the simplification drops.}
The model omits four classes of mechanism:
\textbf{(i)~NMDA-receptor voltage-dependence}, which makes $g_E$ depend on
$V$, introducing a multiplicative gain on excitation;
\textbf{(ii)~regenerative dendritic events} (Ca$^{2+}$ plateaus, NMDA spikes,
Na$^+$ spikelets) that produce all-or-none nonlinear amplification;
\textbf{(iii)~active backpropagation} of axonal action potentials into the
tree, which can gate plasticity and resting state;
and \textbf{(iv)~time-dependent synaptic dynamics} (short-term depression,
facilitation). Each of these has been studied extensively
\citep{Larkum2013, StuartSpruston2015, Gidon2020} and is known to add to,
rather than replace, the passive integration the model captures.

\paragraph{Scope of active mechanisms.}
The gain--load--alignment principle is geometric rather than a detailed
biophysical claim. It may extend to active mechanisms when their net
input-output map remains approximately ratio-like over the relevant operating
range. Mechanisms that switch compartments into regenerative or all-or-none
regimes, such as Ca$^{2+}$ plateaus or strong NMDA spikes, fall outside the
present analysis
\citep{TranVanMinhEtAl2015, CazeHumphriesGutkin2013, PayeurBeiqueNaud2019}.

\paragraph{Multi-compartment abstraction.}
The model retains the local conductance divisor suggested by passive current
balance, but replaces the full coupled cable solve with a proximal recursion over
subunit outputs. Thus the axial terms $g_{n\leftarrow b}$ in
Eq.~\eqref{eq:node-voltage} should be read as effective feed-forward coupling
gains, not as a claim that the complete passive cable equation has been solved
exactly.

\subsubsection{Pooled moments from high-dimensional inputs}
\label{methods:pooled-moments}

All strict conductance-readout analyses start from nonnegative presynaptic
activities. Given nonnegative synaptic weights $\we,\wi\ge 0$ and
class-conditional input moments, the pooled drives
$E=\we^\top\xe$, $I=\wi^\top\xI$ have
\begin{align}
  \mu_{E,c} &= \we^\top\mu_{e,c}, \qquad
  \mu_{I,c} = \wi^\top\mu_{i,c},
  \label{eq:pooled-means}
  \\
  \Var(E\mid c) &= \we^\top\Sigma_{ee,c}\,\we, \qquad
  \Var(I\mid c) = \wi^\top\Sigma_{ii,c}\,\wi, \qquad
  \Cov(E,I\mid c) = \we^\top\Sigma_{ei,c}\,\wi.
  \label{eq:pooled-covs}
\end{align}
The class-averaged operating point is
$(E_0,I_0)=(\we^\top\bar\mu_e,\,\wi^\top\bar\mu_i)$ with
$\bar\mu_e=\tfrac12(\mu_{e,0}+\mu_{e,1})$ and
$\bar\mu_i=\tfrac12(\mu_{i,0}+\mu_{i,1})$.
The class means are nonnegative because $\xe,\xI\ge 0$; covariance and
cross-covariance entries can have either sign because they describe fluctuations
around nonnegative rates, not signed conductances.
For real V1 fixed-budget readouts, $\xe$ and $\xI$ are raw calcium-response
rates and are not z-scored. For bounded synthetic population codes, they are
generated directly in $[0,1]$. If a benchmark begins with signed features, it
must either use an additive comparator, an ON/OFF or dual-rail encoding, or a
nonnegative transfer activation before being interpreted as conductance input.
Accordingly, the bounded contextual-stream, positive Stringer, and raw V1
fixed-budget analyses are treated as strict conductance results, whereas
normalized image or normalized-neural-feature experiments are treated as ML
extensions rather than direct conductance-readout evidence.

\paragraph{Implementation of E/I pathways.}
In code, the strict E/I pathway construction has three steps. First, the upstream
representation is converted into nonnegative pathway activities
$(\xe,\xI)$, either because the data are already nonnegative rates or because a
nonnegative transfer function is applied before the conductance interpretation.
Second, excitatory and inhibitory synaptic parameters are mapped through a
positive transform, usually softplus, before forming pooled drives. Third, the
chosen compartment rule combines those pooled drives. The additive comparator
$E-I$ is allowed to be signed because it represents a hidden voltage or current
proxy; it is not interpreted as a negative firing rate. Any quantity that is
passed forward as biological activity is converted back to a nonnegative rate.

This distinction is important for covariance terms. Although $E$ and $I$ are
nonnegative on every trial, fluctuations around their class means can be
positively or negatively correlated. A negative entry in a covariance or
cross-covariance matrix therefore does not imply a negative conductance. It only
means that two nonnegative rates tend to fluctuate in opposite directions around
their means.

\subsubsection{Cone-constrained linear discriminant analysis}
\label{methods:cone-lda}

Here we give the full linear-discriminant-analysis (LDA) derivation and its
Karush--Kuhn--Tucker (KKT) optimality conditions.

\begin{algorithm}[t]
\caption{Exact two-orientation nonnegative least-squares (NNLS) solver for
cone-constrained LDA. Because the Rayleigh quotient is squared, solve both
$s\in\{-1,+1\}$ orientations and retain the larger primal--dual feasible
value.}
\label{alg:cone-lda}
\begin{algorithmic}[1]
\REQUIRE Comparator-coded class-mean gap $\Delta\mu_Z$, pooled covariance $\bar\Sigma_Z$ (coordinates $Z=(\xe,-\xI)$)
\ENSURE Cone-optimal direction $u\ge0$, selected orientation $s^*$, and optimum $\dprime^{2,*}$
\STATE Regularize $\bar\Sigma_Z$ only if needed for positive definiteness and factor $\bar\Sigma_Z=A^\top A$.
\FOR{$s\in\{-1,+1\}$}
  \STATE $\Delta_s\gets s\Delta\mu_Z$,\quad $b_s\gets A^{-\top}\Delta_s$
  \STATE $q_s\gets\arg\min_{q\ge0}\|Aq-b_s\|_2^2$ \COMMENT{nonnegative least squares}
  \STATE $g_s\gets\bar\Sigma_Zq_s-\Delta_s$,\quad $d_s\gets\Delta_s^\top q_s$
  \IF{$d_s\le0$ or $q_s\not\ge0$ or $g_s\not\ge0$ or $q_s\odot g_s\not\approx0$}
    \STATE mark orientation $s$ infeasible and \textbf{continue}
  \ENDIF
  \STATE $u^{(s)}\gets q_s/d_s$,\quad $A_s\gets\{j:q_{s,j}>0\}$
\ENDFOR
\STATE $s^*\gets\arg\max_{s\ \mathrm{feasible}}d_s$,\quad $\dprime^{2,*}\gets d_{s^*}$
\STATE \RETURN $u^{(s^*)}=(\we,\wi)$, $s^*$, and $\dprime^{2,*}$
\end{algorithmic}
\end{algorithm}

The cone program is written in comparator coordinates
$Z=(\xe,-\xI)$, but the feasible variables are always biological synaptic
weights $u=(\we,\wi)\ge 0$ applied to nonnegative inputs.
Unconstrained signed LDA is therefore a diagnostic baseline, not the feasible
conductance readout used in the theorem.
We solve $\max_{u\ge 0}(u^\top\Delta\mu_Z)^2/(u^\top\bar\Sigma_Z u)$ by
considering both orientations $\Delta_s=s\Delta\mu_Z$, $s\in\{-1,+1\}$.
For each feasible orientation, fix the scale $u^\top\Delta_s=1$ and minimize
$u^\top\bar\Sigma_Zu$ subject to $u\ge0$; the larger of the two feasible
Rayleigh values is the squared-objective optimum.
The Lagrangian is
$\mathcal{L}_s=u^\top\bar\Sigma_Z u - \lambda^\top u + \nu(u^\top\Delta_s - 1)$
with $\lambda\ge 0$.
The KKT conditions are:
\begin{align}
  2\bar\Sigma_Z u - \lambda + \nu\Delta_s &= 0,
  \label{eq:kkt-stat}
  \\
  u\ge 0,\quad \lambda\ge 0,\quad \lambda\odot u &= 0,
  \label{eq:kkt-comp}
  \\
  u^\top\Delta_s &= 1.
  \label{eq:kkt-scale}
\end{align}
Let $A:=\{j: u_j>0\}$ denote the active set.
Complementarity~\eqref{eq:kkt-comp} gives $\lambda_A=0$.
Restricting stationarity~\eqref{eq:kkt-stat} to $A$:
$2(\bar\Sigma_Z)_{AA}u_A + \nu(\Delta_s)_A = 0$, hence
\begin{equation}
  u_A \propto (\bar\Sigma_Z)_{AA}^{-1}(\Delta_s)_A,
  \qquad u_{A^c}=0.
  \label{eq:active-set-solution}
\end{equation}
The active set is valid if $u_A\ge 0$ (primal feasibility) and the implied
multipliers on $A^c$ satisfy $\lambda_{A^c}\ge 0$ (dual feasibility).
This provides the exact closed-form solution given the active set; in practice
we factor $\bar\Sigma_Z=A^\top A$ and solve the equivalent nonnegative
least-squares problem for each orientation. This enforces both primal and
inactive-coordinate dual feasibility. All reported numerical solutions use
this two-orientation procedure.

For orientation $s$, the optimal value is
$d_s=(\Delta_s)_{A_s}^\top(\bar\Sigma_Z)_{A_sA_s}^{-1}(\Delta_s)_{A_s}$;
the squared-objective optimum is $\dprime^{2,*}=\max_s d_s$ over feasible
orientations.
When $s\bar\Sigma_Z^{-1}\Delta\mu_Z\ge0$ for either orientation $s$, the
corresponding active set is the full index set and the
constrained optimum equals the unconstrained Fisher value
$\Delta\mu_Z^\top\bar\Sigma_Z^{-1}\Delta\mu_Z$.
As usual for a Rayleigh quotient, only the optimal cone direction and active set
are intrinsic; absolute weight scale is fixed by the chosen normalization
(here $u^\top\Delta_s=1$ in each oriented subproblem).

\subsubsection{Self-consistent local optimization for shunting}
\label{methods:self-consistent}

\begin{algorithm}[t]
\caption{Closed-form test and construction for a self-consistent local
shunting realization. The cone solve evaluates both signal orientations. The
Rayleigh quotient fixes only a ray, so realizability is tested after allowing
its otherwise arbitrary positive scale to vary.}
\label{alg:self-consistent}
\begin{algorithmic}[1]
\REQUIRE Class-conditional pooled means $\bar\mu_e,\bar\mu_i$; comparator-coded
moments $(\Delta\mu_Z,\bar\Sigma_Z)$; chosen interior fraction $0<f<1$
\ENSURE Self-consistent shunting weights $(\we^\star,\wi^\star)$ or failure
\STATE $(q,s^*,d^*)\gets\operatorname{ConeLDA}(\Delta\mu_Z,\bar\Sigma_Z)$ \COMMENT{best ray over $s\in\{-1,+1\}$}
\STATE $e\gets q_e^\top\bar\mu_e$,\quad $i\gets q_i^\top\bar\mu_i$
\IF{$e\le i$ or $e\le0$}
  \STATE \RETURN failure \COMMENT{the optimal ray has no positive shunting realization}
\ENDIF
\STATE $I_0\gets i/(e-i)$,\quad $r_{\max}\gets[4(1+I_0)]^{-1}$,
       \quad $r\gets f r_{\max}$
\STATE $\alpha\gets r/(e-i)$,\quad $k\gets r(1+I_0)$
\STATE $y\gets 2k/[1-2k+\sqrt{1-4k}]$,\quad $E_0\gets(1+I_0)y$
       \COMMENT{stable low-conductance root}
\STATE $D_0\gets1+E_0+I_0$,\quad
       $a_E\gets(1+I_0)/D_0^2$,\quad $|a_I|\gets E_0/D_0^2$
\STATE $\we^\star\gets\alpha q_e/a_E$,\quad
       $\wi^\star\gets\alpha q_i/|a_I|$
\STATE verify $\we^{\star\top}\bar\mu_e=E_0$,
       $\wi^{\star\top}\bar\mu_i=I_0$, and
       $(a_E\we^\star,|a_I|\wi^\star)=\alpha q$
\STATE \RETURN $(\we^\star,\wi^\star)$
\end{algorithmic}
\end{algorithm}

\paragraph{Why the operating points differ and how realizability is tested.}
The additive cone-LDA solution gives optimal weights $u^*_{\add}=(\we^{\star,\add},\wi^{\star,\add})$
directly from the comparator-coded moment statistics.
For shunting, the linearization coefficients $(a_E,|a_I|)$ depend on the
operating point $(E_0,I_0)=(\we^\top\bar\mu_e,\wi^\top\bar\mu_i)$, which
depends on the weights.
The two-orientation cone-LDA solution on the \emph{effective} weights
$\tilde u=(a_E\we,|a_I|\wi)$ supplies the additive-optimal candidate ray, but
mapping it back to synaptic weights gives
$\we^{\star,\sh}=(\tilde u^*_e)/a_E$ and $\wi^{\star,\sh}=(\tilde u^*_i)/|a_I|$.
Because $a_E\neq 1$ and $|a_I|\neq 1$ generically,
$(\we^{\star,\sh},\wi^{\star,\sh})\neq(\we^{\star,\add},\wi^{\star,\add})$.
These different weights induce a different operating point
$(E_0^{\sh},I_0^{\sh})\neq(E_0^{\add},I_0^{\add})$,
which in turn changes the linearization coefficients.
The apparent circularity can be resolved algebraically because the Rayleigh
objective fixes a direction but not its positive scale.

\paragraph{Equivalent operating-point self-consistency relations.}
Let $m_E:=\tilde u_e^\top\bar\mu_e$ and $m_I:=\tilde u_i^\top\bar\mu_i$ at a fixed point,
with $\tilde u=(a_E\we,|a_I|\wi)$ and $D_0=1+E_0+I_0$.
Using $a_E=(1+I_0)/D_0^2$ and $|a_I|=E_0/D_0^2$ gives
\begin{equation}
  m_E=\frac{E_0(1+I_0)}{D_0^2},
  \qquad
  m_I=\frac{E_0I_0}{D_0^2},
  \qquad
  \frac{m_I}{m_E}=\frac{I_0}{1+I_0}.
  \label{eq:op-self-consistent}
\end{equation}
Equation~\eqref{eq:op-self-consistent} makes explicit that the local operating
point is induced by the synaptic weights, not externally tuned. Write an
additive-optimal effective ray as $q$ and define
$e=q_e^\top\bar\mu_e$ and $i=q_i^\top\bar\mu_i$. A positive shunting
realization of some scaled version $\tilde u=\alpha q$ exists if and only if
$e>i$ (apart from degenerate zero-signal cases). Then
$I_0=i/(e-i)$ and
\[
  \alpha(e-i)=\frac{E_0}{(1+I_0+E_0)^2}.
\]
The right-hand side has maximum $[4(1+I_0)]^{-1}$, so any sufficiently small
$\alpha>0$ is feasible. Algorithm~\ref{alg:self-consistent} chooses a fixed
interior fraction of that maximum and verifies the induced means and effective
ray to numerical precision. If $e\le i$, changing scale cannot repair the mean
ratio and local equality with that additive-optimal ray is unavailable.

\paragraph{Practical numerical use.}
The implementation follows the same separation between direction and scale. For
local cone-LDA diagnostics, the two-orientation nonnegative least-squares solver returns a direction in the
positive comparator cone and a discriminability value. A normalization convention
then fixes a representative scale only for plotting operating points and
linearization coefficients. A pure upper mass cap alone does not change the
frozen-$D$ Rayleigh ray set. An equality or lower-bound budget can instead pin
the induced operating scale, and simultaneous E/I-specific constraints or
ratio bounds can remove rays; those cases must be checked together with
self-consistency and are analyzed as resource-regime changes.

For nonlinear shunting simulations, we do not infer an operating point from the
linear theorem alone. The pooled means and variances are either generated
directly by the simulation or induced by the learned weights, and the exact
nonlinear voltage is evaluated at those induced values. This is why the local
theorem and the beyond-local regime map play different roles: the theorem
establishes the matched local baseline, whereas the simulations measure the
curvature, gain, load, and depth effects that remain after leaving that baseline.

\paragraph{When the conditional equality is attained.}
At a fixed positive operating point, the diagonal rescaling
$(a_E,|a_I|)$ maps every shunting Jacobian ray into the additive cone. If the
two-orientation additive-optimal ray $q$ satisfies
$q_e^\top\bar\mu_e>q_i^\top\bar\mu_i$, the closed-form construction above
provides a positive self-consistent operating point that attains the same local
Rayleigh value; the mechanisms then
differ in the synaptic weights and operating points that realize it. If no such
ray exists, Theorem~\ref{thm:local-equivalence} guarantees only local
containment and the additive optimum is an upper bound.

\paragraph{Which additional constraints preserve the feasible-ray relation.}
Because $\dprime^2$ is scale invariant, the relevant object is the set of
feasible \emph{rays}, not equality of the vector-valued feasible sets.
Nonnegativity is preserved by the positive diagonal map
$u\mapsto \tilde u=(a_E\we,|a_I|\wi)$, and support/cardinality is preserved
because positive rescaling does not change which coordinates are active.
Equality still requires self-consistent realization of an optimal ray. In the
support-constrained simulation, 64/120 $K\le2$ additive optima admit such a
realization, and every realizable case ties exactly; Appendix
Supplementary Fig.~\ref{fig:local-equivalence-supp} shows this boundary.

A pure upper mass cap, $\sum_j w_j\le M$, also leaves the frozen-$D$ Rayleigh
ray set unchanged: any nonzero positive direction can be scaled down to satisfy
it. Writing the mapped cap as
\[
  \sum_{j\in E}\frac{\tilde u_j}{a_E}+\sum_{j\in I}\frac{\tilde u_j}{|a_I|}\le B
\]
changes coordinate scale but not the available positive rays. Budgets become
scientifically consequential when they pin or lower-bound the scale and thus
alter the induced $D(E_0,I_0)$ or self-consistent realizability, or when
multiple E/I-specific constraints restrict ratios and remove directions.
Gain-ratio bounds and simultaneous E/I equalities are examples of the latter.
By contrast, explicit per-branch count sweeps in a learned tree are a beyond-local
architectural regime change, not a counterexample to the local theorem itself.
At the learned-network level, analogous resource changes are best read as
beyond-local operating-point and architecture changes rather than
counterexamples to the cone-preserving local theorem.

\paragraph{Beyond local: how different operating points create divergence.}
At the self-consistent weights, the actual (nonlinear) shunting $\dprime^2$
differs from its linearized value by curvature, saturation, and
noise-propagation corrections that depend on operating point.
Because $(E_0^{\sh},I_0^{\sh})\neq(E_0^{\add},I_0^{\add})$, these
corrections differ between mechanisms, producing the beyond-local
divergences quantified in the regime map (Fig.~\ref{fig:regime-map}).

\subsubsection{Full proof of Theorem~\ref{thm:local-equivalence}}
\label{methods:local-equivalence-proof}

\begin{proof}
\emph{Domain.}
Throughout the proof, $\xe,\xI$ are nonnegative presynaptic activities and
$u=(\we,\wi)$ lies in the nonnegative synaptic cone.
The signed object $Z=(\xe,-\xI)$ is only a notational device for the additive
comparator voltage.

\emph{Containment.}
For shunting at any positive operating point $(E_0,I_0)$, linearization gives
$(\we,\wi)\mapsto\tilde u=(a_E\we,|a_I|\wi)$ with
$a_E=(1+I_0)/D_0^2>0$, $|a_I|=E_0/D_0^2>0$, and
$D_0=1+E_0+I_0$. Hence every positive shunting weight vector produces an
effective vector $\tilde u\ge0$ in the additive comparator cone. Taking the
union only over induced self-consistent operating points can restrict this set
but cannot enlarge it, so
$\mathcal C_{\sh}^{\rm loc}\subseteq\mathcal C_{\add}$. Maximizing the same
squared Rayleigh objective over a subset cannot exceed the maximum over the
full additive cone, proving
$\dprime^{2,*}_{\sh,\mathrm{loc}}\le\dprime^{2,*}_{\add}$.

\emph{Conditional equality.}
Algorithm~\ref{alg:cone-lda} evaluates both $s\Delta\mu_Z$ orientations and
returns an additive-optimal ray $\tilde u^*$. If that ray admits positive
weights satisfying $\tilde u^*=D(E_0,I_0)u$ at their own induced operating
point, then $\tilde u^*\in\mathcal C_{\sh}^{\rm loc}$. The subset and superset
then contain a common maximizer, so their optimal local values are equal. If
the realization does not exist, containment alone gives no equality claim. For
an optimal ray $q$ with positive input means, the construction in
Algorithm~\ref{alg:self-consistent} shows that
$q_e^\top\bar\mu_e>q_i^\top\bar\mu_i$ is the exact ray-scale realizability
test; equality may still hold through a different ray when the additive optimum
is degenerate.

\emph{Weights and operating points.}
When equality is realized, the inverse map is
$\we=(\tilde u^*)_e/a_E$, $\wi=(\tilde u^*)_i/|a_I|$. Since generically
$a_E\neq|a_I|$, these weights differ from the additive coordinates and induce a
different operating point. The degenerate $E_0=0$ case has $|a_I|=0$ and is not
an invertible positive realization.
\end{proof}

\subsubsection{Proofs of Corollary~\ref{cor:activation-local-tie} and
Proposition~\ref{prop:constraint-equivariant}}
\label{methods:corollary-proposition-proofs}

\begin{proof}[Proof of Corollary~\ref{cor:activation-local-tie}]
Apply a differentiable scalar activation $\phi$ with $\phi'(V_0)>0$ at the induced
operating point and linearize $y=\phi(V)\approx\phi(V_0)+\phi'(V_0)\,(V-V_0)$.
The constant $\phi(V_0)$ shifts both class-conditional means equally and cancels in
the mean difference, so $\Delta\mu_y=\phi'(V_0)\,\Delta\mu_V$, while
$\Var(y\mid c)=\phi'(V_0)^2\,\Var(V\mid c)$ to first order. The positive scalar
$\phi'(V_0)$ therefore cancels in the discriminability ratio,
\[
  \dprime^2(y)
  =\frac{(\Delta\mu_y)^2}{\tfrac12\Var(y\mid0)+\tfrac12\Var(y\mid1)}
  =\frac{\phi'(V_0)^2(\Delta\mu_V)^2}
        {\phi'(V_0)^2\big(\tfrac12\Var(V\mid0)+\tfrac12\Var(V\mid1)\big)}
  =\dprime^2(V).
\]
Hence applying $\phi$ to either local scalar readout leaves that readout's
moment-based discriminability unchanged. The containment of
Theorem~\ref{thm:local-equivalence} is therefore preserved, and equality holds
only under its self-consistent-realizability condition. Curvature of $\phi$
enters at the same order as the beyond-local terms and is outside this local
relation. Internal activation is instead described by the activated-tree
Jacobian in Eqs.~\eqref{eq:activated-tree-dprime}--
\eqref{eq:activated-shunting-local} and tested empirically in
Fig.~\ref{fig:ei-circuit}E--F.
\end{proof}

There is also an exact finite-amplitude statement for rank and threshold
metrics. If $\phi$ is strictly increasing, then
$V_i<V_j\Leftrightarrow\phi(V_i)<\phi(V_j)$, and every threshold set
$\{V>t\}$ is identical to $\{\phi(V)>\phi(t)\}$. Thus receiver operating
characteristic (ROC) AUC and the minimum
classification error over scalar thresholds are unchanged by a terminal
activation. This exact invariance does not extend to moment-based $d'^2$, to
calibration-sensitive losses, or to internal activation.

\begin{proof}[Proof of Proposition~\ref{prop:constraint-equivariant}]
The additive constrained problem maximizes the homogeneous Rayleigh quotient
\[
  \dprime^2_{\add}(u)
  =\frac{(u^\top\Delta\mu_Z)^2}{u^\top\bar\Sigma_Z u}
\]
over $u\in\mathbb{R}_+^{d_E+d_I}\cap\mathcal F$. In the shunting effective
coordinates
\[
  \tilde u=Du,
  \qquad
  D=\diag(a_E\1_{d_E},|a_I|\1_{d_I})\succ0,
\]
the shunting problem maximizes the \emph{same} Rayleigh quotient over
\[
  \{\tilde u:\,D^{-1}\tilde u\in\mathbb{R}_+^{d_E+d_I}\cap\mathcal F\}
  =
  D\big(\mathbb{R}_+^{d_E+d_I}\cap\mathcal F\big).
\]
Because $D$ is positive diagonal, it maps the nonnegative cone bijectively to
itself. The objective is constant on $[u]=\{cu:c>0\}$, so only the
projective feasible sets $\mathcal R(\mathcal F)$ and
$\mathcal R(D\mathcal F)$ matter. If those ray sets are equal, then at that
frozen operating point the two coordinate descriptions optimize the same
objective over the same local rays. A constrained tie follows only when a
common optimal ray has a positive, self-consistent shunting representative
$u_{\sh}\in\mathcal F$ whose induced map is this $D$ (and hence
$[Du_{\sh}]$ is the common effective ray). The unconstrained realization in
Theorem~\ref{thm:local-equivalence} need not satisfy a scale-pinning
$\mathcal F$; without an $\mathcal F$-feasible realization, only containment is
established.
If the ray sets differ, the same Rayleigh quotient is maximized over different
directions and the optima need not coincide.

This projective statement prevents a common error. An upper scale cap such as
$\sum_j w_j\le M$ is not, by itself, a local Rayleigh restriction: every
positive ray intersects that ball. The same is true of a single positive
equality normalization at a \emph{frozen} $D$. Such a budget can matter in the
self-consistent shunting problem only because it pins or bounds the operating
scale and thereby changes $D(E_0,I_0)$ or makes a ray unrealizable. Constraints
on E/I ratios, simultaneous E- and I-specific equalities/lower bounds, or other
conditions that remove projective rays can create a direct frozen-$D$ mismatch.
\end{proof}

\subsubsection{Delta-method approximation for shunting}
\label{methods:delta-method}

Let $f(E,I)=E/(1+E+I)$.
Here $E$ and $I$ are nonnegative pooled conductances.
When we use Gaussian moments in this subsection, the Gaussian is a local
approximation to pooled fluctuations around positive means, not a model that
allows negative conductances in the underlying biological circuit.
The gradient and Hessian at operating point $(\bar E,\bar I)$ are:
\begin{align}
  \nabla f &= \frac{1}{\bar D^2}
  \begin{pmatrix} 1+\bar I \\ -\bar E \end{pmatrix},
  \qquad \bar D:=1+\bar E+\bar I,
  \label{eq:gradient}
  \\
  H_f &= \frac{1}{\bar D^3}
  \begin{pmatrix}
    -2(1+\bar I) & \bar E-(1+\bar I) \\
    \bar E-(1+\bar I) & 2\bar E
  \end{pmatrix}.
  \label{eq:hessian}
\end{align}
This also gives a non-asymptotic scope check for the local theorem. For
$\delta=(\delta E,\delta I)$, suppose the segment from $(\bar E,\bar I)$ to
$(\bar E,\bar I)+\delta$ remains in the positive domain with
$1+E+I\ge s_{\min}>0$. Taylor's theorem and the maximum row-sum bound on
Eq.~\eqref{eq:hessian} imply
\begin{equation}
 \left|f(\bar E+\delta E,\bar I+\delta I)-f(\bar E,\bar I)
 -\nabla f(\bar E,\bar I)^\top\delta\right|
 \le \frac{3}{2s_{\min}^2}\lVert\delta\rVert_2^2.
 \label{eq:local-curvature-bound}
\end{equation}
The measured ratio of this exact remainder to the first-order voltage change
is therefore the appropriate local-fidelity diagnostic; morphology labels or
small weights alone do not certify the local regime.
A second-order Taylor expansion yields the delta-method approximations for
the class-conditional mean and variance of $V^{\sh}$:
\begin{align}
  \E[V^{\sh}\mid c]
  &\approx f(\mu_c) + \tfrac12\tr(H_f(\mu_c)\Sigma_c),
  \label{eq:dm-mean}
  \\
  \Var(V^{\sh}\mid c)
  &\approx \nabla f(\mu_c)^\top\Sigma_c\,\nabla f(\mu_c)
  + \tfrac12\tr(H_f\Sigma_c H_f\Sigma_c).
  \label{eq:dm-var}
\end{align}
Substituting into~\eqref{eq:dprime} yields a closed-form surrogate
$\widehat{\dprime}^2_{\sh}(\we,\wi)$ with explicit gradients for
projected-gradient optimization.

The first-order term in~\eqref{eq:dm-var} is the linearized (local) variance;
the second-order correction $\tfrac12\tr(H_f\Sigma_c H_f\Sigma_c)$ is the
curvature penalty, which is nonnegative and vanishes for additive readouts
($H_f=0$).
The approximation is accurate when
$\|\Sigma_c^{1/2}\|/\bar D\ll 1$, i.e., when pooled fluctuations are small
relative to the total conductance.

\subsubsection{Cox/shared-gain Poisson: explicit penalty}
\label{methods:cox}

For Poisson inputs $X_j\mid(g,\theta=c)\sim\mathrm{Poisson}(g\lambda_{c,j})$
with shared gain $g$ satisfying $\E[g]=1$, $\Var(g)=\sigma_g^2$, let
$\bar\lambda=(\lambda_0+\lambda_1)/2$ and
$\Delta\lambda=\lambda_1-\lambda_0$. The exact equal-prior pooled within-class
covariance is
\begin{equation}
  \bar\Sigma
  = \underbrace{\diag(\bar\lambda)
  +\sigma_g^2\bar\lambda\bar\lambda^\top}_{\Sigma_0}
  +\frac{\sigma_g^2}{4}\Delta\lambda\Delta\lambda^\top.
  \label{eq:cox-cov}
\end{equation}
Applying Sherman--Morrison first to $\Sigma_0$ gives
$\Sigma_0^{-1}=D^{-1}-\frac{\sigma_g^2 D^{-1}\bar\lambda\bar\lambda^\top D^{-1}}
{1+\sigma_g^2\bar\lambda^\top D^{-1}\bar\lambda}$
with $D=\diag(\bar\lambda)$, whence
$D^{-1}\bar\lambda=\1$ and $\bar\lambda^\top D^{-1}\bar\lambda=\sum_j\bar\lambda_j$.
The optimal weight ray is
$w^\star\propto D^{-1}\Delta\lambda
-\frac{\sigma_g^2}{1+\sigma_g^2\sum_j\bar\lambda_j}\1(\1^\top\Delta\lambda)$
because the final rank-one update changes $\bar\Sigma^{-1}\Delta\lambda$ only
by a scalar. Define the midpoint-rank-one discriminability
\[
  A
  :=
  \sum_{j}\frac{(\Delta\lambda_j)^2}{\bar\lambda_j}
  -
  \frac{\sigma_g^2}{1+\sigma_g^2\sum_j\bar\lambda_j}
  \Big(\sum_j\Delta\lambda_j\Big)^2.
\]
A second Sherman--Morrison update for the finite-contrast term in
Eq.~\eqref{eq:cox-cov} yields the exact optimal additive discriminability
\begin{equation}
  \dprime^{2,*}_{\add}
  = \frac{A}{1+(\sigma_g^2/4)A}.
  \label{eq:cox-penalty}
\end{equation}

The shared-gain correction inside $A$ grows with $\sigma_g^2$ and with
$|\sum_j\Delta\lambda_j|$ (the total signal that is aligned with the shared-gain
mode); the outer denominator is the additional finite-class-separation
covariance term and tends to one in the local small-$\Delta\lambda$ limit.
Branch-local shunting can reduce the corresponding gain variance through
divisive suppression because $\partial V^{\sh}/\partial g\propto \bar D^{-2}$, but
this is not a generic Poisson advantage: if the additive comparator is
unconstrained or has a matched E/I common-mode cancellation channel, much of this
rank-1 gain can already be removed. The shunting claim is therefore made for the
constrained, fixed-template, positive-cone readouts used in the paper, and for
gain modes that remain accessible to branch-local divisors.

\subsection{Gain--load--alignment theory}
\label{methods:regime-map}

We summarize the gain-vs-load competition with dimensionless parameters.

\emph{Gain magnitude $\Gamma$, alignment $\alpha$, and effective harmful gain
$G_{\mathrm{eff}}$}: let $a$ be the shared-gain loading and
$\Sigma_{\mathrm{ind}}$ the independent-noise covariance. We define
\[
  \Gamma(r):=\sigma_g^2\,a^\top\Sigma_{\mathrm{ind}}^{-1}a,
  \qquad
  \alpha(r):=
  \frac{(\Delta\mu^\top\Sigma_{\mathrm{ind}}^{-1}a)^2}
       {(\Delta\mu^\top\Sigma_{\mathrm{ind}}^{-1}\Delta\mu)
        (a^\top\Sigma_{\mathrm{ind}}^{-1}a)}.
\]
The harmful gain strength along the signal/readout direction is
$G_{\mathrm{eff}}(r):=\alpha(r)\Gamma(r)$. Thus $\alpha=0$ corresponds to an
orthogonal nuisance, $\alpha=1$ to a fully signal-aligned gain mode, and
$G_{\mathrm{eff}}\gg 1$ means signal-aligned gain dominates the accessible
readout variance.

\emph{Residual denominator-variability strength $\Lambda$}:
\[
  \Lambda :=
  \frac{(\bar E/\bar D^2)^2\sigma_\xi^2}{\Var_{\mathrm{ind}}(V)}
\]
for a single divisive stage, where $\xi=\delta L$ is explicitly zero mean and
$\Var_{\mathrm{ind}}(V)$ is the shunting-output variance at the same frozen
operating point with shared gain and the externally injected fluctuation removed.
$\Lambda\gg 1$ means residual
denominator variability is the dominant limitation for shunting. A constant
mean $L_0$ has $\Lambda=0$ and enters through the operating point and $B$;
class-correlated $\Delta L_y$ additionally changes the signal. Raw synthetic
input means labeled 0/3 elsewhere are not values of $\Lambda$.
For a general ratio $R=N/Q$, numerator--denominator covariance makes the
effective added term proportional to
$r_0^2\Var(\delta Q)-2r_0\Cov(\delta N,\delta Q)$ and it can be signed.
We reserve the nonnegative symbol $\Lambda$
for the independent, externally injected $\delta L$ intervention; correlated
pool activity is analyzed with the full covariance expression.

\emph{Baseline fidelity $B(\mathcal T,r)$}:
Relative to the nuisance-free optimized passive baseline,
$B(\mathcal T,r):=d_{0,\sh}^{\prime\,2}(\mathcal T,r)/d_{0,\add}^{\prime\,2}(r)$, measured with
shared gain and residual denominator variability removed. The additive denominator is the optimized
passive additive comparator after deterministic path gains have been absorbed
into reoptimized weights. In fixed passive-cascade diagnostics we instead report
$B_{\mathrm{cas}}(\mathcal T,r):=
d_{0,\sh,\mathrm{cas}}^{\prime\,2}(\mathcal T,r)/
d_{0,\add,\mathrm{cas}}^{\prime\,2}(\mathcal T,r)$; this realized cascade quantity can
vary with morphology for both mechanisms. In either case, values near one mean
the denominator estimates a nuisance while preserving the baseline signal,
whereas values below one flag signal-bearing division that attenuates useful
discriminability.

\emph{Fixed-template and reoptimized rank-one criteria.}
Let $\Delta\mu$ be the signal and
$\Sigma=\Sigma_{\mathrm{ind}}+\sigma_g^2aa^\top$. The nuisance-free optimal
direction is $w_0\propto\Sigma_{\mathrm{ind}}^{-1}\Delta\mu$, and along that
fixed direction
\[
 \frac{\sigma_g^2(w_0^\top a)^2}
 {w_0^\top\Sigma_{\mathrm{ind}}w_0}=\alpha\Gamma.
\]
This gives
$d_{\add,\mathrm{fix}}^{\prime\,2}/d_0^{\prime\,2}
=1/(1+\alpha\Gamma)$, the additive arm of
Eq.~\eqref{eq:regime-caricature}. By contrast, reoptimizing unconstrained LDA
under the perturbed covariance gives exactly
\[
 \frac{\Delta\mu^\top\Sigma^{-1}\Delta\mu}
 {\Delta\mu^\top\Sigma_{\mathrm{ind}}^{-1}\Delta\mu}
 =1-\frac{\alpha\Gamma}{1+\Gamma},
\]
by Sherman--Morrison. These expressions are equal at $\alpha=1$ and differ away
from it; for example, at $\alpha=0.1$ and $\Gamma=100$ they retain $0.091$ and
$0.901$ of baseline discriminability, respectively. This additive identity is
exact for any positive-definite $\Sigma_{\rm ind}$ with one rank-one nuisance
mode. For the shunting surrogate, we additionally assume a direction-preserving
geometry: residual load rescales normalized independent covariance by the scalar
$c_{\mathcal T}=1+\Lambda_{\mathcal T}$, shunting rescales the same nuisance direction to strength
$M_g\Gamma$, and signal--nuisance alignment remains $\alpha$. Under those
assumptions, multiplying by baseline fidelity $B$ yields
Eq.~\eqref{eq:noise-resilience-reoptimized}. Generic tree-induced rotations or
full-covariance changes require numerical reoptimization. The fixed-template comparison yields
Eq.~\eqref{eq:tree-criterion}; the reoptimized comparison uses
$R_{\sh}^{\mathrm{opt}}>R_{\add}^{\mathrm{opt}}$.

For a matched-input comparison, both frozen scores can acquire gain and pool
terms. Writing
\[
 G_j=\frac{d_{0,j}^{\prime2}}{d_{g,j}^{\prime2}}-1,
 \qquad
 L_j=\frac{d_{0,j}^{\prime2}}{d_{l,j}^{\prime2}}-1,
 \qquad j\in\{a,\sh\},
\]
the separated-denominator approximation gives
$\widehat d_{g+l,j}^{\prime2}=d_{0,j}^{\prime2}/(1+G_j+L_j)$.
Comparing the two expressions yields Eq.~\eqref{eq:matched-input-margin}.
These empirical $L_j$ can be signed because sensor corruption can change the
signal and numerator--denominator covariance; they are not identified with the
nonnegative independent-load coordinate $\Lambda_{\mathcal T}$. Likewise $G_j$ is a
reciprocal-$d'^2$ degradation, not pure variance inflation unless the signal gap
is preserved. The approximation accommodates multirank gain but does not
replace full-covariance decoder reoptimization.

For shunting through a tree, local linearization propagates gain and load
sensitivities to the soma. Operationally, under the same direction-preserving
surrogate used above, $M_g$ is the dimensionless factor satisfying
\[
 \frac{\Var_g(w_{\sh}^{\top}h_{\sh})}
      {\Var_{\rm ind}(w_{\sh}^{\top}h_{\sh})}
 =\alpha M_g(\mathcal T,r)\Gamma .
\]
This defines the product $\alpha M_g\Gamma$ even at $\alpha\Gamma=0$. When a
single gain loading is propagated through positive paths, a normalized path
approximation is
\[
  M_g(\mathcal T,r)
  \approx
  \left(
  \frac{\sum_{p\in T}\omega_p\prod_{n\in p}\kappa_n}
       {\sum_{p\in T}\omega_p}
  \right)^2,
\]
where $\kappa_n$ is the local gain-sensitivity factor at node $n$ and
$\omega_p\ge0$ is the operating-point path weight. The normalization makes
$M_g=1$ when every stage transmits the gain unchanged; generic rotations or
signed effective paths require direct variance measurement rather than this
approximation. For independent load injected at distinct nodes, the load term
accumulates as
\[
  \Lambda_{\mathcal T}\approx
  \frac{1}{\Var_{\mathrm{ind}}(V_{\rm soma})}
  \sum_{n\in T}
  \left(\frac{\partial V_{\rm soma}}{\partial V_n}\right)^2
  \left(\frac{\bar N_n}{\bar D_n^2}\right)^2\sigma_{\xi,n}^2,
\]
where $N_n=E_n+C_n$ is the complete local numerator, including child current,
$\bar D_n=1+\bar E_n+\bar I_n+\sum_b g_{n\leftarrow b}=1+\bar T_n$ is the
operating-point denominator, and
$\partial V_{\rm soma}/\partial V_n$ is only the downstream
transfer from node voltage to soma. This avoids both omitting child current and
counting the local denominator derivative twice.
Combining baseline fidelity, surviving gain, and load yields
Eq.~\eqref{eq:regime-caricature}, the sign rule in
Eq.~\eqref{eq:tree-criterion}, and the depth increment in
Eq.~\eqref{eq:depth-advantage-margin}. The familiar tie boundary
$(1-\gamma_g^{2L})G_{\mathrm{eff}}=\Lambda_L$ is the special chain case with
$\alpha=1$, $B(\mathcal T,r)\approx1$, $M_g(\mathcal T,r)=\gamma_g^{2L}$, and
depth-accumulated load $\Lambda_L$.
The tree caricature holds the nuisance-free readout direction fixed and is a leading-order
(small-penalty) approximation rather than a re-solved cone program at every grid
point: it tracks the gain-induced variance inflation and the accumulated
denominator load to first order, and treats the gain and load contributions as
decoupling in the denominator. Its additive form is exact for the
fixed-direction rank-one covariance model. It agrees to first order with the
subtractive Cox/shared-gain penalty in Eq.~\eqref{eq:cox-penalty} once
$G_{\mathrm{eff}}=\alpha\Gamma$ is identified with the signal-aligned
gain-to-independent variance ratio \emph{and} the finite-contrast correction is
local, $(\sigma_g^2/4)A\ll1$. Thus the reduced
caricature and the subtractive Cox optimum describe the same shared-gain effect
in their common local regime. Figure~\ref{fig:regime-map}C instead displays the
fixed-direction and covariance-reoptimized rank-one values separately.
Here the per-stage gain-suppression factors $\kappa_1,\ldots,\kappa_L$ whose
geometric mean is $\gamma_g$ are distinct from the axial coupling coefficients
$\kappa_{n\leftarrow b}$ of the tree linearization
(Eq.~\eqref{eq:tree-coefs-kappa}); the two are related but not assumed equal.

\paragraph{What is held fixed in the regime map.}
Figure~\ref{fig:regime-map} is not a fresh positive-cone synaptic optimization at
every grid point. It is a reduced pooled-drive analysis that holds the baseline
signal and independent variance fixed while varying shared-gain strength and
denominator load. Panel B shows the common fully aligned boundary; panel C
separates a fixed nuisance-free decoder from covariance-reoptimized unconstrained
LDA. Neither curve resolves a binding positive-E/I cone or globally optimizes the
nonlinear shunting map. The random-library local optimization in
Supplementary Fig.~\ref{fig:local-equivalence-supp} is the separate check that the
matched local comparator is implemented correctly.

The same distinction applies to the empirical diagnostics. The V1 analysis is
reported as a fixed-budget population-data diagnostic: it compares outer-fold
held-out ranking, threshold error, and calibration across candidate access and a
prespecified effective leak/input-scale grid, with training-only scale selection
and label-free running correction as sensitivities. We do not fit a unique
$(G_{\mathrm{eff}},\Lambda)$ coordinate for V1, because that would require an
additional generative model for the recorded population and the unobserved
denominator noise. The V1 result therefore tests a bounded fixed-feature
signature; it does not identify a load coordinate or estimate the exact
phase-boundary location.

\subsubsection{Additive tree flattening}
\label{methods:tree-flattening}

Algorithm~\ref{alg:tree-flatten} computes the effective path gains used in
Proposition~\ref{prop:additive-flattening}.

\begin{algorithm}[t]
\caption{Additive tree flattening: leaf-to-soma path-gain accumulation yielding the equivalent single-layer effective weights $a_n^{\add}$ (Prop.~\ref{prop:additive-flattening}).}
\label{alg:tree-flatten}
\begin{algorithmic}[1]
\REQUIRE Rooted tree $\mathcal{T}$ with soma $s$; for each node $n$ its parent $p(n)$, children $\{b: p(b)=n\}$, axial couplings $g_{n\leftarrow b}\ge0$; mode $\in\{\text{analytic},\text{current}\}$.
\ENSURE Path gains $\{a_n^{\add}\}_{n\in\mathcal{T}}$ with $V_s^{\add}=\sum_{n\in\mathcal{T}} a_n^{\add}(E_n-I_n)$.
\FOR[{set self- and child-edge transfers $t$}]{each node $n\in\mathcal{T}$}
  \IF[{leak-normalized analytic cascade}]{mode $=$ analytic}
    \STATE $Z_n^{\add} \gets 1+\sum_{b'} g_{n\leftarrow b'}$
    \STATE $t_{n\leftarrow n}\gets 1/Z_n^{\add}$;\quad $t_{n\leftarrow b}\gets g_{n\leftarrow b}/Z_n^{\add}$ for each child $b$
  \ELSE[{raw current-mode}]
    \STATE $t_{n\leftarrow n}\gets 1$;\quad $t_{n\leftarrow b}\gets g_{n\leftarrow b}$ for each child $b$
  \ENDIF
\ENDFOR
\STATE $q_s\gets 1$;\quad $a_s^{\add}\gets q_s t_{s\leftarrow s}$
\COMMENT{empty soma path product}
\FOR{each non-soma node $n\in\mathcal{T}$ in order of nondecreasing depth from $s$}
  \STATE $q_n\gets q_{p(n)}t_{p(n)\leftarrow n}$;\quad
  $a_n^{\add}\gets q_n t_{n\leftarrow n}$
  \COMMENT{edge-path product times the source self-gain}
\ENDFOR
\STATE \RETURN $\{a_n^{\add}\}_{n\in\mathcal{T}}$
\end{algorithmic}
\end{algorithm}

\begin{proposition}[Exact flattening]
\label{prop:additive-flattening}
For any rooted tree with passive additive recursion
$V_n^{\add}=\eta_n(E_n-I_n)+\sum_b \tau_{n\leftarrow b}V_b^{\add}$ and
$\eta_n,\tau_{n\leftarrow b}\ge0$,
the somatic voltage is an exact linear function of all local signed drives:
$V_s^{\add}=\sum_{n\in\mathcal{T}} a_n^{\add}(E_n-I_n)$.
\end{proposition}

\begin{proof}
Define $t_{n\leftarrow n}:=\eta_n$ and
$t_{n\leftarrow b}:=\tau_{n\leftarrow b}$. The recursion is already linear:
$V_n^{\add}=t_{n\leftarrow n}(E_n-I_n)+\sum_b t_{n\leftarrow b}V_b^{\add}$.
Unrolling from leaves to soma:
$V_s^{\add}=\sum_{n\in\mathcal{T}} a_n^{\add}(E_n-I_n)$
where the path gain $a_n^{\add}=t_{p_L\leftarrow p_{L-1}}\cdots t_{p_1\leftarrow p_0}\cdot t_{p_0\leftarrow p_0}$ for the unique path $n=p_0\to p_1\to\cdots\to p_L=s$.
For the leak-normalized analytic cascade,
$t_{n\leftarrow n}=1/(1+\sum_{b'}g_{n\leftarrow b'})$ and
$t_{n\leftarrow b}=g_{n\leftarrow b}/(1+\sum_{b'}g_{n\leftarrow b'})$, so distal
path gains are attenuated multiplicatively. For the raw current-mode
implementation, $t_{n\leftarrow n}=1$ and $t_{n\leftarrow b}=g_{n\leftarrow b}$;
the same flattening identity holds, but fixed-budget scaling and gradients can
differ from the leak-normalized analytic cascade.
\end{proof}

\subsubsection{Shunting tree linearization}
\label{methods:tree-linearization}

At each node $n$, differentiating~\eqref{eq:node-voltage} at the operating point gives
\begin{equation}
  \delta V_n^{\sh}
  \approx
  \alpha_n\,\delta E_n - \beta_n\,\delta I_n
  + \sum_b \kappa_{n\leftarrow b}\,\delta V_b^{\sh},
  \label{eq:tree-linearization}
\end{equation}
with explicit coefficients:
\begin{align}
  \alpha_n &= \frac{1+\bar I_n+\sum_b g_{n\leftarrow b}(1-\bar V_b)}{\bar D_n^2},
  \qquad
  \beta_n = \frac{\bar E_n+\sum_b g_{n\leftarrow b}\bar V_b}{\bar D_n^2},
  \label{eq:tree-coefs-ab}
  \\
  \kappa_{n\leftarrow b} &= \frac{g_{n\leftarrow b}}{\bar D_n},
  \qquad
  \bar D_n = 1+\bar E_n+\bar I_n+\sum_{b'} g_{n\leftarrow b'}=1+\bar T_n.
  \label{eq:tree-coefs-kappa}
\end{align}
Unrolling yields~\eqref{eq:sh-tree-linearize} with shunting path gains
$a_n^{\sh}$ being products of $\kappa$ and $\alpha$ factors along the path.

For the two-stage chain (distal node 1, proximal node 0, coupling $g$):
$\alpha_1=(1+\bar I_1)/\bar D_1^2$,
$\beta_1=\bar E_1/\bar D_1^2$,
$\kappa=g/\bar D_0$, and
$\rho_1=\bar E_1/(1+\bar I_1)$,
$\rho_0=(\bar E_0+g\bar V_1)/(1+\bar I_0+g(1-\bar V_1))$.

\subsubsection{Tree optimization and self-consistent chains}
\label{methods:tree-optimization}

For additive trees, the exact linearity~\eqref{eq:add-tree-flatten} means
optimizing $\dprime^2$ over all synaptic weights is a single cone-constrained
LDA problem on the morphology-weighted super-library.
If nonlinear post-voltage activations are inserted between tree stages, this exact
flattening no longer holds globally; the exact result is specific to passive
additive recursion.

For shunting trees, all nodewise linearization coefficients depend jointly on
the induced operating points. A putative local optimum must therefore satisfy a
coupled fixed-point condition. For a two-stage chain one can define
\begin{equation}
  (w_{e,0}^\star,w_{e,1}^\star,w_{i,0}^\star,w_{i,1}^\star)
  =
  \mathcal{T}_{\mathrm{chain}}(w_{e,0}^\star,w_{e,1}^\star,w_{i,0}^\star,w_{i,1}^\star),
  \label{eq:chain-fixed-point}
\end{equation}
where $\mathcal{T}_{\mathrm{chain}}$ is the composition
``weights $\mapsto$ operating points $\mapsto$ linearization coefficients
$\mapsto$ exact two-orientation cone-LDA $\mapsto$ weights.'' Unlike the
single-compartment ray construction, naive repeated application of this map has
no general contraction or global-optimality guarantee. Any numerical candidate
would need explicit primal/dual feasibility and fixed-point-residual checks.
We therefore evaluate specified pooled-drive recursions in the analytic tree
experiments and optimize the full nonlinear loss by backpropagation in trained
trees; we do not claim a globally optimized shunting-tree cone solution.
With a smooth monotone post-voltage activation $\phi$, the appropriate local
object is the activated-tree Jacobian,
\begin{equation}
  d_{q,\phi}^{\prime\,2}
  \approx
  \frac{\big(J_{q,\phi}\Delta\mu\big)^2}
       {J_{q,\phi}\bar\Sigma J_{q,\phi}^{\top}},
  \qquad
  q\in\{\add,\sh\}.
  \label{eq:activated-tree-dprime}
\end{equation}
For the shifted-tanh family used in the learned-network runs (with parameters
allowed to differ by mechanism),
\[
  \phi_n^q(V)=\frac{1+\tanh(\kappa_{\phi,n}^q(V-b_n^q))}{2},
  \qquad
  \gamma_n^q=(\phi_n^q)'(\bar V_n^q)
  =\frac{\kappa_{\phi,n}^q}{2}\operatorname{sech}^2\!\big(\kappa_{\phi,n}^q(\bar V_n^q-b_n^q)\big).
\]
Equivalently, each passive local derivative is multiplied by the relevant
node-wise slope before it is propagated to the soma. For an activated additive
branch,
\begin{equation}
  \delta Y_n^{\add}
  =
  \gamma_n^{\add}
  \left(
    \delta E_n-\delta I_n+\sum_c a_{nc}\delta Y_c^{\add}
  \right),
  \label{eq:activated-additive-local}
\end{equation}
whereas if the passive shunting voltage is $V_n=N_n/D_n$, then at the
activated-tree operating point
\begin{equation}
  \delta Y_n^{\sh}
  =
  \gamma_n^{\sh}
  \left(
    \frac{\delta N_n}{\bar D_n}
    -
    \frac{\bar N_n}{\bar D_n^2}\delta D_n
  \right).
  \label{eq:activated-shunting-local}
\end{equation}
Unrolling either recursion shows why a terminal activation and internal
activation are not equivalent.  The somatic slope is common to every input
path and cancels from the local scalar $d'^2$, but each internal path also
acquires the product of the slopes of the branches it traverses.  Paths of
different depth or operating point are therefore reweighted differently.
Saturation or class-dependent operating points add further operating-point
dependence and can change the full-tree optimum beyond the passive local theorem.

For a common-gain sweep, the corresponding scalar chain rule is exact wherever
$Y=\phi(V(g))>0$ and $V\ne0$:
\begin{equation}
 \eta_Y(g)
 :=\frac{\partial\log Y}{\partial\log g}
 =\frac{V\phi'(V)}{\phi(V)}\,
  \frac{\partial\log|V|}{\partial\log g}
 =A_\phi(V)\eta_V(g).
 \label{eq:activated-elasticity-appendix}
\end{equation}
For $\phi(V)=[1+\tanh(\kappa_\phi(V-b))]/2$,
$A_\phi(V)=\kappa_\phi V[1-\tanh(\kappa_\phi(V-b))]
=2\kappa_\phi V[1-\phi(V)]$.
Thus an activated additive response has elasticity $A_\phi(V)$ rather than a
constant one, whereas an activated shunting response has
$A_\phi(V)/(1+gT_0)$ in the positive scalar scaling example.  Activation
saturation and divisive gain suppression are distinct multiplicative factors.

For Supplementary Fig.~\ref{fig:activation-sensitivity-supp}, we avoid
confounding activation with the different native voltage scales of the two
rules. Let $\bar V_q>0$ be the passive baseline voltage for
$q\in\{\add,\sh\}$ and define the matched terminal gate
\begin{equation}
 \phi_q(V;\kappa,z)
 =\frac{1+\tanh\!\left[\kappa(V/\bar V_q-1)-z\right]}{2}.
 \label{eq:matched-terminal-activation}
\end{equation}
Equivalently, $\kappa_{\phi,q}=\kappa/\bar V_q$ and
$b_q=\bar V_q(1+z/\kappa)$. At $z=0$ both rules are centered at their own
baseline and have the same dimensionless local multiplier
$A_{\phi_q}(\bar V_q)=\kappa$. Varying $\kappa$ therefore tests curvature while
holding the baseline comparison locally matched; varying $z$ tests midpoint
or operating-point mismatch. We apply this gate to the exact Monte Carlo
voltages generated on the same $7\times9$ gain--load grid as
Fig.~\ref{fig:ei-circuit}D, with eight paired seeds and 120,000
trials per seed. The gate is terminal only; applying it recursively would
instead require the pathwise Jacobian products above.

Under the same direction-preserving variance decomposition used for
Eq.~\eqref{eq:regime-caricature}, the gain--load--alignment principle retains its
algebraic form for activated trees if $B$, $M_g$, and $\Lambda_{\mathcal T}$ are all
defined from $J_{q,\phi}$ (or from the corresponding finite-amplitude activated
representations).  Reusing their passive values would not be valid: internal
$\gamma_n$ factors can alter both the surviving gain mode and baseline signal
fidelity.

\subsubsection{Ordered depth importance and Schur increments}
\label{methods:importance}

For a Gaussian representation $Z=(Z_1,\ldots,Z_L)$ and an explicitly chosen
conditioning set $C$, the Schur contribution of block $\ell$ beyond $C$ is
\begin{equation}
  \mathcal{I}_{\ell\mid C}
  =
  \Delta\mu_{\ell\mid C}^\top\Sigma_{\ell\mid C}^{-1}
  \Delta\mu_{\ell\mid C},
  \label{eq:schur-layer}
\end{equation}
where $\Sigma_{\ell\mid C}=\Sigma_{\ell\ell}-\Sigma_{\ell C}\Sigma_{CC}^{-1}\Sigma_{C\ell}$
is the conditional covariance and
$\Delta\mu_{\ell\mid C}=\Delta\mu_\ell-\Sigma_{\ell C}\Sigma_{CC}^{-1}\Delta\mu_C$
is the residual mean difference.
For the ordered profile in Supplementary Fig.~\ref{fig:tree-depth}C, layer 1 is
distal and indices increase toward the soma; the code sets
$C_\ell=\{1,\ldots,\ell-1\}$ and equivalently computes the increment in
$d'^2(Z_1,\ldots,Z_\ell)$ when the next, more-proximal layer is appended. By
block matrix inversion, these sequential increments sum to the Fisher
discriminability of the full ordered representation, and their normalized
profile $\pi_\ell$ sums to one. This attribution is exact for the declared
distal-to-proximal order but is not order invariant. Appendix
Supplementary Fig.~\ref{fig:tree-depth}E instead conditions each of two sibling
branches on its complement to compare that branch-specific Schur value with a
matching ablation; those two conditional values are not treated as one additive
partition.

\paragraph{Spatial placement of correlated inputs.}
The Schur increment makes explicit why branch placement matters. If two
gain-correlated inputs land on the same branch, a branch-local divisor can reduce
their shared multiplicative component before that branch is combined with the
rest of the tree. This lowers the covariance that downstream compartments see.
If the same correlated inputs are spread across different branches, each branch
sees only part of the nuisance, and a shared gain component can remain visible
at the soma until a later, less local divisor is applied. Conversely, if useful
signal dimensions are split across branches but denominator noise is local and
large, the tree can pay multiple load penalties for the same downstream signal.
The depth simulations therefore test not only the algebraic form of division but
also whether divisor locality matches the spatial locality of nuisance
variation.

\begin{lemma}[Depth-wise gain-correlation bound under branch-specific divisors]
\label{lem:decorrelation}
Fix a depth $\ell$ and two sibling branch contributions at the soma,
$Y_{\ell,1},Y_{\ell,2}$, under shared gain $g$.
Assume the decomposition
\[
Y_{\ell,j}=a_\ell Z_g+\varepsilon_{\ell,j},\qquad j\in\{1,2\},
\]
where $Z_g$ is the shared gain component with variance $\sigma_g^2$,
$a_\ell=a_0\prod_{k=1}^{\ell}\kappa_k$ with $0<\kappa_k<1$,
and private terms satisfy
$\Cov(\varepsilon_{\ell,1},\varepsilon_{\ell,2})=0$ and
$\Cov(\varepsilon_{\ell,j},Z_g)=0$ for $j\in\{1,2\}$, and
$\Var(\varepsilon_{\ell,j})\ge \sigma_{\mathrm{priv}}^2>0$.
Then the sibling correlation satisfies
\begin{equation}
  \eta_\ell^{\sh}
  :=
  \Corr(Y_{\ell,1},Y_{\ell,2})
  \le
  \frac{a_0^2\sigma_g^2}{\sigma_{\mathrm{priv}}^2}
  \prod_{k=1}^{\ell}\kappa_k^2.
  \label{eq:decorrelation-bound}
\end{equation}
If $\kappa_k\equiv\kappa$, this gives the explicit geometric bound
$\eta_\ell^{\sh}\le C\kappa^{2\ell}$ with
$C=a_0^2\sigma_g^2/\sigma_{\mathrm{priv}}^2$.
\end{lemma}

\begin{proof}
Under the stated model,
\[
\Cov(Y_{\ell,1},Y_{\ell,2})
=
a_\ell^2\sigma_g^2,
\]
because the private terms are cross-sibling independent and independent of $Z_g$.
Also,
\[
\Var(Y_{\ell,j})
=
a_\ell^2\sigma_g^2+\Var(\varepsilon_{\ell,j})
\ge
a_\ell^2\sigma_g^2+\sigma_{\mathrm{priv}}^2.
\]
Hence
\[
\eta_\ell^{\sh}
=
\frac{a_\ell^2\sigma_g^2}
{\sqrt{\Var(Y_{\ell,1})\Var(Y_{\ell,2})}}
\le
\frac{a_\ell^2\sigma_g^2}{\sigma_{\mathrm{priv}}^2}
=
\frac{a_0^2\sigma_g^2}{\sigma_{\mathrm{priv}}^2}
\prod_{k=1}^{\ell}\kappa_k^2,
\]
proving~\eqref{eq:decorrelation-bound}.
\end{proof}

\subsubsection{Two-stage chain: explicit decomposition}
\label{methods:two-stage}

For the additive two-stage chain with signed drives $U_j=E_j-I_j$ ($j\in\{0,1\}$),
the soma output is $V_0^{\add}=a_0 U_0+a_1 U_1$ with
$a_0=1/(1+g)$, $a_1=g/(1+g)$.
The discriminability is
$\dprime^2=(a_0\Delta\mu_0+a_1\Delta\mu_1)^2/(a_0^2\sigma_0^2+a_1^2\sigma_1^2+2a_0 a_1\sigma_{01})$.
The optimized two-coordinate readout decomposes as
\begin{equation}
  \Delta\mu^\top\Sigma^{-1}\Delta\mu
  =
  \frac{\Delta\mu_0^2}{\sigma_0^2}
  +
  \frac{\left(\Delta\mu_1-\frac{\sigma_{01}}{\sigma_0^2}\Delta\mu_0\right)^2}
       {\sigma_1^2-\sigma_{01}^2/\sigma_0^2}.
  \label{eq:chain-decomposition}
\end{equation}
The decomposition~\eqref{eq:chain-decomposition} follows from the block inversion
identity for a $2\times 2$ matrix: if
$\Sigma=\bigl(\begin{smallmatrix}\sigma_0^2&\sigma_{01}\\\sigma_{01}&\sigma_1^2\end{smallmatrix}\bigr)$,
then $\Delta\mu^\top\Sigma^{-1}\Delta\mu=\Delta\mu_0^2/\sigma_0^2+(\Delta\mu_1-(\sigma_{01}/\sigma_0^2)\Delta\mu_0)^2/(\sigma_1^2-\sigma_{01}^2/\sigma_0^2)$.

\emph{Gain- and load-sensitivity in the chain.}
Keep the axial coupling $g$ fixed and let a separate common input gain $z$
multiply baseline drives, $E_j(z)=z\widetilde E_j$ and
$I_j(z)=z\widetilde I_j$. At the reference $z=1$, the linearized somatic
gain sensitivity is
$\left.\partial V_0^{\sh}/\partial z\right|_{z=1}
\approx\alpha_0\widetilde E_0-\beta_0\widetilde I_0
+\kappa(\alpha_1\widetilde E_1-\beta_1\widetilde I_1)$.
For additive inhibitory loads $I_j\mapsto I_j+\xi_j$, the sensitivities are
$\partial V_0^{\sh}/\partial\xi_0\approx-\beta_0$,
$\partial V_0^{\sh}/\partial\xi_1\approx-\kappa\beta_1$.
The load-induced output variance is therefore
$\Var_\xi(V_0)\approx\beta_0^2\sigma_{\xi_0}^2+(\kappa\beta_1)^2\sigma_{\xi_1}^2$,
making explicit that load penalties compound with depth while gain sensitivity
can be suppressed.

\subsubsection{Multiclass extension}
\label{methods:multiclass}

For $K$ classes with equal priors $\pi_k=1/K$, class-conditional means
$\mu_k=\E[Z\mid c{=}k]$, and a common within-class covariance
$\Sigma_k=\Sigma=:S_W$ for every class, the average pairwise $\dprime^2$ along
a linear score $y=w^\top Z$ satisfies
\begin{equation}
\overline{\dprime^2}(w)
\;=\;
\frac{2K}{K-1}\;\frac{w^\top S_B\,w}{w^\top S_W\,w},
\label{eq:avg-dprime-multiclass}
\end{equation}
where $S_B=\frac{1}{K}\sum_k(\mu_k-\bar\mu)(\mu_k-\bar\mu)^\top$ is the
between-class scatter.
The identity follows from
$\sum_{a<b}(\mu_a{-}\mu_b)(\mu_a{-}\mu_b)^\top = K^2 S_B$ (a one-line
ordered-pair expansion). With heterogeneous class covariances, averaging
pair-specific ratios does not generally equal this single generalized Rayleigh
quotient; that case requires an explicitly chosen pooled denominator or a
heteroscedastic objective.

\emph{Dendritic constraints.}\;
For additive trees and linearized shunting trees, the somatic output has
the form $w=Au$, $u\ge 0$, where $A$ encodes sign flips and
morphology-dependent path gains (Sec.~\ref{methods:cone-lda}).
In this parameterization, optimizing \eqref{eq:avg-dprime-multiclass}
becomes a nonnegativity-constrained generalized Rayleigh quotient in~$u$:
\[
  \max_{u\ge 0}\;
  \frac{u^\top(A^\top S_B A)\,u}{u^\top(A^\top S_W A)\,u},
\]
which is structurally identical to the two-class cone-constrained Fisher
optimization with $\Delta\mu\Delta\mu^\top$ replaced by $S_B$.
The KKT conditions (stationarity on the active set, complementary slackness)
carry over because the dual-cone structure of the feasible set is
unchanged.

\emph{Beyond local shunting.}\;
When the exact shunting nonlinearity matters (shared gain, denominator variability),
one can still use the average pairwise $\dprime^2$ computed from the
scalar output distribution.  A delta-method expansion around a common
operating point yields a linear surrogate with an \emph{effective}
within-scatter matrix (Jacobian pullback), reducing the problem to
\eqref{eq:avg-dprime-multiclass} with a modified $S_W$.
If class covariances differ substantially (heteroscedasticity),
quadratic-discriminant-analysis (QDA) objectives become appropriate; the
present framework focuses
on the homoscedastic setting.

\subsection{Information and population scaling}
\subsubsection{Information-theoretic diagnostics}
\label{methods:info-theory}

The $\dprime^2$-based comparisons are complemented by a Gaussian analytic map
and by non-Gaussian cross-validated decoder lower bounds.

\emph{Binary equal-variance Gaussian surrogate.}
For an equiprobable binary stimulus and a scalar readout $V$ with equal-variance
Gaussian class-conditional distributions,
the mutual information between stimulus $\theta$ and readout is
\begin{equation}
  \MI(\theta;V)
  =
  1 - \E_{Z\sim\mathcal{N}(0,1)}\!\bigl[\log_2\bigl(1+e^{-\dprime\, Z - \dprime^2/2}\bigr)\bigr],
  \label{eq:mi-biawgn}
\end{equation}
which is the capacity of a binary-input additive white Gaussian noise channel
(BIAWGN) with signal-to-noise ratio (SNR) $\dprime^2/4$. This provides an
analytic monotone mapping
from $\dprime^2$ to bits only under those channel assumptions; it is not used as
empirical mutual information for non-Gaussian readouts.

\emph{Conditional decoder information.}
The exact conditional mutual information is
$\MI(\theta;Y\mid X)=\MI(\theta;X,Y)-\MI(\theta;X)$.
Under the equal-covariance Gaussian surrogate, the Schur complement supplies an
\emph{increment in discriminability}, not bits. If
$\dprime^2_{X,Y}=\dprime_X^2+\Delta\dprime^2_{Y\mid X}$, then the corresponding
channel approximation is
\[
\MI(\theta;Y\mid X)
\approx
f_{\mathrm{BIAWGN}}\!\left(\dprime_X^2+\Delta\dprime^2_{Y\mid X}\right)
-f_{\mathrm{BIAWGN}}\!\left(\dprime_X^2\right),
\]
where $f_{\mathrm{BIAWGN}}$ denotes Eq.~\eqref{eq:mi-biawgn} as a function of
$\dprime^2$. Only when $X$ is uninformative does this reduce to applying the map
to the Schur increment alone.
For non-Gaussian rate inputs or nonlinear learned representations, we use
cross-validated probabilistic decoders to estimate the variational lower bound
$I_{\rm LB}(q,V)=H(\theta)+\mathbb E[\log_2 q(\theta\mid V)]$. Each decoder's
$I_{\rm LB}$ is a lower bound on its representation's mutual information, but a
difference of two lower bounds is not a lower bound on the corresponding mutual-
information difference and need not preserve its sign. When two decoders share
the same held-out labels, the entropy term cancels, so their $\Delta I_{\rm LB}$
is exactly the additive-minus-shunting held-out categorical-log-loss difference
in bits. We therefore treat
it as an operational decoder-bound score rather than an information ordering.
We do not label residualized decoder increments as partial-information-
decomposition (PID) atoms or unique
information without an identified PID functional.

\subsubsection{Population-coding nuisance model and representation-linearization metrics}
\label{methods:population-linearization}

We model population responses under latent state variation as
\[
  X = g\,f(\theta)+\varepsilon,
\]
where $X$ is the population activity vector, $\theta$ is the stimulus, $g$ is a
low-dimensional state/contrast nuisance factor,
and $\varepsilon$ is independent residual variability with covariance
$\Sigma_{\mathrm{ind}}$.
The equivalent covariance form is Eq.~\eqref{eq:latent-covariance}.
The alignment scalar $\alpha(\theta)$ in Eq.~\eqref{eq:alignment-alpha} quantifies when
shared variability is readout-limiting for constrained linear readout.

For tree analyses, we treat the vector of branch/node voltages at depth $\ell$ as a
representation:
\[
  h^{(\ell)} = \big(V_1^{(\ell)},\dots,V_{B_\ell}^{(\ell)}\big).
\]
We characterize representations across depth with five metrics:

\emph{(i) Linear discriminability/linear Fisher information (LFI) at depth.}
For a two-class contrast, we compute
\begin{equation}
  d_{\mathrm{lin}}'^2\!\big(h^{(\ell)}\big)
  :=
  \Delta\mu_\ell^\top \Sigma_\ell^{-1}\Delta\mu_\ell,
  \label{eq:lfi-depth}
\end{equation}
where $\Delta\mu_\ell$ and $\Sigma_\ell$ are the class-mean difference and pooled covariance
of $h^{(\ell)}$.

\emph{(ii) Linear-probe decodability.}
We train a linear decoder from $h^{(\ell)}$ to class label and track held-out accuracy
versus depth $\ell$. In the controlled aligned-support diagnostic, shunting can
increase probe performance with depth when local divisors suppress harmful gain
while preserving signal fidelity; the effect can reverse when accumulated load
or fidelity loss dominates. Gain magnitude alone does not determine the sign.

\emph{(iii) Gain-sensitivity index.}
We report
\begin{equation}
  S_g(\ell)
  :=
  \frac{\mathbb{E}_{s}\!\left[
    \tr\Cov_{g}\!\left(\mathbb{E}[h^{(\ell)}\!\mid s,g]\right)\right]}
       {\mathbb{E}_{s}\!\left[
    \tr\Cov\!\left(h^{(\ell)}\!\mid s\right)\right]},
  \label{eq:gain-sensitivity-index}
\end{equation}
which isolates the fraction of depth-$\ell$ representation variance attributable to
latent gain fluctuations.

\emph{(iv) Nuisance-mode fraction.}
We quantify low-rank nuisance dominance in $h^{(\ell)}$ by the top-eigenvalue fraction
of within-class covariance,
\begin{equation}
  \nu_\ell := \frac{\lambda_1\!\left(\Cov(h^{(\ell)} \mid y)\right)}
                   {\sum_k \lambda_k\!\left(\Cov(h^{(\ell)} \mid y)\right)},
  \label{eq:nuisance-fraction}
\end{equation}
reported in Supplementary Fig.~\ref{fig:tree-depth}A.

\emph{(v) Signal--gain principal angle.}
We define a gain tangent at depth $\ell$ by a context split (high- vs low-gain means)
$a_\ell$ and report the principal angle between $a_\ell$ and the stimulus tangent
$\Delta\mu_\ell$ in the representation-geometry diagnostics.

\subsubsection{Exchangeable correlation model and saturation}
\label{methods:exchangeable}

A useful special case for understanding synapse selection and depth effects
uses exchangeable within-pool correlations.
Let $K$ synapses have equal signal $\delta$, variance $\sigma^2$, and
pairwise correlation $\rho\in[0,1)$.
The pooled $\dprime^2$ from the $K$ synapses is
\begin{equation}
  \dprime^2(K;\rho)
  =
  \frac{K\,\delta^2}{\sigma^2\bigl(1+(K-1)\rho\bigr)}.
  \label{eq:exchangeable-dprime}
\end{equation}
For $\rho=0$, discriminability grows linearly in $K$ (independent synapses).
For $\rho>0$, it saturates to $\delta^2/(\sigma^2\rho)$ as $K\to\infty$.
This scalar exchangeable model illustrates one decoder-accessible bottleneck:
reducing the effective $\rho$ raises its ceiling. A matched branch-local divisor
can reduce the shared-gain contribution under the assumptions of
Lemma~\ref{lem:decorrelation}; the current depth surrogate motivates, but does
not by itself establish, this mechanism in trained trees.
Combining this saturation law with Eq.~\eqref{eq:regime-caricature} gives the
population-level interpretation of the same morphology term: when branch-local
division lowers the surviving gain factor $M_g(\mathcal T,r)$, the effective shared
correlation floor $\rho_{\mathrm{eff}}$ falls in the one-dominant-nuisance
caricature, increasing the accessible ceiling
$\delta^2/(\sigma^2\rho_{\mathrm{eff}})$ and hence the monotone scalar
discriminability summary. When a bit-valued summary is needed for a binary
stimulus, we use the BIAWGN mapping in Eq.~\eqref{eq:mi-biawgn} rather than the
Gaussian-input capacity expression.
Within this exchangeable model, reducing signal-aligned covariance raises the
asymptote of the constrained linear readout. It does not imply increased mutual
information in the full input representation or establish empirical population
scaling.

\subsubsection{Population scaling with the implemented readouts and a differential-noise ceiling}
\label{methods:identified-ceiling}

We evaluate population scaling using the same terminal equations as the network
implementation. For balanced $y\in\{-1,1\}$ and population unit $j$,
\begin{equation}
 E_j=cG\left(\mu+y\frac{\Delta}{2}+\varepsilon_j\right),
 \qquad
 I_j=c\kappa G L_j,
 \label{eq:code-faithful-population-generator}
\end{equation}
where $G$ is one mean-one global lognormal gain, $L_j$ is mean-one lognormal
sensor load, and $\varepsilon_j$ is private Gaussian variability. In the
independent/mismatched-sensor condition, $I_j$ instead contains an independent
gain $G'_j$. This intervention changes nuisance correlation and rank; it is not
a literal support permutation. We fix $\mu=2$, $\Delta=0.4$,
$\sigma_\varepsilon=0.35$, and $\kappa=1$, and cross
$P\in\{1,2,4,8,16,32,64,128,256\}$,
$c\in\{0.25,1,4,16,64\}$, gain log-SD
$s_g\in\{0,0.2,0.45,0.7\}$, sensor-load log-SD
$s_L\in\{0,0.08,0.26,0.55\}$, both sensor relations, and eight seeds.
For a mean-one lognormal gain, the linear-scale variance used in the theory is
$\sigma_g^2=\exp(s_g^2)-1$.
The complete Cartesian grid is retained. The $L_j$ fluctuation is private
across population units, so this experiment does not test a shared denominator
load that could impose its own off-diagonal covariance floor.

Every comparator receives the same $P$ paired observations. We compare raw
additive $E-I$; the affine tangent of $E/(1+E+I)$ at
$(E_0,I_0)=(c\mu,c\kappa)$; an optimized symmetric positive-E/I response
$E-\beta_P I$; the implemented shunt $E/(1+E+I)$; and the leak-free limit
$E/(E+I)$ as a high-conductance diagnostic. Closed-form moments of the
unclipped generator set
\begin{equation}
 \beta_P=\max\!\left\{0,
 \frac{\Cov(\sum_{j=1}^{P}E_j,\sum_{j=1}^{P}I_j)}
 {\Var(\sum_{j=1}^{P}I_j)}\right\}.
 \label{eq:code-faithful-population-beta}
\end{equation}
Because $I$ has no class-mean difference and the generator is exchangeable,
this minimizes pooled within-class variance in the symmetric positive-E/I
linear family. At $P=2,8$, unrestricted $2P$ LDA gives the same ray and test
$d'^2$ to relative error below $7.1\times10^{-15}$; positivity does not bind
anywhere in the grid. This comparator estimates the performance of the linear
family rather than generic nonlinear-network capacity.

For each nonlinear cell and seed, independent calibration and test splits each
contain 100,000 balanced trials. Means and diagonal variances are sampled;
off-diagonal covariance is Rao--Blackwellized from
$\Var_G[\E(V_j\mid G,y)]$ with converged Gauss--Hermite quadrature. Exact
lognormal moments are used for the linear arms to avoid subtracting large noisy
shared terms near cancellation. The implemented positive-bracket guard clips at
most $5\times10^{-6}$ of sampled draws; the closed forms omit that negligible
truncation. Equation~\eqref{eq:exchangeable-dprime} then evaluates equal-weight
linear $d'^2(P)$. We verified the terminal equations against the codebase
branch helper and recorded diagonal and off-diagonal covariance separately.
Here $P$ adds one E/I observation pair per unit. The endpoint is constrained
linear discriminability; it should not be interpreted as mutual information,
a Bayes-information bound, total recorded-population information, or
information per synapse.

The separate strict negative control removes the auxiliary gain measurement,
caps the active numerator contacts at $K_{\rm act}\le256$, and adds one
global differential component,
$X_b=\mu+y\Delta/2+\Delta z+\bar\varepsilon_b$, with
$z\sim\mathcal N(0,\nu)$ and
$\bar\varepsilon_b\sim\mathcal N(0,\sigma_{\rm pr}^2/k)$.
Here $B_{\rm dend}=8$ is the branch count,
$k=K_{\rm act}/B_{\rm dend}$ is the number of active numerator contacts per
branch, and $\sigma_{\rm pr}^2$ is the per-contact private-noise variance, so
the branch average $\bar\varepsilon_b$ has variance $\sigma_{\rm pr}^2/k$. Both readouts are then the same statistic and
\begin{equation}
 d^{\prime2}=\frac{A}{1+\nu A},
 \qquad A=\frac{K_{\rm act}\Delta^2}{\sigma_{\rm pr}^2}
 =\frac{B_{\rm dend}k\Delta^2}{\sigma_{\rm pr}^2},
 \label{eq:identified-differential-ceiling}
\end{equation}
with exact ceiling $1/\nu$. An orthogonal low-rank control leaves the signal-axis
value $A$ unchanged. Closed forms were compared with 16 Monte Carlo repeats of
12,000 balanced trials per population size and condition. Because both readouts
in this negative control are identical, no deterministic transform can raise
the inaccessible ceiling. This control is separate from
Eq.~\eqref{eq:code-faithful-population-generator}; neither experiment
establishes that a cortical circuit observes the modeled nuisance sensor.

\subsubsection{Exact-inventory hierarchical gain construction}
\label{methods:exact-inventory-hierarchy}

The fixed-inventory construction in Fig.~\ref{fig:criterion-morphology}A--C
uses the exact passive shunting recursion in Eq.~\eqref{eq:node-voltage}, the
tangent control described below, and the implemented raw-additive recursion.
Across-output additive normalization is disabled:
\begin{equation}
 V_n^{\add,\mathrm{fix}}
 =E_n-I_n+g\sum_{b\in\mathcal D_n}V_b.
 \label{eq:exact-inventory-fixed-additive}
\end{equation}
For an equiprobable label $y\in\{-1,+1\}$ with
$\Pr(y=+1)=\Pr(y=-1)=1/2$, draw four fine gains $f_j$, two coarse gains
$c_k$, and one global gain $h$ independently as mean-one lognormal variables
with common log-SD $s_g$. With $q(1)=q(2)=1$ and
$q(3)=q(4)=2$, the four signal-node excitatory observations are
\begin{equation}
 E_j=\max\!\left\{\epsilon_{\rm loc},
 (\mu+y\Delta)f_jc_{q(j)}h+\xi_j\right\},
 \qquad \xi_j\sim\mathcal N(0,\sigma_{\rm pr}^2).
 \label{eq:hierarchical-inventory-signal}
\end{equation}
Their local inhibitory observations are $I_j=s f_j$. Two sensor-only nodes
have $(E,I)=(\epsilon_{\rm loc},sc_1)$ and
$(\epsilon_{\rm loc},sc_2)$, and two further nodes carry redundant observations
$(\epsilon_{\rm loc},sh)$. Thus every morphology receives
the identical multiset of eight local E/I node observations; the duplicated
global observation makes eight physical sensors but seven independent gain
values. The parameters are $\mu=1$, $\Delta=0.12$,
$\sigma_{\rm pr}=0.20$, sensor conductance $s=12$, axial conductance $g=0.4$,
local positivity floor $\epsilon_{\rm loc}=10^{-6}$, and recursion stabilizer
$\epsilon_{\rm code}=10^{-8}$.

The parallel $[8]$ tree attaches all eight nodes directly to the child-only
soma. The $[2,3]$ tree places the four signal/fine and two coarse-sensor nodes
distally beneath two global-sensor parents. The serial $[2,1,2]$ tree places
two signal/fine children below each coarse-sensor node and each coarse node
below a global-sensor node. Each tree has eight non-somatic nodes and eight
axial edges. Somatic child-only normalization is a positive constant rescaling
of the summed root voltages and is absorbed by the scalar decoder. At the fixed
$g=0.4$, Eq.~\eqref{eq:exact-inventory-fixed-additive} gives coefficients
$(1,1,1)$ for signal/fine, coarse, and global nodes in $[8]$;
$(0.4,0.4,1)$ in $[2,3]$; and $(0.16,0.4,1)$ in $[2,1,2]$. Hence the fixed
raw additive trees are linear but not topology-invariant. Their scores become
identical under these identity local gains only at $g=1$; more generally, with
the same observation inventory, independently optimized local weights can
absorb nonzero path gains if no magnitude, sharing, or regularization constraint
binds.

The aligned construction uses the coarse sensor above its matching two signal
nodes. The shuffled control exchanges the two coarse sensors while preserving
the complete raw observation multiset. The gain-sensor-noise control multiplies the
coarse and global sensor observations by independent mean-one lognormal gains
with log-SD $0.55$. It therefore changes sensor fidelity without changing the
latent gains in Eq.~\eqref{eq:hierarchical-inventory-signal}.

At each $s_g\in\{0,0.1,0.2,\ldots,1.0\}$, we generate independent
training and test samples of size 12,000 for each seed. The fixed-additive and
shunting classifiers are scalar logistic readouts of their exact tree outputs,
evaluated separately for all three morphologies. A second additive comparator
is a sign-unconstrained logistic linear classifier of all eight raw E and eight
raw I observations. This raw-linear feature class contains the passive
positive-E/I additive-tree family, but the finite training fit is not itself a
function-class optimum or a particular tree. Because its features do not
depend on topology, it is fit once rather than repeated under tree labels. We
use it as the primary additive capacity comparator; the fixed raw cascade is
retained to show the behavior of the implemented additive tree.
Features are standardized using training-split moments, and all scalar and
linear decoders use L2-regularized logistic fits with $C=10$
separately at each noise level. These curves test
performance after the decoder is refit separately at every noise level.
At $s_g=0.5$, varying $C$ from $0.1$ to $10^6$ changes mean accuracy by
only $0.024$ percentage points, and the $C=10$ accuracy is identical to the
unpenalized maximum-likelihood fit.

The morphology-specific tangent-additive comparison uses disjoint seeds
200--207 and repeats the complete 11-point
noise grid, three regimes, three morphologies, and 12,000/12,000 train/test
sizes. For each seed and morphology, the aligned, sensor-noise-free
$s_g=0$ training inventory is passed through the shunting tree once. At
each tree depth, $N$ and $T$ are pooled over training samples and all nodes in
that depth to set one layer-global scalar anchor $(N_0,T_0)$, matching the
implemented \texttt{tangent\_matched} control. The resulting affine node rule
from Eq.~\eqref{eq:tangent-matched-node} is then frozen across every noise level
and regime; no test observation sets an anchor. Each tangent tree, shunting
tree, and raw-linear comparator receives its own train-fitted logistic decoder.
This analysis tests whether the result survives local value, Jacobian, and
scale matching.

The calibration-to-test coordinate bridge in
Fig.~\ref{fig:criterion-morphology}E uses a third disjoint seed block, 300--307.
For each outer seed, a 12,000-example clean training split fixes the
morphology-specific tangent anchors and a sign-unconstrained linear template.
A separate 12,000-example calibration split estimates clean, gain-only, and
sensor-corruption-only $d'^2$ values at gain log-SD $0.2,0.5,0.8$. Independent
12,000-example test splits then combine gain and sensor corruption for all
three morphologies, aligned/shuffled supports, and tangent/linear comparators.
No combined-test observation fits a score, anchor, coordinate, threshold, or
parameter. The resulting $3\times2\times3\times2=36$ condition means test the sign of
Eq.~\eqref{eq:matched-input-margin}. The linear template is fixed from the
clean training split rather than refit at every nuisance condition.
The bridge reports $G_{\sh}/G_a$ only as a normalized $d'^2$-degradation ratio,
not as the variance-defined $M_g$: gain can also change the class-mean signal
gap, so those two quantities need not coincide.

The operating-sensitivity analysis evaluated the grid
$s_g\in\{0.5,0.8\}$, $s\in\{4,8,12,20\}$, and
$g\in\{0.2,0.4,0.8\}$. All 24 conditions used paired held-out seeds
100--107 with 12,000 train and 12,000 test samples. One topology-invariant
sign-unconstrained additive fit per seed and gain level was paired with every
shunting condition; positive sensor rescaling is absorbed by training-feature
standardization. Cell intervals use paired-seed
$t_{0.975,7}\operatorname{SEM}$.

\paragraph{Passive nonlinear-capacity and structural-reduction controls.}
The nonlinear-capacity control reuses the exact eight-excitatory,
eight-inhibitory observation vector of the hierarchical construction. It
compares the deep passive shunting tree plus a two-parameter logistic decoder
with a fitted linear model, a quadratic logistic model, validation-selected
multilayer perceptrons, and a sign-constrained E/I network whose hidden width
matches the selected multilayer perceptron. Training-set sizes
are nested prefixes of a common draw; within each size, 75\% fits candidate
models and 25\% selects architecture and regularization by log loss. An
independent 12,000-example test set is opened only after selection. The primary
aligned condition selects each nonlinear family independently at each sample
size using seeds 500--507; secondary gain, sensor-corruption, and
oracle-support-shuffle conditions reuse those selected settings. Accuracy and
log loss are reported separately rather than using one endpoint to select the
other. For multi-initialization arms, configurations are selected by mean
validation log loss and their mean test endpoints are reported; no
initialization is selected individually.

The structural-reduction control holds the passive inputs, divisor sensors,
two fitted decoder parameters, and seeds fixed while replacing the tree by
serial, serial-plus-axial, flat, or level-permuted compositions. Equivalence of
serial and tree arms is tested with a prespecified $\pm1$ percentage-point
margin over seeds 600--607. Neither control applies a DendriNet post-voltage
activation: the external nonlinear predictors test raw-observation capacity,
not activated dendritic trees.

Parameters were selected using pilot seeds 0--1. The held-out test used paired
seeds 100--107; trials and noise-grid points were not treated as
replicates. Curves show the seed mean and
$t_{0.975,7}\operatorname{SEM}$. In addition to accuracy, AUC, and categorical
log loss, we report a gain-sensitivity diagnostic: within each class we
demean the scalar decision score and the seven true log gains, then report the
$R^2$ of their multivariate linear regression. A decrease in this $R^2$ is
interpreted as useful nuisance rejection only when accuracy or AUC is preserved;
signal collapse can also reduce $R^2$.

\subsection{Model and implementation}
\label{methods:grp-model}

\subsubsection{Dendritic network implementation}
\label{methods:dendritic-implementation}

Algorithm~\ref{alg:dendrinet-forward} summarizes the implemented
distal-to-soma forward sweep.

\begin{algorithm}[t]
\caption{DendriNet forward pass: distal-to-soma branch-voltage sweep.}
\label{alg:dendrinet-forward}
\begin{algorithmic}[1]
\REQUIRE Excitatory and inhibitory presynaptic pools; branch factors $(b_1,\ldots,b_L)$, $n_s$ somatic units; per-branch contact counts $N_E,N_I$; positive child couplings $g_c$; rule $\in\{\mathrm{add},\mathrm{sh}\}$; flags \texttt{somatic\_synapses}, \texttt{reactivation}; stability $\epsilon$
\ENSURE Somatic output activity (excitatory population)
\STATE The $n_s\prod_{\ell=1}^{L}b_\ell$ level-$L$ leaves have no children;
level $0$ contains the $n_s$ somatic units
\FOR[{distal $\to$ soma}]{level $\ell = L$ \textbf{down to} $0$}
  \FOR{each branch/unit $b$ in level $\ell$}
    \IF{$\ell>0$ \textbf{or} \texttt{somatic\_synapses}}
      \STATE $E_b \gets$ Top-$N_E$ pooled excitatory drive onto $b$
      \STATE $I_b \gets$ Top-$N_I$ pooled inhibitory drive onto $b$
    \ELSE
      \STATE $E_b \gets 0;\ I_b \gets 0$\COMMENT{level-0 soma aggregates children only}
    \ENDIF
    \STATE Gather child voltages $\{V_c\}$ of $b$ (empty at distal leaves)
    \IF{rule $=\mathrm{sh}$}
      \STATE $V_b \gets \dfrac{E_b+\sum_c g_c V_c}{1+E_b+I_b+\sum_c g_c+\epsilon}$
    \ELSE
      \STATE $V_b \gets E_b - I_b + \sum_c g_c V_c$
    \ENDIF
    \IF{\texttt{reactivation}}
      \STATE $V_b \gets \phi(V_b)=\tfrac{1}{2}\bigl(1+\tanh(\kappa_\phi(V_b-b))\bigr)$
    \ENDIF
  \ENDFOR
\ENDFOR
\STATE \RETURN somatic branch voltages $\{V_b\}_{\text{layer }0}$ ($n_s$ units)
\end{algorithmic}
\end{algorithm}

The trainable framework used in the learned-network experiments is a
feed-forward E/I population network built from three nested modules. The outer
\emph{E/I network} applies an input transfer layer, then stacks one or more
E/I layers. Each E/I layer contains an excitatory population and, when the
configuration calls for learned inhibitory cells, an inhibitory population. Each
cell population is implemented as a \emph{DendriNet}: a balanced tree of branch
layers ending at a somatic output. Each branch layer contains sparse excitatory
and inhibitory contact pools, optional child-branch inputs, a passive
additive-or-shunting voltage rule, and an optional post-voltage activation.
The learned-network results in Fig.~\ref{fig:ei-circuit} and the related
appendix controls are neural-network simulations that generate activity through
explicit dendritic morphology and either additive or shunting branch integration;
non-dendritic multilayer-perceptron (MLP) baselines are treated separately as
capacity controls. The
trained V1 analysis (Fig.~\ref{fig:v1-main}), expanded by
the resource-matched topology analysis in Appendix
Supplementary Fig.~\ref{fig:v1-exact-resource-depth-noise-supp}, uses this same
population-network/DendriNet codepath with matched additive and shunting modules;
the models are fitted by backpropagation on nonnegative V1 responses.
Supplementary Fig.~\ref{fig:v1-fixed-template-supp} separately applies fixed scalar
readouts directly to recorded responses.

\paragraph{Input transfer and pathway modes.}
The transfer layer returns an excitatory pathway and an inhibitory pathway. With
\texttt{independent\_pathways=false}, both pathways receive the same transferred
input; with \texttt{independent\_pathways=true}, the input vector is split into
separate E- and I-driving coordinates. The configuration flag
\texttt{input\_mode} determines how that inhibitory pathway is used. In
\texttt{input\_mode=0}, inhibitory cells can be built inside the E/I layer and
their activity supplies the inhibitory drive to excitatory cells. In
\texttt{input\_mode=1}, the transferred inhibitory pathway is used directly as
the inhibitory presynaptic pool, and no separate inhibitory cell population is
built for that layer. The fixed-template V1 diagnostic does not instantiate the
trainable DendriNet and therefore does not use an \texttt{input\_mode}; it
constructs fixed E and I pools directly from recorded neurons.

For strict conductance runs, positive-conductance shunting networks require the
outputs of the transfer layer to be nonnegative. The current
\texttt{PopulationNetwork} path routes transfer outputs directly and does not
enforce this condition at the shunting boundary; strict-regime experiments must
therefore establish it from the dataset and transfer configuration. If raw data
are signed or normalized, the configuration must use a nonnegative transfer
output activation such as ReLU, softplus, or sigmoid, or the experiment is
treated as an ML extension rather than a strict conductance-level model.

\paragraph{Tree topology and branch layers.}
A DendriNet with $n_s$ somatic units and branch factors
$(b_1,\ldots,b_L)$ represents a balanced tree in which each somatic unit receives
$b_1$ child branches, each of those receives $b_2$ child branches, and so on.
The distal branch count is therefore $n_s\prod_{\ell=1}^L b_\ell$. Computation
proceeds from distal branches to soma. At a layer where $b$ child branches
converge onto one parent branch, the child voltages are grouped into blocks of
size $b$ and passed through a positive block-linear coupling. These block
weights are the effective axial conductances in the feed-forward tree
abstraction. When \texttt{somatic\_synapses=false}, the final somatic layer
aggregates child branches without adding a new set of direct E/I synapses at the
soma; when it is true, the soma also receives direct sparse E/I contacts.
Every non-somatic level receives a fresh projection from the original global E/I
input streams unless its per-level feed-forward contact count is explicitly set
to zero. Consequently, increasing depth in the default configuration changes
both propagation topology and repeated access to the raw input.

For one somatic unit define
\begin{equation}
  B_{\rm dend}:=\sum_{\ell=1}^{L}\prod_{r=1}^{\ell}b_r,
  \label{eq:branch-resource-count}
\end{equation}
the number of non-somatic branches. For branch $b$, the exact realized counts
are the forward-mask sums
$N_{E,b}^{\rm real}=\sum_jM^E_{bj}\le\min(N_E,D_{E,b})$ and analogously for I;
$N_E,N_I$ are configured caps. Equality holds for the unmasked experiments
reported here and is verified from the forward mask in the fixed-mask controls, but
need not hold under an arbitrary row-wise connection mask. With $P_E$
somatic units, \texttt{somatic\_synapses=false}, uniform unmasked candidate
dimensions, and uniform per-branch caps (giving common counts
$N_E^{\rm real},N_I^{\rm real}$), the realized feed-forward contact budget is
$P_EB_{\rm dend}(N_E^{\rm real}+N_I^{\rm real})$ (or the corresponding sum
over levels when dimensions differ), the number of child couplings is
$P_EB_{\rm dend}$, and a learned shifted-tanh gate contributes
$2P_E(B_{\rm dend}+1)$ parameters. These quantities, not $N_E$ and $N_I$
alone, define a morphology comparison's active resource budget.

\paragraph{Sparse synaptic contacts.}
Feed-forward excitatory and inhibitory synapses are implemented with a Top-$K$
linear layer. For each output branch, the layer keeps the $K$ largest raw
synaptic scores and masks the rest on every forward pass. The active synaptic
weights are then mapped through the configured transform. The strict positive
E/I experiments use a positive transform, usually softplus, so the active
conductances are nonnegative while the underlying trainable parameters remain
unconstrained. All reported experiments use the deterministic
\texttt{standard} Top-$K$ layer. We use $N_E$ and $N_I$ for the configured numbers of excitatory and
inhibitory feed-forward contacts per branch in trained-network sweeps; the
corresponding implementation fields remain \texttt{k\_e} and \texttt{k\_i}. For the
V1 diagnostic, $N_+$ and $N_-$ instead denote the numbers of training-assigned
positive- and negative-evidence inputs; these signs label downstream routing
rather than biological cell identity.
If a requested $N_E$ or $N_I$ exceeds its available candidate dimension, the
implementation silently caps it at that dimension, so source tables must record
both requested and realized contacts. Standard Top-$K$ retains a dense trainable
score matrix and activates only $K$ entries per branch; changing $K$ changes
active connectivity but does not generally change the dense raw score count.
When branch diagnostics are enabled, the implementation records the exact
forward-pass decomposition $(E,I,C,G,N,T,V)$ and obtains
\texttt{realized\_K} from the mask actually used by that forward call (or from
the fixed indexed connectivity), rather than inferring it from a configured
cap.

\paragraph{Optimized and fixed quantities.}
During training, the optimized parameters are the raw Top-$K$ synaptic scores,
the positive child-branch coupling weights, the optional activation parameters,
and the downstream decoder weights. The branch factors, number of somatic units,
per-branch E/I contact counts, integration rule, pathway mode, and presence or
absence of post-voltage activation are fixed by the experiment configuration. Paired runs
share tree shape, requested contact counts, pathway construction, data split,
and decoder family. When each model uses its own selected setup, the effective
initialization can differ: the passive rule, branch/axial scale, and gate
initialization change together. The resulting comparison is between complete
trained models; same-policy and matched-effective-operating-point controls ask
the narrower rule-isolation question.

\paragraph{Branch voltage rule.}
For a branch receiving pooled excitatory drive $E_b$, pooled inhibitory drive
$I_b$, child-branch voltages $V_c$, and positive child couplings $g_c$, the
shunting implementation computes
\begin{equation}
  V_b^{\sh}
  =
  \frac{E_b+\sum_c g_cV_c}
       {1+E_b+I_b+\sum_c g_c+\epsilon},
  \label{eq:implementation-shunting}
\end{equation}
with a small $\epsilon$ only for numerical stability. The additive counterpart
uses the same sparse contacts and child couplings but removes the denominator:
\begin{equation}
  V_b^{\add}
  =
  E_b-I_b+\sum_c g_cV_c.
  \label{eq:implementation-additive}
\end{equation}
Writing $\ell=1+\epsilon$ and dividing $E_b,I_b,g_c$ by $\ell$ maps
Eq.~\eqref{eq:implementation-shunting} exactly to the unit-leak equations in
the theory. Consequently, all exact identities stated there apply to the code
after this change of conductance units; no $10^{-8}$ term is discarded.
Thus paired additive and shunting branches use the same requested contact counts,
tree topology, and child-branch wiring pattern, while their passive transfer and
learned coupling values can differ. The
reported learned-network experiments use this raw current-mode additive rule.
The analytic sections use the generalized passive additive tree of
Proposition~\ref{prop:additive-flattening}, which includes this raw rule and the
leak-normalized axial cascade as special cases; both flatten exactly, but their
fixed-budget path gains and training dynamics need not be identical. The
\texttt{dendritic\_normalized\_additive} option in the code is a separate
per-sample, across-output post-voltage z-score scale control, not the
leak-normalized axial recursion, and is disabled in the reported main experiments. The code also
contains optional recurrent-current paths for recurrent experiments; these are
disabled in the feed-forward experiments reported in this paper.

\paragraph{Post-voltage activation and outputs.}
The implementation calls this option \texttt{reactivation}; in reader-facing
text we call it post-voltage activation. After the passive voltage, a branch can
apply the learned shifted tanh
\[
  \phi(V)=\frac{1+\tanh(\kappa_\phi(V-b))}{2},
\]
with positive learned slope $\kappa_\phi$ and learned midpoint $b$. This keeps propagated
branch activity in a nonnegative rate-like range. The same functional family is
applied after either passive branch rule, but its initialization can be
mechanism-dependent: the branch first computes $V_b^{\add}$ or
$V_b^{\sh}$ and then, when enabled, propagates $\phi(V_b)$ to the next layer. The
activation can be disabled to test passive-voltage-only trees. In an E/I layer, inhibitory cells are
computed first when they exist, and excitatory cells then receive their
inhibitory drive. The network output used by the classifier is the final
excitatory population activity; the classifier head is an ordinary MLP decoder
in the learned-context experiments.

\begin{table}[htbp]
\centering
\footnotesize
\setlength{\tabcolsep}{5pt}
\renewcommand{\arraystretch}{1.06}
\caption{Post-voltage activation settings. ``Passive'' means that no activation
follows the branch voltage. Within each comparison, additive and shunting arms
use the same activation family.}
\label{tab:activation-settings}
\begin{tabularx}{\linewidth}{@{}>{\raggedright\arraybackslash}p{0.25\linewidth}>{\raggedright\arraybackslash}X@{}}
\toprule
Setting & Analyses \\
\midrule
Passive branch voltage &
Main Figs.~\ref{fig:framework-main}C--F, \ref{fig:regime-map}--\ref{fig:criterion-morphology},
and \ref{fig:ei-circuit}D; Supplementary Figs.~\ref{fig:local-equivalence-supp},
\ref{fig:nongaussian-supp}--\ref{fig:capacity-reduction-supp},
\ref{fig:locality-ei-fixed-supp}, and \ref{fig:topk-comparator-ladder-supp}D--F. \\
Learned shifted tanh after every branch, including the soma &
Main Figs.~\ref{fig:ei-circuit}B--C,E--F and \ref{fig:v1-main}B--E;
Supplementary Figs.~\ref{fig:resource-matched-morph-supp},
\ref{fig:contextual-factorial-supp}, and
\ref{fig:v1-exact-resource-depth-noise-supp}. \\
Training-calibrated shifted tanh, not backpropagated &
Top-$K$ controls in Fig.~\ref{fig:ei-circuit}C and Supplementary
Fig.~\ref{fig:topk-comparator-ladder-supp}A--C. \\
Matched terminal activation only &
Supplementary Fig.~\ref{fig:activation-sensitivity-supp}; not internal tree activation. \\
No DendriNet activation &
Main Figs.~\ref{fig:cifar-normalization-main} and \ref{fig:v1-main}F;
Supplementary Figs.~\ref{fig:v1-fixed-template-supp} and
\ref{fig:v1-running-recovery-supp}. \\
\bottomrule
\end{tabularx}
\end{table}

\paragraph{Setup selection.}
Initialization is not a variable in the theory. For the complete trained-model
comparisons, each mechanism receives the same validation information and search
budget, freezes its selected setup, and is then evaluated once on the test set.
The synthetic trained comparisons use additive occupancy-quantile and shunting
analytical initialization; both V1 mechanisms selected analytical initialization.
Supplementary Figs.~\ref{fig:contextual-factorial-supp} and
\ref{fig:topk-comparator-ladder-supp} contain the shared-policy and fixed-state
implementation controls.

\subsection{Experimental protocols}
\label{methods:grp-experiments}

\subsubsection{Matched comparator controls}
\label{methods:controls}

We use a nested comparator family to separate effects of finite-amplitude
division from those of coefficient fitting, conductance normalization, and
optimization policy.

\emph{Free additive comparator.}
At a terminal node fit $A_{\rho,b}=E-\rho I-b$ with $\rho,b$ selected on
training data. By Proposition~\ref{prop:scalar-threshold}, this is the relevant
fitted additive baseline for the binary accuracy of a scalar shunt.

\emph{Tangent-matched additive control.}
At a fixed operating point calibrated on training data, compare the shunting node with
$A_{\rm tan}$ from Eq.~\eqref{eq:tangent-matched-node}. This isolates the
sample-dependent total-drive gate in Eq.~\eqref{eq:node-gated-contrast} from a
mere change in local E/I coefficients.

 \emph{Conductance-normalized additive control.}
To isolate denominator normalization from numerator sign structure, use
\begin{equation}
 V_b^{\rm norm-add}
 =\frac{E_b-I_b+\sum_cg_cV_c}
 {1+E_b+I_b+\sum_cg_c+\epsilon}.
 \label{eq:true-normalized-additive}
\end{equation}
This differs from the per-sample, across-output
\texttt{use\_additive\_normalization} z-score. The numerator coefficient is
fixed at one; the terminal free-$\rho$ comparator above instead fits the
decision rule. These node controls correspond to the
\texttt{additive\_mode}: \texttt{raw},
\texttt{conductance\_normalized}, and \texttt{tangent\_matched}. The tangent
mode uses one layer-global scalar anchor $(N_0,T_0)$, pooled over all samples
and output coordinates in its explicit calibration batch, and evaluates the
affine tangent of the implemented $N/(1+T+\epsilon)$ rule. It is never
calibrated on an evaluation batch. The trained comparisons in the main text use
\texttt{raw}; the additional modes appear in the Top-$K$ comparator analysis.

\emph{Conductance-based subtractive control.}
$V^{\mathrm{cond}}=(E+\chi_{\mathrm{rev}} I)/(1+E+I)$ with
$\chi_{\mathrm{rev}}\le 0$ shares the shunting normalization pool and therefore
the same sensitivity to denominator variability.

\emph{Frozen-divisor control.}
Replace the random denominator $1+E+I$ by its operating-point value $D_0$,
giving $V^{\mathrm{frozen}}=E/D_0$.
This isolates divisor variability from the numerator nonlinearity in the
terminal scalar analysis. The Top-$K$ implementation analogously divides
the complete raw subtractive branch voltage $E-I+\sum_c g_cV_c$ by the paired
shunting reference's fixed training-only mean denominator; it does not replace
that numerator by $E$.

For trained networks we report three comparisons. The first uses the same
initialization policy in both arms. The second matches the effective operating
point using training-only output quantiles, saturation, and mean gate derivative.
The third lets each mechanism use the setup with the highest validation
accuracy under the same tuning budget. Setup selection uses only validation
data, averages over tuning seeds, and is frozen before test evaluation. This is
the whole-model comparison used in the main text; the first two isolate the
rule and quantify setup sensitivity. The supplement reports all three:
same-policy contrasts, a Top-$K$ activation-policy control, and the
equal-budget model-specific comparison. In the operating-point control, paired models use identical mechanism-neutral
initialization. Without labels or
evaluation data, each activation maps pooled training-voltage q10/q90 to
outputs 0.1/0.9; its parameters are excluded from backpropagation. The
maintained schedule recalibrates from eight training batches before each epoch,
whereas the initialization-only schedule quantifies drift. We record output
quantiles, saturation, raw and scale-normalized gate sensitivity,
pre-gate voltage quantiles, realized $K$, and parameter/mask hashes
(Supplementary Fig.~\ref{fig:topk-comparator-ladder-supp}A--C). This matches
post-activation occupancy; raw additive and shunting voltages, E/I currents, and
class-conditional branch distributions may still differ.
Identical raw initializer numbers are not required because additive and shunting
voltages have different scales. The same policy is therefore a robustness
comparison, not automatically an operating-point match. A rule-swap comparison
requires matched operating points; a matched whole-model comparison instead
requires equal tuning information and budgets and an untouched test set.

\paragraph{Top-$K$ comparator ladder.}
The comparator analysis uses the implemented passive DendriNet with branch factors
$[2,2]$, 64 somatic units, two directly innervated dendritic levels, exactly
$N_E/N_I=16/12$ contacts per synaptic branch, child-only somatic pooling, no
post-voltage activation, and an MLP decoder of widths 64 and 32. Seven arms are crossed
with seeds 12640--12647: shunting; raw additive; conductance-normalized
additive; the implemented layer-global tangent; a layer-wise positive affine
response match; the frozen-divisor control; and the per-sample across-output
z-score. All arms begin from identical mechanism-neutral
parameters, buffers, Top-$K$ masks, data, and decoder values. The tangent,
response-affine, and frozen-divisor constants use the same fixed-order,
label-free prefix of eight training batches and are fixed thereafter; no test
input or label is used for calibration or model selection. Primary estimands
are paired out-of-distribution (OOD) accuracy and log loss for each of the six
comparators versus shunting,
with Holm correction applied separately to the six accuracy and six log-loss tests.
Bootstrap intervals, AUC, the $<80\%$ failure count, and a training-fitted
free-$\rho$ leaf-population readout are diagnostics. The latter averages distal
$E$ and $I$ drives and bypasses proximal branches, soma, and the trained decoder,
so it is not a seventh internal comparator. We also verify
$V>\tau\Longleftrightarrow N-\tau T>\tau(1+\epsilon)$ exactly at $\tau=0.5$.

\paragraph{DendriNet denominator-only intervention.}
This causal control uses the implemented DendriNet with one passive $[8]$ level,
32 somatic units, $N_E/N_I=8/8$ contacts per distal branch, no somatic synapses or
post-voltage activation, and a linear population decoder. An external positive conductance
is passed directly to
\texttt{DendriticBranchLayer.voltage\_from\_currents}: it appears in the
shunting denominator but is absent from the raw-additive graph. Positive
lognormal and gamma inputs are each crossed with eight paired seeds
(13000--13007) and four interventions: zero, constant conductance 3,
independent conductance variability, and class-correlated adverse conductance.
The three nonzero conductance conditions have exactly matched pooled means
within every split. All pairs share E/I
samples, minibatch order, optimizer, decoder initialization, initial parameters,
and initial masks. The primary estimand trains both mechanisms with zero external conductance, freezes the network
and decoder, and cross-evaluates all four held-out interventions. The four
prespecified metrics are balanced error, AUC, balanced log loss, and
fixed-decoder margin $\dprime^2$; mechanism-by-intervention and within-shunting
contrast types were also prespecified. We apply one exact sign-flip Holm
correction to 60 family-specific, cross-family-average, and load-decomposition
tests. Cross-family averages were not a design factor, so their intervals are
descriptive. A secondary analysis retrains within each regime to measure
adaptation. We verified raw-additive output invariance, invariance of the
directly intervened distal numerator, initial parameter and mask identity, and
shared training order. The 16 array tasks produced 80 unique fits, all of which
passed the stated inclusion checks.

\subsubsection{Learned dendritic networks under latent gain}
\label{methods:ei-circuit}

We train dendritic networks with nonnegative input rates and separate E/I
pathways to test the gain--load--alignment predictions under hierarchical latent
gain.

\emph{Architecture.}
The contextual experiments in Fig.~\ref{fig:ei-circuit}A--C use transfer mode~1.
There is no separate learned
inhibitory-cell population; the transferred inhibitory pathway is used directly
as the inhibitory presynaptic pool. Because
\texttt{independent\_pathways=false} in this contextual grid, the same
transferred feature vector is routed through separate direct E and I synaptic
pathways; these are not statistically independent input streams. The reported
network has one learned excitatory DendriNet population. The alignment grid
uses the same $[2,2]$ branch topology throughout and varies only the
architecture variables documented here:
excitatory dendritic units
$P_E\in\{64,128,256\}$, excitatory contacts per branch
$N_E\in\{16,32\}$, inhibitory contacts per branch $N_I\in\{4,12\}$, and
alignment $\alpha\in\{0,0.25,0.5,0.75,1\}$. Each condition has $12$ matched
seeds for additive and shunting, giving $1440$ model runs and $720$
additive--shunting pairs. The
resulting excitatory population is followed by an MLP decoder with hidden widths
$64$ and $32$.
Weights are transformed by softplus so that synaptic conductances are
nonnegative.
The key architectural variable is the dendritic integration rule in E~cells:
shunting ($V=E/(1{+}E{+}I)$) or additive ($V=E{-}I$) at the passive-voltage stage.
In the default learned-network setting this voltage is followed by a post-voltage
shifted parametric tanh rate,
\[
  \phi(V)=\frac{1+\tanh(\kappa_\phi(V-b))}{2}\in(0,1),
\]
with learned slope $\kappa_\phi$ and midpoint $b$; appendix activation sweeps compare this
against no post-voltage activation. This does not change the passive single-node
local relation to first order, but it can change whole-tree operating points and
optimization. In the mechanism-specific-recipe alignment grid the activation
policy was omitted, so the implementation defaulted to
\texttt{analytical} initialization for shunting and
\texttt{occupancy\_quantile} calibration for additive.
Figure~\ref{fig:ei-circuit}B--C reports the resulting held-out model comparison.
The supplement reports the
full same-policy factorial.
Analytical branch and axial initialization is itself mechanism aware, including
different initial gate slopes and child-coupling scales. The additive arm uses
the raw current-mode rule
\texttt{use\_additive\_normalization=false}. The normalized-additive codepath is
therefore not used in this contextual grid.

\emph{Contextual task.}
The contextual task summarized in Fig.~\ref{fig:ei-circuit}A--C is a bounded
nonnegative-rate stream benchmark. Each example contains two streams,
$x_A,x_B\in[0,1]^{64}$, and a two-dimensional context cue. The context selects
the behaviorally relevant stream, so the target is $y=y_A$ when $c=0$ and
$y=y_B$ when $c=1$. Class means differ by a random direction with signal
strength $0.35$; Gaussian rate noise has standard deviation $0.08$ on the
relevant stream and $0.18$ on the irrelevant stream, followed by clipping to
$[0,1]$. Training uses mild midpoint gains on both streams
($g_{\rm rel}\in[0.9,1.1]$, $g_{\rm irr}\in[0.8,1.2]$). At test time the
irrelevant stream receives a stronger midpoint gain $g_{\rm irr}=3$, and the
parameter $\alpha$ rotates the nuisance direction from orthogonal to
stimulus-aligned.

\emph{Training and evaluation protocol.}
Models are trained with categorical log loss using Adam for up to $250$
epochs, batch size $256$, early
stopping with patience $40$, best-state restoration, gradient clipping at $5$,
and the configured dendritic weight-decay/boosting stabilizers
(\texttt{weight\_decay\_rate=0.01}, \texttt{weight\_boosting=true}; optimizer
\texttt{weight\_decay=0}). The alignment grid uses $10^{-3}$ for
synapse-selection, block-linear, activation, and decoder parameters.
Figure~\ref{fig:ei-circuit}B reports absolute held-out accuracies, and panel C
reports the paired whole-model comparison and fixed-Top-$K$ controls.
Supplementary Fig.~\ref{fig:contextual-factorial-supp} reports setup sensitivity,
and Fig.~\ref{fig:ei-circuit}D reports the
single-branch check described below.
For the trained comparisons, both metrics are oriented so that positive values
favor shunting: $\Delta_{\rm acc}=\operatorname{Acc}_{\sh}-
\operatorname{Acc}_{\add}$ and $\Delta_{\rm loss}=\operatorname{Loss}_{\add}-
\operatorname{Loss}_{\sh}$.
Panel C uses paired-seed percentile-bootstrap intervals; panels B and E--F use
two-sided 95\% $t$ intervals over the stated paired seeds.

The gain--load benchmark in Fig.~\ref{fig:ei-circuit}E--F uses
the same population-network implementation and post-voltage activation after
each branch.
Inputs are ordered positive E/I streams from the \path{branch_local_gain_load}
generator in balanced-ratio mode: class identity shifts E and I in opposite
directions, each with magnitude $0.25$, approximately preserving total E+I
drive. The stream dimension is $64$, with $8$ gain groups,
independent noise standard deviation $0.18$, lognormal gain standard deviation
$2.2$ during both training and testing, alignment
$\alpha\in\{0,0.5,1\}$, I-stream background mean $\bar L\in\{0,1,3\}$,
background log-SD $0.35$, and signal support fraction $0.25$ or $0.5$. The background is added
to the inhibitory input stream, so it changes both the additive numerator and
the shunting denominator. Gain groups are coordinate groups accessible to every
unmasked branch, not imposed branch-local support. The network has
one learned excitatory DendriNet
population with $P_E=64$, direct excitatory and inhibitory input streams, no
learned inhibitory-cell population, branch factors $[2]$, $[2,2]$, $[2,2,2]$,
or $[3,3]$, $N_E=16$ excitatory and $N_I=12$ inhibitory contacts per branch, raw
additive or shunting branch integration, shifted parametric-tanh activation,
and a linear decoder. The experiment includes both initialization policies
(\texttt{analytical} and \texttt{occupancy\_quantile}). Both are
mechanism-aware and use mechanism-specific initialization equations. The main
Fig.~\ref{fig:ei-circuit}E--F comparison uses the prespecified complete-model
pair (additive/occupancy-quantile and shunting/analytical); same-named-policy
comparisons remain supplementary setup controls. Within each policy, every
alignment/background/signal-fraction/morphology condition has $8$ matched additive
and shunting seeds (base seed $9200$). The factorial contains $2304$ completed
model configurations and $1152$ within-policy
additive--shunting pairs. We report held-out test accuracy and categorical
log-loss gaps (the generated binary task is class-balanced). For
Fig.~\ref{fig:ei-circuit}E--F, the four fixed
morphology settings are first averaged within each paired seed; intervals and
inference therefore use $n=8$ paired seed blocks per plotted point, not the 32
raw morphology-by-seed pairs.

\emph{Layer implementation.}
The model uses the dendritic implementation described in
Methods~\ref{methods:dendritic-implementation}. In the activated runs, the
shifted parametric tanh above is applied after each passive branch voltage
before the branch output is passed forward; in the passive-voltage ablation,
this internal activation is removed.

\emph{Paired comparisons and estimand.}
The additive and shunting networks in a paired comparison share the same data
splits, random seed, branch topology, contact counts, decoder architecture,
optimizer settings, and early-stopping protocol. In mechanism-specific-recipe runs the passive
rule and effective initialization change together; same-policy and
matched-effective-operating-point comparisons are reported separately. When
flat or point controls are included,
the control is matched to the closest available parameter budget and input/output
interface, but removes a specific ingredient: intermediate hierarchy/depth for
the single-level star (``flat'') alias, or the E/I/dendritic decomposition for
point-neuron MLP controls. The star alias retains local terminal branches and
their divisors. These ablation
controls are treated as appendix diagnostics rather than the main Fig.~\ref{fig:ei-circuit}
evidence.

\paragraph{Locality and fixed-total-contact controls.}
\label{methods:locality-ei-controls}

Both controls use the codebase's \path{branch_local_gain_load} generator
with 64 ordered nonnegative E features and 64 ordered nonnegative I features,
8,000 requested examples per independently generated train, validation, and
test split before exact class balancing, independent input-noise standard deviation
0.18, matched train/test gain statistics, and positive feature-wise I-stream background with
log-space standard deviation 0.35 when its mean is 3. Models have one $[8]$
excitatory DendriNet population with $P_E=64$, no learned inhibitory-cell
population, and a common linear $64\!\to\!2$ decoder. Reactivation and adaptive
initialization are disabled; raw additive and shunting arms use the same
mechanism-neutral initializer, optimizer, stopping rule, and paired data seed.
Training uses Adam at $10^{-3}$ with zero Adam weight decay, batches of 256, at
most 140 epochs, early-stop patience 25, and gradient clipping at 5. Separately,
the codebase's active/inactive synaptic-maintenance rule uses configured rate
0.01, as in the V1 protocol below; this is not optimizer weight decay.

The locality factorial (Supplementary
Fig.~\ref{fig:locality-ei-fixed-supp}) divides each stream into eight exact
contiguous groups.
Every one of 512 distal branches receives eight E contacts from one group and
eight inhibitory-pathway contacts either from that group or from the group four
positions away; E support is unchanged. It crosses one global versus eight
independent local gains (log-SD $s_g=1.2$), I-stream background mean 0 versus
3, raw additive versus shunting, and eight paired seeds (128 runs). The class
signal is E-only, has amplitude 0.75, and occupies a seed-sampled 25\% of
features. The fixed-total allocation control instead uses a balanced-ratio
class signal with $\Delta E=\Delta I=0.25$, eight local gain groups (log-SD
$s_g=2.2$),
and $(N_E,N_I)\in\{(24,4),(20,8),(16,12),(12,16),(8,20),(4,24)\}$ at I-stream background means 0
and 3 (192 runs). Thus every allocation has exactly 28 contacts per branch,
14,336 active contacts, 65,536 dense synaptic scores, and 66,178 trainable
parameters. The common ``inhibitory share'' label in
Fig.~\ref{fig:framework-main}E--F denotes $N_I/(N_E+N_I)$. The continuous
control in Fig.~\ref{fig:framework-main}D instead uses inhibitory conductance
mass divided by total conductance mass. For panel E, we
identify the allocation with maximum
validation accuracy separately within each seed, mechanism, and background
condition; test accuracy does not enter the choice, and the two exact validation
ties are both plotted.

All 320 included rows matched the source version and design recorded in the
reproducibility manifest. We verified realized masks, resource counts, and
identical initial parameters and state in every reconstructed mechanism pair. We
report paired-seed means and 95\% $t$ intervals, paired $t$ tests, two-sided
exact sign-flip tests, and Holm correction within the stated contrast families.
The all-allocation linear mechanism-by-allocation-by-load interaction is
exploratory because it was specified during aggregation. Seeds are computational replicates;
$N_E:N_I$ is a contact allocation, not an E/I-cell ratio.

\emph{Why internal activation matters.}
The scalar-output corollary in the main text concerns a monotone activation
applied after a single optimized readout. Internal activation in a tree is a
different operation: it changes the operating point and slope of every
intermediate branch before downstream branches are evaluated. Removing it can
compress passive shunting voltages, reduce useful gradients, and move the
learned network into a load- or saturation-dominated regime. The sign reversal
observed in the mechanism-specific-recipe ablation confounds the
branch rule with its optimization policy. We therefore interpret it only as an
operating-regime diagnostic. Same-policy and matched-operating-point
comparisons, which mostly span zero, provide the rule comparison and are
consistent with the local theorem.

\emph{Single-branch gain--load simulation.}
To directly compare the gain--load calculation with the implementation, we also
simulate one branch using the exact passive voltage rules from
Methods~\ref{methods:dendritic-implementation}: raw additive
$V^{\add}=E-I$ and shunting $V^{\sh}=E/(1+E+I+L)$. This benchmark was designed
before looking at learned-network performance: it varies exactly the two
quantities that the criterion says should control the sign in this simplified
setting, shared-gain log-SD $s_g$ and positive-load log-SD $s_L$. Positive class-conditioned E/I drives are
perturbed by lognormal shared gain and positive lognormal load. For each grid
point we compute a local delta-method prediction from the branch-rule Jacobian at
the baseline operating point, then compare it to exact Monte Carlo $\dprime^2$
from the nonlinear branch equations. The operating point is
$(\bar E,\bar I,\bar L)=(3,1,0.5)$, with
$(\Delta E,\Delta I)=(0.6,-0.2)$ and $\sigma_{\rm pr}=0.25$.
The mean-preserving lognormal grids are
$s_g\in\{0,0.1,0.2,0.3,0.4,0.6,0.8\}$ and
$s_L\in\{0,0.1,0.2,0.35,0.5,0.75,1,1.25,1.5\}$, where both coordinates
are log-SDs. Each grid point uses eight seeds and 120,000 trials per seed. This
untrained simulation directly compares the local prediction with the nonlinear
voltage rule, independently of the contextual-stream task
(Algorithm~\ref{alg:branch-mc}).

\label{methods:branch-rule-check}
\begin{algorithm}[t]
\caption{Branch-rule Monte Carlo check: delta-method vs.\ exact $\dprime^2$ over the gain--load grid (Fig.~\ref{fig:ei-circuit}D).}
\label{alg:branch-mc}
\begin{algorithmic}[1]
\REQUIRE Positive midpoint drives $(\bar E,\bar I)$, class differences $(\Delta E,\Delta I)$, private-noise SD $\sigma_{\rm pr}$, and mean external denominator drive $\bar L$; mechanism $m\in\{\add,\sh\}$; log-SD grid $\{(s_g,s_L)\}$; samples $N$.
\ENSURE Predicted $\widehat{\dprime}^2_{m}(s_g,s_L)$ and Monte Carlo $\dprime^2_{m}(s_g,s_L)$ on the grid.
\STATE Branch voltage rules: $V^{\add}=E-I$, \quad $V^{\sh}=E/(1+E+I+L)$.
\FOR{each grid point $(s_g,s_L)$}
  \STATE \textbf{// Delta method at $(\bar E,\bar I,\bar L)$}
  \STATE $v_g\gets\exp(s_g^2)-1$, \quad $v_L\gets\bar L^2[\exp(s_L^2)-1]$ \COMMENT{exact variances for mean-one lognormal factors}
  \STATE Compute $\partial_E V^m,\partial_I V^m,\partial_L V^m$ and $\partial_gV^m=\bar E\partial_EV^m+\bar I\partial_IV^m$
  \STATE $\Delta V^m\gets\partial_EV^m\Delta E+\partial_IV^m\Delta I$
  \STATE $s_m^2\gets\sigma_{\rm pr}^2[(\partial_EV^m)^2+(\partial_IV^m)^2]+v_g(\partial_gV^m)^2+v_L(\partial_LV^m)^2$
  \STATE $\widehat{\dprime}^2_m\gets(\Delta V^m)^2/s_m^2$ \COMMENT{first order; omits gain--private-noise products}
  \STATE \textbf{// Exact Monte Carlo through the nonlinear branch}
  \FOR{$\theta\in\{0,1\}$, $\ n=1,\dots,N$}
    \STATE $g \gets \exp(s_g z_g-\tfrac12s_g^2)$, \quad $L \gets \bar L\exp(s_L z_L-\tfrac12s_L^2)$, \quad $z_g,z_L\sim\mathcal N(0,1)$
    \STATE $E\gets g\max\{\bar E+(2\theta-1)\Delta E/2+\epsilon_E,10^{-5}\}$, \quad $I\gets g\max\{\bar I+(2\theta-1)\Delta I/2+\epsilon_I,10^{-5}\}$, \quad $\epsilon_E,\epsilon_I\sim\mathcal N(0,\sigma_{\rm pr}^2)$
    \STATE $V^{(n)}_\theta \gets V^m(E,I,L)$ \COMMENT{exact rule, no linearization}
  \ENDFOR
  \STATE $\dprime^2_{m} \gets \big(\widehat\E[V_1]-\widehat\E[V_0]\big)^2 \big/ \big[\tfrac12\widehat\Var(V_0)+\tfrac12\widehat\Var(V_1)\big]$
\ENDFOR
\STATE \RETURN $\{\widehat{\dprime}^2_{m}\}$ and $\{\dprime^2_{m}\}$ over the $(s_g,s_L)$ grid \COMMENT{compare per mechanism}
\end{algorithmic}
\end{algorithm}

\subsubsection{Frozen-feature CIFAR-10 normalization bridge}
\label{methods:cifar-normalization}

We use CIFAR-10~\citep{Krizhevsky2009CIFAR} as a conventional machine-learning
bridge, not as biological evidence. An ImageNet-pretrained ResNet-18
~\citep{HeEtAl2016} is frozen before the CIFAR-10 split is opened. Images are
resized to $64\times64$ and passed through layer 4; adaptive $2\times2$ pooling
gives a nonnegative $512\times2\times2$ feature map. A fixed random split uses
45,000/5,000 training/validation images and the standard 10,000-image test set.

The additive arm applies a linear ten-class readout to the 2,048 features. The
group-divisive arm---the exact ANN analogue of the conductance rule in this
experiment---first partitions the 512 channels into eight contiguous groups
and maps each group $h_g$ to $h_g/(0.05+\overline h_g)$, where the mean spans
its 64 channels and four spatial positions. This mean-only local response
normalization is related to earlier divisive feature normalization
~\citep{JarrettEtAl2009,KrizhevskyEtAl2012AlexNet}; it is not GroupNorm
relabelled as a dendritic tree. Comparators are a
$2048\!\to\!256\!\to\!10$ ReLU MLP, LayerNorm, eight-group GroupNorm, and an
oracle that divides by the realized gain before applying the additive readout
~\citep{IoffeSzegedy2015,BaKirosHinton2016,WuHe2018GroupNorm}.
Within each of five fixed seeds, every trained arm receives the same
four-condition AdamW search over learning rates
$\{10^{-3},3\times10^{-3}\}$ and weight decays $\{0,10^{-4}\}$. Clean
validation log loss selects the complete model; training stops after at most 35
epochs with patience 6, and the test labels are not used for selection.

At test time, the same mean-one lognormal draws are paired across models.
Gain is global, shared within the eight channel groups aligned to the shunting
divisors, shared over the four spatial positions, or assigned by a fixed
balanced shuffle that mixes gain supports within each divisor. We test
$s_g\in\{0,0.4,0.8,1.2\}$. A separate aligned-gain control multiplies each
estimated shunting divisor by independent mean-one lognormal noise with
log-SD in $\{0,0.4,0.8,1.2\}$. The primary endpoint is categorical log loss;
accuracy and calibration gap are secondary. Computational seeds are paired
replicates. This experiment asks whether support-matched local normalization
protects a conventional frozen representation; it does not test morphology,
biological conductances, or benchmark superiority.

\subsubsection{Mouse V1 (Stringer et al.\ 2019) preprocessing and readout protocol}
\label{methods:v1-preprocessing}

\paragraph{Trained-DendriNet protocol.}
The V1 experiment uses raw nonnegative trial responses, eight uniform
direction classes, and no input normalization. For each computational seed,
candidate neurons are stably ranked by variance using the training partition
only; taking prefixes of that one ranking gives the nested upstream pools
$N_{\rm in}=100,250,500,1000,2000$. Held-out trials therefore affect neither
neuron identity nor order. The clean/noise fits use a seeded, non-stratified
90/10 train/test split followed by a non-stratified 85/15 split of the training
portion, for nominal train/validation/test fractions 76.5/13.5/10\%. The exact
counts for sessions GT1/GT2/GT3 are respectively 3275/578/429,
3418/604/447, and 3397/600/445. The direct-pathway population network uses
\texttt{input\_mode=1} with the identity input mapping, shared transferred input
(\texttt{independent\_pathways=false}), and no learned inhibitory population.
Each branch receives 40 excitatory and 10 inhibitory feedforward contacts, has
no recurrent or somatic contacts, and uses either the implemented additive or
shunting voltage rule followed by a trainable \texttt{param\_tanh}
post-voltage activation. Shallow $[8]$, intermediate $[2,3]$, and deep $[2,1,2]$ trees all
contain eight branch units per output. We verify equality of active
contacts, dense candidate slots, couplings, gates, and total trainable
parameters across morphology and mechanism within every $(N_{\rm in},P_E)$ cell.
Within a paired additive--shunting cell, the same base seed determines the data
partition, training-only neuron ranking, topology and masks, parameter
initialization, minibatch order, and evaluation/probe draws. Changing the base
seed changes all of these streams; the mechanism-specific analytical equations
remain the only initialization-policy difference.

\paragraph{Optimization recipe.}
An equal-budget validation sweep allowed mechanism-specific initializer and
learning-rate selection; both mechanisms selected analytical activation
initialization with adaptive operating-point preservation and branch/activation
learning rate $3\times10^{-4}$. Top-$K$ scores and the linear decoder use
$2\times10^{-3}$. The initial upstream and state fits used at most
200 epochs, patience 30, and best-state restoration. We then extended only the
training horizon while holding all scientific settings fixed. Of the upstream
fits, 454
required the 400-epoch, patience-80 boundary extension. For running state, 184,
18, and 86 fits were resolved at the 400/80, 800/160, and 1600/320 stages;
all 43 pairs extended from 800 to 1600 retained their gap sign, 82 of 86 fits
stopped before epoch 1600, and none was validation-best at the terminal epoch.
For the final analysis, every upstream cell uses the 400/80 ceiling and every
running-state cell uses the 1600/320 ceiling, with best-state restoration
unchanged. This uniform rule was fixed before the regenerated test endpoints
were examined. Training otherwise uses batch size 128, gradient clipping
at 5, and no automatic mixed precision (AMP). Top-$K$, block-linear, and
activation gradient-scaling
strategies are disabled. Adam's own weight decay is zero. Separately, the
codebase's \texttt{CustomWeightDecayOptimizer} applies model-specific decay to
active synaptic weights and the matched boost to inactive eligible weights with
configured rate 0.01 (the per-step multiplicative amount is learning rate
times that rate). These settings and validation budgets are identical across
the additive and shunting runs. The configuration switch is
\texttt{use\_shunting}, but the selected analytical policy uses
mechanism-specific voltage equations, branch scales, and gate states. The
comparison therefore uses the complete additive and shunting models selected
under the same validation budget; it is not a rule-swap experiment with all
parameters fixed.
As a same-policy sensitivity, occupancy initialization reduced validation
accuracy relative to analytical initialization by 0.186 percentage points for
additive and 0.187 percentage points for shunting (95\% paired intervals excluded
zero), supporting analytical initialization for both mechanisms as the primary
recipe. Occupancy initialization for both mechanisms left small clean
shunting--additive gaps (0.06--0.27 percentage points) but attenuated the
narrow-readout interaction; this sensitivity concerns clean accuracy only and does not validate the
decoder-log-loss or injected-noise endpoints.

\paragraph{Held-out decoder-bound score.}
For each clean test cell, $q(\theta\mid V)$ is the trained eight-way task
decoder evaluated on held-out trials. We compute
\[
 I_{\rm LB}=\widehat H_{\rm test}(\theta)
 +\frac{1}{n_{\rm test}}\sum_{t=1}^{n_{\rm test}}
   \log q(\theta_t\mid V_t),
\]
in nats and divide by $\log 2$ for bits. The empirical label entropy, test
examples, neuron selection, morphology, $N_{\rm in}$, $P_E$, and computational
seed are identical within every additive--shunting pair. Thus
$\Delta I_{\rm LB}=I_{\rm LB}^{\sh}-I_{\rm LB}^{\add}$ is exactly the paired
additive-minus-shunting held-out categorical-log-loss difference in different
units; it is not a lower bound
on $I(\theta;V_{\sh})-I(\theta;V_{\add})$. Pairing is performed before averaging
computational seeds within recording; the three recording-level summaries are
the biological units shown in Fig.~\ref{fig:v1-main}.

\paragraph{Held-out noise and natural-state tests.}
For the noise test, five deterministic perturbation draws are applied
to each fitted clean test model at relative scales
$0,0.1,0.25,0.5,1,2$. Private perturbations are independent across neurons and
scaled by training-partition neuronal SD. The shared control follows the first
principal component of within-class training residuals and is scaled by the SD
of its training scores. After either perturbation, responses are clamped at
zero to preserve the nonnegative direct-E/I input domain used by the fitted
models. No test labels determine either geometry. For the natural-state test,
the bottom and top running-speed quartiles are disjoint and the middle half of
trials is excluded. A seeded, non-stratified split reserves 20\% of the
low-running pool as an untouched same-state endpoint; the remaining 80\% is
split 85/15 for training/validation and training-only neuron ranking. Exact
train/validation/same-state-holdout/full-high-state counts for GT1/GT2/GT3 are
778/138/230/1071, 759/135/224/1118, and 1652/292/487/1111. Evaluation then
subsamples the low and high endpoints to identical per-class counts, yielding
230, 224, and 487 trials per endpoint in GT1/GT2/GT3. The same fitted model,
training-state neuron indices, seed, and target trials are reused unchanged in
each additive--shunting pair.
This sweep uses $N_{\rm in}=100,500,2000$, $P_E=8,16$, shallow and deep exact-resource
trees, and the validation-selected optimization recipe.
Writing $L_{m,s}$ for held-out categorical log loss of mechanism
$m\in\{\add,\sh\}$ at the untouched same-state endpoint $s=0$ or target-state
endpoint $s=1$, the prespecified primary interaction is
\begin{equation}
 \Delta_{\rm state}
 = (L_{\add,1}-L_{\sh,1})-(L_{\add,0}-L_{\sh,0})
 = (L_{\sh,0}-L_{\sh,1})-(L_{\add,0}-L_{\add,1}).
 \label{eq:v1-state-did}
\end{equation}
Positive values mean that the shunting log-loss advantage grows across
the low-to-high running-state transfer. Accuracy difference-in-differences is a
secondary endpoint with the same orientation.

Algorithm~\ref{alg:v1-readout} summarizes the fixed-template outer-fold
calculation.

\begin{algorithm}[t]
\caption{Core of one outer fold in the fixed-template mouse V1 evidence readout:
training-low selection and calibration, followed by held-out high-state
evaluation. Outer folds and resamples repeat this core with deterministic paired
seeds.}
\label{alg:v1-readout}
\begin{algorithmic}[1]
\REQUIRE Session responses $X\in\mathbb{R}^{T\times P}$ (deconvolved, nonnegative); outer-training low-state trials $\mathcal{T}_{\mathrm{lo}}^{\rm tr}$; outer-held-out high-state trials $\mathcal{T}_{\mathrm{hi}}^{\rm te}$; direction labels $\theta$; anchor class $c$; pool size $N$; fixed evidence budget $(N_+,N_-)=(8,2)$; resampling seed.
\ENSURE Held-out oriented AUC, train-fitted balanced threshold error, and
balanced log loss for fixed additive, fixed shunting, and fitted free additive.
\STATE Fix adjacent contrast $a\gets c$, $b\gets(c+1)\bmod n_{\mathrm{classes}}$; relabel trials from classes $\{a,b\}$ as $\theta\in\{0,1\}$
\STATE Draw subpool $S\subset\{1,\dots,P\}$, $|S|=N$, uniformly at random with the configured seed
\STATE $\Delta\mu_p \gets \E[X_{:,p}\mid\theta\!=\!1,\,\mathcal{T}_{\mathrm{lo}}^{\rm tr}]-\E[X_{:,p}\mid\theta\!=\!0,\,\mathcal{T}_{\mathrm{lo}}^{\rm tr}]$ \COMMENT{training-low difference, $p\in S$}
\STATE $S_+ \gets \{p\in S:\Delta\mu_p>0\}$,\quad $S_- \gets \{p\in S:\Delta\mu_p<0\}$ \COMMENT{training-defined evidence signs, not cell types}
\STATE $E\text{-}\mathrm{path} \gets$ top-$N_+$ of $S_+$ by $\Delta\mu_p$;\quad $I\text{-}\mathrm{path} \gets$ top-$N_-$ of $S_-$ by $-\Delta\mu_p$
\FOR{each training-low trial $t\in\mathcal{T}_{\mathrm{lo}}^{\rm tr}$}
  \STATE form pathway means $E_t,I_t$, then $V_t^{\add}=E_t-I_t$ and $V_t^{\sh}=E_t/(1+E_t+I_t+10^{-6})$
\ENDFOR
\STATE On training-low scores only, choose each score orientation, fit its
balanced-error threshold and balanced logistic calibration, and fit the
sign-constrained free comparator $E-\rho I-b$; record the exact threshold-mapped
additive rule from Proposition~\ref{prop:scalar-threshold}.
\FOR[{held-out high-state evaluation; equal pathway weights, no refitting}]{each trial $t\in\mathcal{T}_{\mathrm{hi}}^{\rm te}$}
  \STATE $E_t \gets \mathrm{mean}_{p\in E\text{-}\mathrm{path}} X_{t,p}$,\quad $I_t \gets \mathrm{mean}_{p\in I\text{-}\mathrm{path}} X_{t,p}$
  \STATE $V^{\add}_t \gets E_t-I_t$,\quad $V^{\sh}_t \gets E_t/\bigl(1+\max(E_t,0)+\max(I_t,0)+10^{-6}\bigr)$
\ENDFOR
\STATE Apply the fixed training orientation, threshold, and calibration to
$\mathcal{T}_{\mathrm{hi}}^{\rm te}$; compute AUC, balanced threshold error,
and balanced log loss.
\STATE \RETURN held-out metrics and the threshold-equivalence check.
\end{algorithmic}
\end{algorithm}

\paragraph{Dataset.}
The V1 transfer experiment uses the publicly released Stringer et al.\ 2019
two-photon oriented-stimulus dataset
~\citep{Stringer2019,StringerOrientedStimuliFigshare2019}. The released database
index lists one drifting-gratings recording from each of three mice (GT1--GT3),
and we analyzed all three. The analysis matches each
\path{gratings_drifting_*.npy} response file to \path{database.npy} and uses the
corresponding running-speed trace from \path{all_running.npy}. GT1, GT2, and GT3
contain $20{,}616$, $14{,}034$, and $11{,}311$ neurons measured over $4{,}282$,
$4{,}469$, and $4{,}442$ trials, respectively; together the recordings cover
$45{,}961$ neurons and $13{,}193$ trials. The small number of mice and single
stimulus protocol limit generality.

\paragraph{Fixed-template preprocessing.}
We use the deconvolved activity of each cell averaged over the response window
of each individual trial (one value per cell per trial; no averaging across
trials) as input to the supplementary fixed-template readout analyses.
These raw responses are nonnegative in the released array, so the V1 readout uses
rate-like inputs and does not z-score or mean-center them before evidence-sign selection.
The session stores each trial's drifting-grating direction as a continuous value in
$[0,2\pi)$; the fixed-template analysis bins these angles into
$n_{\mathrm{classes}}=12$ uniform direction classes, whereas the trained
analysis uses eight. No additional smoothing, dimensionality reduction, or quality
filtering is applied: all neurons in each analyzed session are retained for the upstream
candidate pool. The nested adjacent analysis uses base seed 20260615 and derives deterministic
session/pair/pool/resample/fold substreams with \texttt{SeedSequence}; all
mechanisms share the same sampled candidates.

\paragraph{Fixed-template locomotion-state split.}
The locomotion state variable is the per-trial mean running speed. We split
trials within each session into ``low-state'' and ``high-state'' subsets using
a tertile split. GT1 contains 1413/1413 low/high trials and GT2 contains
1475/1475. GT3 has many exact zeros at the lower-tertile threshold; retaining
all ties yields 2431 low-state and 1466 high-state trials. Trials between the
two thresholds are dropped to preserve a clear state contrast.

\paragraph{Adjacent $N_+=8,N_-=2$ fixed-budget candidate-access protocol.}
The analysis uses adjacent class pairs $c$ versus $c+1$, candidate pools
$N\in\{25,50,100,200,500\}$, a fixed $N_+=8,N_-=2$ selected evidence budget, and 80
candidate-pool resamples per class pair and session assigned across five outer
folds. Every selection, score-orientation, threshold, calibration, and
free-additive fit uses only outer-training low-state trials. These fixed scores
are evaluated on outer-held-out high-state trials using oriented AUC, balanced
threshold error, and balanced log loss.

\paragraph{Fixed-template readouts.}
For each neuron-resampling draw, candidate neurons are ranked by the low-state
class-mean difference. Positive-difference candidates are routed to the
downstream $E$ path and negative-difference candidates to the $I$ path. These
are functional evidence assignments, not recorded-cell identities. The code
requires at least the full configured positive- and negative-evidence counts; if either sign pool
is too small, the entire selection block is skipped rather than using a partial
budget. Included sample counts and skipped blocks are recorded in the released
source data.
The pooled drives are simple averages over the selected indices,
$E=\mathrm{mean}(X_E)$ and $I=\mathrm{mean}(X_I)$, which is equivalent to fixed
equal weights within each pathway, not equal weights per selected neuron across
pathways. For example, in the $N_+=8,N_-=2$ budget each negative-evidence neuron
has four times the coefficient magnitude of each positive-evidence neuron. The additive diagnostic uses
$V=E-I$. The shunting diagnostic uses
$V=E/(1+\max(E,0)+\max(I,0)+10^{-6})$ in code; the clipping is a numerical guard
and is inactive for the raw nonnegative responses except for floating-point
roundoff. This is a fixed equal-weight pathway-average convention, not the
implementation's \texttt{*\_normalized\_additive} model: the additive diagnostic has no
additive denominator, response normalization, batch normalization, or fitted
post-hoc scaling. Candidate selection is based on low-state trials, and performance is
evaluated on high-state trials. No synaptic weights, thresholds, or readout
parameters are fitted on the high-state evaluation trials.
This fixed scalar is not the trainable DendriNet: it has no Top-$K$ optimization,
tree, child conductance, internal tanh, or trained population decoder. Moreover,
the leak value $1$ is expressed in the arbitrary units of the deconvolved
response, so we include training-only scale/leak sensitivity and
the free-$\rho$ additive comparator of Proposition~\ref{prop:scalar-threshold}.

\paragraph{Training-only effective-leak transfer.}
Only $r=\text{leak}/\text{input scale}$ is identifiable from the fixed-shunt
score. We evaluate $r\in\{0.01,0.1,1,10,100\}$ with a strict
leave-one-session-out rule. For each held-out session, response variant, and
candidate-pool size, the selector minimizes the unweighted mean of the other
two sessions' session-mean outer-training log loss. Ties within $10^{-12}$
prefer $r=1$ and then the smaller value. The selected $r$ is frozen before any
held-out-session endpoint is accessed; this separation is verified by
perturbing the held-out endpoint values. Evidence-pathway selection and score
calibration still use low-running training trials from the evaluated session,
while its high-running trials remain evaluation-only. The training loss used
for selection is a resubstitution estimate after training-only calibration rather than an additional
inner-trial split. This secondary sensitivity is based on only two training
sessions per fold. Increasing $r$ weakens the input-dependent
denominator relative to the leak, so selection of $r=10$ indicates canonical
response-scale mismatch and movement toward an excitation-dominated limit; it
is not evidence that stronger divisive action is beneficial.

\paragraph{Training-only, label-free running control.}
The running-control analysis uses the same nonnegative activity domain for its
raw and corrected arms. For each neuron, a one-step class-fixed-effect
quasi-Poisson score estimates the local slope in
\[
 \mathbb E[X_n\mid c,r]
 =\mu_{cn}\exp\!\left[\beta_n(r-\bar r_{\rm train})\right]
\]
using only an outer-training partition spanning low and high running. The
label-free correction applied to a target trial is
\[
 X_{{\rm corr},n}=\max(X_n,0)
 \exp\!\left[-\beta_n(r-\bar r_{\rm train})\right].
\]
Inner-fold fits provide out-of-fold corrected low-state responses for readout
training; a fit on the full outer-training partition is applied to held-out
high-state trials. Target labels are never passed to the application function.
Candidate populations and positive/negative evidence indices are selected once
from the raw-nonnegative training response and then held identical across raw
and corrected arms. Score orientation, threshold, fitted additive ratio, and
probability calibration remain nested in the same outer fold.

The analysis crosses adjacent/opposite class pairs, candidate pools
$100,200,500,1000,2000$, evidence budgets $8/2$ and $40/10$, and 100 resamples
per pair. In all 15 session--fold applications, numerical exponent clipping is
zero and outputs are finite and nonnegative; the held-out correction-factor
1st/99th percentiles span 0.127--0.433 and 1.424--1.697 across folds. The
recorded responses were already nonnegative, so raw and raw-nonnegative
preprocessing give identical values. We excluded an inverse-\texttt{log1p}
transformation because it placed a session-weighted 32.41\% of held-out
corrected responses at zero.
The retained correction is a nuisance-sensitivity model and does not identify a
cortical gain mechanism.

\paragraph{Independent Allen V1 scope control.}
\label{methods:allen-v1-scope}
To assess transfer beyond the three Stringer recordings, we analyzed five
prespecified Allen Brain Observatory Visual Coding VISp recordings from five
mice~\citep{DeVriesEtAl2020}. For each drifting-grating trial, the input is the
nonnegative difference between mean stimulus-window and prestimulus $\Delta F/F$;
there is one cell value per trial and no averaging across trials. The recordings
contain 598--599 eligible trials and 207--261 cells. Training-only selection and
calibration follow the fixed-feature state-transfer protocol, with each mouse
as the statistical unit. The prespecified consistency rule was not met: the
raw state interaction was positive in three of five mice, the narrow interaction
exceeded the wide interaction in three of five, and running correction reduced
its absolute magnitude in two of five. We therefore retain this analysis as a
heterogeneous nonconfirmation, not an independent replication of the Stringer
pattern or evidence for a cortical shunting circuit.

\paragraph{Statistical unit.}
For the main Stringer analysis, the three recordings from three mice are the
independent biological units. Neuron-pool
subsamples, class pairs, and model seeds quantify within-session algorithmic
variation and are not treated as additional animals. The main figure reports
the mean $\pm$ SEM across the three sessions; session-level values and
session-clustered or hierarchical intervals are retained in the released source
tables. Pooled bootstrap SEM is used only as within-session Monte Carlo
precision.

\subsection{Scope, assumptions, and limitations}
\label{methods:scope-assumptions}

\subsubsection{Consolidated assumptions and limitations}
\label{methods:assumptions}

The mathematical results of this paper rest on a small set of assumptions that
we collect here for ease of reference. Table~\ref{tab:assumptions-summary}
maps each formal claim to its assumptions and to the section where they are
discussed.

\begin{table}[ht]
\centering
\caption{\textbf{Assumptions underlying each formal claim.}}
\label{tab:assumptions-summary}
\footnotesize
\setlength{\tabcolsep}{4pt}
\renewcommand{\arraystretch}{1.08}
\begin{tabularx}{\linewidth}{@{}>{\raggedright\arraybackslash}p{0.22\linewidth}>{\raggedright\arraybackslash}X>{\raggedright\arraybackslash}p{0.22\linewidth}@{}}
\toprule
\textbf{Result} & \textbf{Key assumptions} & \textbf{Discussed in} \\
\midrule
Theorem~\ref{thm:local-equivalence}
  (local containment; conditional equality)
  & Nonnegative presynaptic activities $\xe,\xI\ge 0$ and nonnegative synaptic
    weights $(\we,\wi)\ge 0$; $Z=(\xe,-\xI)$ is only a comparator-coded moment
    representation; finite class-conditional second moments and a nondegenerate
    pooled within-class covariance (Gaussianity is needed only for an LDA
    interpretation, not for the Rayleigh optimization); and a valid first-order
    linearization at a positive operating point. These suffice for containment.
    Equality additionally requires an additive-optimal effective ray with a
    positive self-consistent shunting realization, as defined in
    Methods~\ref{methods:self-consistent}. Extra constrained ties additionally
    require equality of projective feasible-ray sets and an
    $\mathcal F$-feasible self-consistent representative as stated in
    Proposition~\ref{prop:constraint-equivariant}.
  & \S\ref{subsec:local};
    Methods~\ref{methods:cone-lda}, \ref{methods:local-equivalence-proof}\\
Corollary~\ref{cor:activation-local-tie}
  (activation preserves the local relation)
  & Smooth monotone scalar map $\phi(V)$ at the soma with
    $\phi'(V_0)>0$ at the induced operating point; assumptions of
    Theorem~\ref{thm:local-equivalence}.
  & \S\ref{subsec:local}; this corollary applies to scalar
    \emph{output} activation; \emph{internal-tree} activation
    is outside the corollary and is treated by the activated-tree Jacobian. \\
Activation-aware tree extension
  Eq.~\eqref{eq:activated-tree-dprime}
  & Smooth branch activation with stated $(\kappa_\phi,b)$; valid first-order Jacobian
    linearization at the activated operating point; internal activations are
    analyzed as an extension to, not a replacement for, the passive criterion.
  & \S\ref{subsec:trees}; Methods~\ref{methods:tree-optimization},
    \ref{methods:dendritic-implementation}\\
Lemma~\ref{lem:decorrelation}
  (depth-wise decorrelation)
  & Shared component $Z_g$ has finite variance and its coefficient composes as
    $a_\ell=a_0\prod_{k=1}^{\ell}\kappa_k$ with $0<\kappa_k<1$;
    $\Cov(\varepsilon_{\ell,1},\varepsilon_{\ell,2})=0$,
    $\Cov(\varepsilon_{\ell,j},Z_g)=0$, and
    $\Var(\varepsilon_{\ell,j})\ge\sigma_{\mathrm{priv}}^2>0$.
  & \S\ref{subsec:trees}; Methods~\ref{methods:tree-linearization}\\
Proposition~\ref{prop:additive-flattening}
  (additive flattening)
  & Passive linear current-mode additive recursion with nonnegative local
    input and transfer coefficients; no node nonlinearity.
  & Methods~\ref{methods:tree-flattening}\\
Gain--load--alignment principle
  Eqs.~\eqref{eq:regime-caricature}--\eqref{eq:tree-criterion}
  & Positive E/I operating point; local linearization of the tree; measured or
    controlled baseline fidelity $B(\mathcal T,r)$; shared gain component that the
    branch-local divisor can attenuate; residual denominator variability;
    approximate decoupling of gain and load contributions; bounded
    second-order curvature of $V^{\sh}$
    (Methods~\ref{methods:delta-method}).
  & \S\ref{subsec:beyond-local};
    Methods~\ref{methods:regime-map}\\
\bottomrule
\end{tabularx}
\end{table}

\paragraph{Modeling limitations.}
Our framework is passive, quasi-static, and conductance-based. Active
mechanisms (NMDA voltage-dependence, dendritic spikes, sodium
backpropagation) are deferred to follow-up work; Methods~\ref{methods:cable-derivation}
discusses when this approximation is valid and what it likely changes. The
empirical V1 transfer uses one public dataset and one state variable
(locomotion); the framework does not yet incorporate spike timing or
short-term synaptic dynamics. The main V1 analysis trains DendriNet on eight-way
direction classification; the fixed-template appendix uses adjacent direction
pairs as a comparator and input-scale stress test. A complementary direction is to combine this
population-readout criterion with active dendritic learning frameworks that
model compartment-specific teaching signals~\citep{Guerguiev2017,Sacramento2018}.

\subsection{Reproducibility}
\label{methods:grp-repro}

\subsubsection{Statistical analysis and replication units}
\label{methods:statistics}

The replication unit follows the data-generating process. Independently
generated synthetic libraries and paired computational seeds quantify
simulation or optimization variation; they are not biological replicates.
For V1, each recording from a distinct mouse is the biological unit, and neuron resamples, noise
draws, and training seeds are first aggregated within recording. Main-figure
$t$ intervals are two-sided 95\% intervals over the stated paired seed or
recording units; bootstrap intervals are used only where captions explicitly
identify resampling, and SEM bands are descriptive summaries rather than
hypothesis tests. Pairing is performed before contrasts whenever the mechanisms
share data, seeds, or perturbations. Holm correction is applied only within the
prespecified comparison families named in the corresponding protocol; isolated
surface summaries and operating sensitivities are labeled descriptive. The
disjoint-seed tangent validation and exploratory sensitivity panels are
identified as such and are not pooled with confirmatory evidence.

For the exact-inventory hierarchy, parameters were selected on pilot seeds 0--1
before the seed-100--107 held-out test. Tangent matching, the coordinate bridge,
and operating-point sensitivity were subsequent analyses using disjoint seeds;
they are reported as validation or sensitivity analyses rather than
preregistered confirmations.

\subsubsection{Code organization and reproducibility entry points}
\label{methods:code-organization}

The reproducibility materials comprise a reusable model package, the paper
repository, and the upstream data and results needed to rebuild each analysis.
The reusable package contains
the dendritic E/I implementation, branch-local additive and shunting modules,
sparse Top-$K$ layers, pathway transfer, a minimal configuration, CPU examples,
and tests; manuscript sources and large result caches are stored separately.
The Git paper repository contains the manuscript, figure and analysis scripts,
reported configurations, tests, tracked source tables and compact aggregates,
and \path{reproducibility/manifest.yaml}. The
manifest maps every figure to its generating script, source/result
artifacts, and, where necessary, external raw-data or cluster-result roots.

The panel-level evidence registry additionally labels each panel as symbolic
theory, standalone numerical simulation, a forward-equation calculation, a
standard-path trained DendriNet fit, a bespoke DendriNet intervention, or a
fixed-feature empirical analysis. Thus adjacency in a composite figure does
not imply a shared fitting pipeline. Every final trained fit was launched from
one read-only, versioned execution bundle containing a single Git-tracked
DendriNet source tree and the paper configurations and scripts. Before training,
the launcher verified a recursive source-tree hash; each result records that hash, both Git
commits, and the exact configuration hash. Historical runs from earlier code
snapshots are retained as provenance but are not pooled with the final runs.

The figure-generation target validates each local source or aggregate before
rendering. CPU theory and synthetic analyses have a separate target;
trained-network and V1 fits additionally require the cited public data, the
DendriNet training environment, and the result roots listed in the manifest.
The manuscript source archive contains everything needed to compile and inspect
the paper, but not the large recordings or checkpoint collections needed for
full computational reproduction.

\paragraph{Code and data availability.}
\ifneuripsanonymous
\else
The paper source, reproduction scripts, manifest, configurations, and tracked
figure inputs are available at
\par\smallskip\noindent
\url{https://github.com/houman1359/dendritic-gain-load-criterion}
\par\smallskip\noindent
All final trained fits use DendriNet commit \texttt{7904160d2c2}, whose full
hash is archived in the paper's release manifest, from
\url{https://github.com/KempnerInstitute/DendriNet}.
\par\smallskip\noindent
Raw V1 recordings remain available from the cited Figshare repository and
through AllenSDK.
\fi

\subsubsection{Random seeds used}
\label{methods:seeds}

Table~\ref{tab:seeds-summary} lists the randomization and replication for the
main and supplementary experiments. Script, configuration,
source-data, and result paths are not duplicated in the typeset table; their
complete machine-checked mapping is in
\path{reproducibility/manifest.yaml}.

\begin{table}[htbp]
\centering
\caption{\textbf{Random seeds and replication used for each experiment.}}
\label{tab:seeds-summary}
\footnotesize
\setlength{\tabcolsep}{4pt}
\renewcommand{\arraystretch}{1.06}
\begin{tabularx}{\linewidth}{@{}>{\raggedright\arraybackslash}p{0.37\linewidth}>{\raggedright\arraybackslash}X@{}}
\toprule
\textbf{Experiment} & \textbf{Randomization or replication} \\
\midrule
Local-containment/conditional-equality sweep (Theorem~\ref{thm:local-equivalence};
Supplementary Fig.~\ref{fig:local-equivalence-supp}) &
160 random libraries; exact two-orientation NNLS and self-consistency check. \\
Conductance-mass boundary (Fig.~\ref{fig:framework-main}D) &
Seed 660; 360 two-class multivariate-lognormal four-input libraries; analytic
pooled moments over the pinned mass/share grid. \\
Reduced gain--load boundary and decoder policies
(Fig.~\ref{fig:regime-map}B--C) &
Equation-derived fixed-template and direction-preserving reoptimized curves;
no fitted samples. \\
Split-sample second-order local surrogate (Fig.~\ref{fig:regime-map}E) &
Seed 20260714; 160 lognormal and 160 gamma attempts; 1,500/2,500/8,000
observations per class for independent fit/calibration/test splits; 166
verified positive realizations; 154 attempts had no positive realization. \\
Implemented-equation population scaling (Fig.~\ref{fig:regime-map}D;
Supplementary Fig.~\ref{fig:identified-ceiling-supp}A--E) &
Seeds 400--407; independent 100,000-trial calibration and test splits for all
57,600 complete-grid rows; exact linear moments and converged nonlinear
quadrature. \\
Contact-capped inaccessible ceiling (Appendix
Supplementary Fig.~\ref{fig:identified-ceiling-supp}F) &
Exact differential-noise curve plus 16 Monte Carlo seeds and 12,000 balanced
trials per size and condition. \\
Exact-inventory hierarchical gain construction
(Fig.~\ref{fig:criterion-morphology}A--C) &
Pilot seeds 0--1 excluded; held-out paired confirmation seeds 100--107;
12,000 train and 12,000 test trials per gain level. \\
Exact-inventory tangent-additive validation
(Fig.~\ref{fig:criterion-morphology}B--C) &
Independent paired validation seeds 200--207; clean training-only
layer-global anchors; 12,000 train and 12,000 test trials per gain level. \\
Exact-inventory coordinate bridge (Fig.~\ref{fig:criterion-morphology}E) &
Disjoint outer seeds 300--307; 12,000 clean-train, calibration, and fresh
combined-nuisance test examples per cell. \\
\bottomrule
\end{tabularx}
\end{table}

\begin{table}[htbp]
\centering
\footnotesize
\setlength{\tabcolsep}{4pt}
\renewcommand{\arraystretch}{1.06}
\begin{tabularx}{\linewidth}{@{}>{\raggedright\arraybackslash}p{0.37\linewidth}>{\raggedright\arraybackslash}X@{}}
\multicolumn{2}{@{}l}{\textbf{Table~\ref{tab:seeds-summary} continued.}} \\[2pt]
\toprule
\textbf{Experiment} & \textbf{Randomization or replication} \\
\midrule
Exact-inventory operating sensitivity
(Supplementary Fig.~\ref{fig:exact-inventory-operating-sensitivity-supp}) &
Twenty-four-cell grid; paired seeds 100--107; 12,000 train and 12,000 test
trials per cell. \\
Passive nonlinear-capacity control
(Supplementary Fig.~\ref{fig:capacity-reduction-supp}) &
Held-out paired seeds 500--507; nested 125--12,000-example training prefixes,
validation-only selection, and a separate 12,000-example test draw. \\
Passive structural-reduction control
(Fig.~\ref{fig:criterion-morphology}F) &
Held-out paired seeds 600--607; identical observations, divisor sensors, and
two-parameter scalar decoders across five structural arms. \\
Tree-depth simulation (Fig.~\ref{fig:criterion-morphology}D;
Supplementary Fig.~\ref{fig:tree-depth}) &
10,000 trials per class per depth; base seed 11. \\
Distributional matched-comparator atlas (Supplementary Fig.~\ref{fig:nongaussian-supp}) &
Paired seeds 1729--1740; 96 distributional cells and 1,152 seed-level threshold
comparisons. \\
Contextual optimization-policy factorial (Appendix
Supplementary Fig.~\ref{fig:contextual-factorial-supp}) &
Seeds 10420--10427; 192 completed runs and 64 unique effective-learning-rate
mechanism pairs after duplicate collapse. \\
Resource-matched branch redistribution (Appendix
Supplementary Fig.~\ref{fig:resource-matched-morph-supp}) &
Seeds 11000--11007; 192 runs and 96 paired mechanism contrasts. \\
Nuisance-support factorial (Fig.~\ref{fig:locality}C--D;
Supplementary Fig.~\ref{fig:locality-ei-fixed-supp}) &
Seeds 12000--12007; 128 runs and 64 raw paired mechanism contrasts. \\
Fixed-total E/I contact-allocation surface
(Fig.~\ref{fig:framework-main}E--F) &
Seeds 12100--12107; 192 runs and 96 paired mechanism contrasts. \\
Theory/codebase branch-rule check (Fig.~\ref{fig:ei-circuit}D) &
Seeds 7301--7308; 120,000 Monte Carlo samples per seed and grid point. \\
Trained ordered-E/I gain--load sweep
(Fig.~\ref{fig:ei-circuit}E--F) &
Eight seeds in each of 72 architecture/data cells; the plotted $\alpha=1$ slice
has 384 runs and 192 raw pairs, with four morphologies averaged within each
seed block ($n=8$ per point). \\
\bottomrule
\end{tabularx}
\end{table}

\begin{table}[htbp]
\centering
\footnotesize
\setlength{\tabcolsep}{4pt}
\renewcommand{\arraystretch}{1.06}
\begin{tabularx}{\linewidth}{@{}>{\raggedright\arraybackslash}p{0.37\linewidth}>{\raggedright\arraybackslash}X@{}}
\multicolumn{2}{@{}l}{\textbf{Table~\ref{tab:seeds-summary} continued.}} \\[2pt]
\toprule
\textbf{Experiment} & \textbf{Randomization or replication} \\
\midrule
Complete branch-local mechanism--initializer factorial (text) &
Eight seeds in each of 72 cells and all four arms: 2,304 runs, 1,152
within-policy mechanism pairs, and 576 complete four-arm seed blocks. \\
Matched Top-$K$ activation policy (Appendix
Supplementary Fig.~\ref{fig:topk-comparator-ladder-supp}A--C) &
Eight paired seeds, two mechanisms, and maintained versus initialization-only
schedules; 32 completed runs. \\
Seven-arm Top-$K$ comparator (Appendix
Supplementary Fig.~\ref{fig:topk-comparator-ladder-supp}D--F) &
Seeds 12640--12647; 56 completed runs and 48 prespecified
comparator--shunting pairs. \\
DendriNet denominator-only intervention (Appendix
Supplementary Fig.~\ref{fig:denominator-only-load-supp}) &
Seeds 13000--13007 in lognormal and gamma families; 16 tasks, 80 unique fits,
and 768 paired contrast rows. \\
Frozen-feature CIFAR-10 normalization bridge
(Fig.~\ref{fig:cifar-normalization-main}) &
Seeds 4101--4105; frozen ResNet-18 features; five trained readouts plus an
oracle; four gain supports, four gain levels, and four divisor-noise levels. \\
Resource-matched trained V1 population readout
(Fig.~\ref{fig:v1-main}) &
1,080 frozen clean/private-noise fits plus 288 prespecified
running-state-interaction fits with an untouched same-state holdout; three
sessions, four fresh seeds, and training-only nested neuron pools. \\
Adjacent $N_+=8,N_-=2$ fixed-template candidate-access diagnostic
(Supplementary Fig.~\ref{fig:v1-fixed-template-supp}) &
Three Stringer sessions; seed 20260615; 80 candidate-pool resamples per pair and
session assigned across five outer folds. \\
Training-only V1 multiplicative running sensitivity
(Supplementary Fig.~\ref{fig:v1-running-recovery-supp}) &
Three sessions; seed 20260718; five outer and five inner folds; 100 resamples per
pair; identical raw-derived E/I indices across response variants. \\
Trained V1 fixed-total contact allocation &
96 runs and 48 mechanism pairs; four $(N_E,N_I)$ values with $N_E+N_I=50$
across three sessions and four model seeds. \\
\bottomrule
\end{tabularx}
\end{table}

\subsubsection{Compute resources, runtime, and reproducibility environment}
\label{methods:compute}
Analytic simulations, cone feasibility checks, V1 fixed-template analyses, figure
generation, and manuscript compilation were executed on Linux CPU nodes
\ifneuripsanonymous
\else
(institutional shared compute cluster; x86\_64).
\fi
A representative CPU environment had Intel Xeon Gold 6426Y processors
(32 logical CPUs visible) and approximately 1.0\,TiB RAM. Each trained-network
task requested one GPU. Regenerating all 4,480 retained fits from the immutable
release used 101.42 allocated GPU-hours: 55.25 for the 2,304-run complete
mechanism--initializer factorial, 30.30 for the 1,368 horizon-corrected V1 fits,
and 15.87 for the remaining standard and bespoke DendriNet controls. The 1,368
deterministic V1 post-fit evaluations used another 10.21 GPU-hours, for a final
release total of 111.63 GPU-hours. These values are sums of Slurm array-task
wall time with one allocated GPU per task. Historical development, tuning,
failed jobs, and superseded migration outputs are retained in the manifest but
are neither required for this reproduction total nor pooled into reported
estimators.

In a representative reproduction run, validation checks took 7.8\,s, main figures 204.7\,s,
supplementary figures 86.4\,s, additional diagnostics 13.0\,s, and manuscript
compilation 11.8\,s, for a total of approximately 5.4 minutes. This timing
excludes model training.

\subsubsection{External assets and license compliance}
\label{methods:licenses}
Table~\ref{tab:asset-licenses} lists the external datasets used in the study,
together with their versions and licensing terms. Pinned Python package
versions for the analysis and figure environment are given in
\path{requirements-figures.txt}, and per-figure provenance (generating script,
source and result tables, and result roots) is recorded in
\path{reproducibility/manifest.yaml}.

\begin{table}[htbp]
\centering
\caption{\textbf{External datasets and licensing terms used in this study.}}
\label{tab:asset-licenses}
\footnotesize
\begin{tabularx}{\textwidth}{@{}>{\raggedright\arraybackslash}p{0.25\textwidth}>{\raggedright\arraybackslash}p{0.31\textwidth}>{\raggedright\arraybackslash}p{0.31\textwidth}@{}}
\toprule
Asset & Use in study & License / terms \\
\midrule
CIFAR-10 &
Frozen-feature machine-learning normalization bridge
(Fig.~\ref{fig:cifar-normalization-main}) &
Original University of Toronto distribution and technical report cited; the
distribution page does not state a standalone license
~\citep{Krizhevsky2009CIFAR}. \\
Torchvision ResNet-18 \texttt{IMAGENET1K\_V1} weights &
Frozen feature extractor; weights are downloaded at analysis time and are not
redistributed with the paper &
ResNet architecture and ImageNet source cited; use remains subject to the
upstream terms~\citep{HeEtAl2016,DengEtAl2009ImageNet}. \\
Stringer et al.\ 2019 oriented-stimulus V1 dataset &
Fixed-budget V1 readout plus trained implementation and robustness diagnostics
(Fig.~\ref{fig:v1-main}; supplementary V1 controls) &
Janelia Figshare dataset DOI 10.25378/janelia.8279387.v3; CC BY-NC 4.0; original paper and dataset cited~\citep{Stringer2019,StringerOrientedStimuliFigshare2019} \\
Allen Brain Observatory Visual Coding &
Five-mouse independent scope control for the downstream state-transfer
diagnostic &
Public VISp two-photon recordings accessed through AllenSDK; dataset paper
cited~\citep{DeVriesEtAl2020}; use remains subject to the upstream Allen
Institute terms. \\
\bottomrule
\end{tabularx}
\end{table}

\clearpage
\section{Supplementary Results}
\label{sec:app-supp-results}
\renewcommand{\figurename}{Supplementary Fig.}
\renewcommand{\thefigure}{S\arabic{figure}}
\setcounter{figure}{0}

The figures are organized around five scientific questions: the local
comparator boundary (Sec.~\ref{supp:theory-boundaries}), tree structure and
population access (Sec.~\ref{supp:tree-population}), resource-matched morphology
and locality (Sec.~\ref{supp:morphology-locality}), optimization in learned
networks (Sec.~\ref{supp:learned-fairness}), and transfer to V1 recordings
(Sec.~\ref{supp:v1-analyses}).

\subsection{Theory and comparator boundaries}
\label{supp:theory-boundaries}

The first group tests the local-containment theorem across input distributions
and separates gain suppression from variability introduced by the denominator.

\begin{figure}[htbp]
\centering
\includegraphics[width=0.99\textwidth]{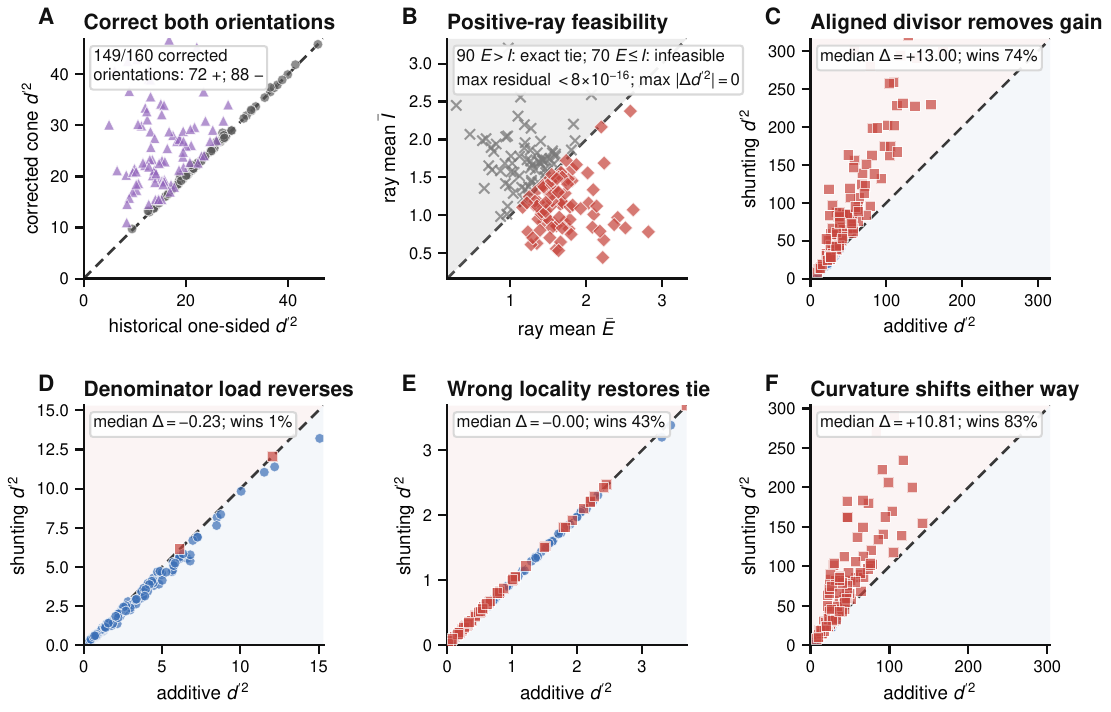}
\caption{\textbf{Local containment and finite-amplitude deviations.}
\textbf{(A)}~The two cone orientations are selected in 72 and 88 of 160
libraries; restricting the solve to one orientation changes 149 optima.
\textbf{(B)}~All 90 rays whose mean excitatory drive exceeds their mean
inhibitory drive have a positive self-consistent shunting realization, whereas
the other 70 do not. Realizable rays show numerical equality between local
shunting and additive discriminability (maximum absolute residual
$<8\!\times\!10^{-16}$). \textbf{(C--F)}~Finite-amplitude examples show that an
aligned divisor can remove shared gain, denominator load can reverse the
advantage, mismatched divisor locality can restore a tie, and curvature can
shift performance in either direction. Each condition contains 150
libraries.}
\label{fig:local-equivalence-supp}
\end{figure}

\begin{figure}[htbp]
\centering
\includegraphics[width=0.99\textwidth]{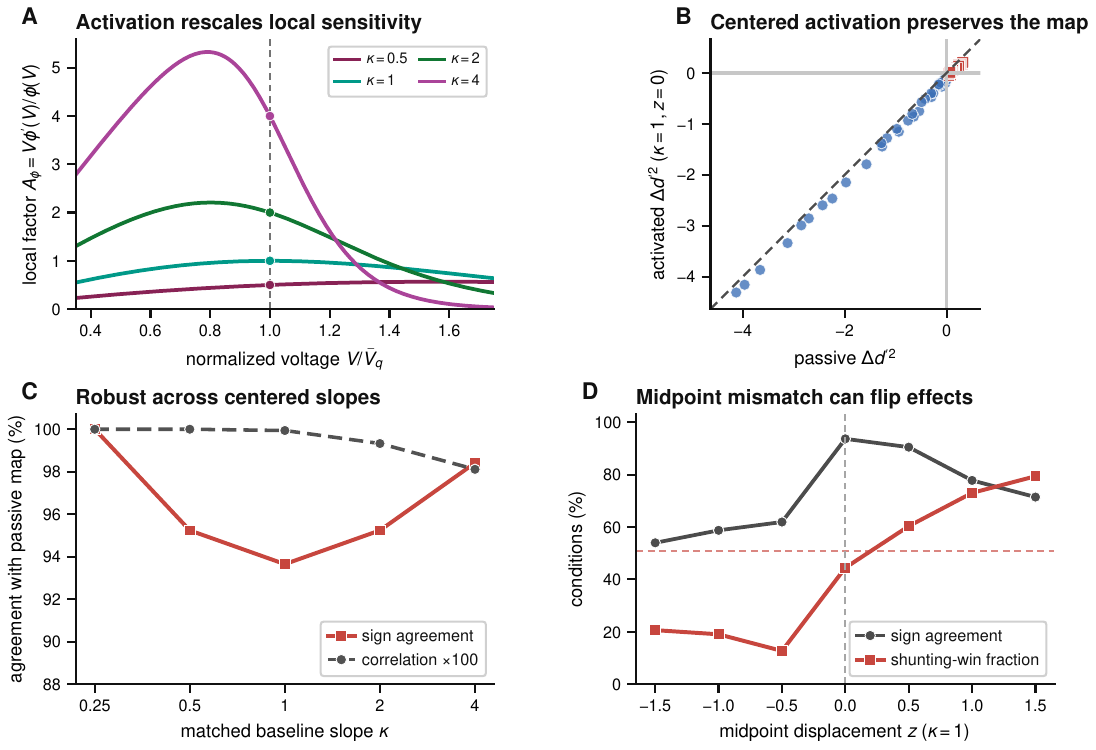}
\caption{\textbf{Terminal activation preserves the centered theory regime but
can expose operating-point mismatch.}
We apply
$\phi_q(V)=\{1+\tanh[\kappa(V/\bar V_q-1)-z]\}/2$ to the additive and
shunting voltages in the exact single-branch gain--load experiment underlying
Fig.~\ref{fig:ei-circuit}D. Here $\bar V_q$ is each rule's own
passive baseline voltage, so the comparison is not confounded by their
different voltage scales. Each of the 63 gain--load cells uses eight paired
seeds and 120,000 trials per seed.
\textbf{(A)}~The exact local multiplier is
$A_\phi=\kappa(V/\bar V_q)\{1-\tanh[\kappa(V/\bar V_q-1)-z]\}$; at the matched
midpoint $z=0$ and baseline $V=\bar V_q$, both rules receive the same factor
$A_\phi=\kappa$.
\textbf{(B)}~For $\kappa=1,z=0$, activated versus passive
shunting-minus-additive $d'^2$ has $r=0.999$ and 94\% sign agreement (circles:
passive additive-favored; squares: passive shunting-favored).
\textbf{(C)}~Across centered slopes $\kappa=0.25$--4, correlations are
0.981--1.000 and sign agreement is 94--100\%.
\textbf{(D)}~Displacing the midpoint changes the fraction of shunting-favored
cells and can change which mechanism has larger $d'^2$, demonstrating that activation robustness is
conditional on the operating regime. A strictly increasing terminal
activation nevertheless preserves sample ordering, ROC AUC, and optimal
scalar-threshold error exactly. These panels test terminal activation;
internal activation instead changes pathwise Jacobians as derived in
Methods.}
\label{fig:activation-sensitivity-supp}
\end{figure}

\begin{figure}[htbp]
\centering
\includegraphics[width=0.99\textwidth]{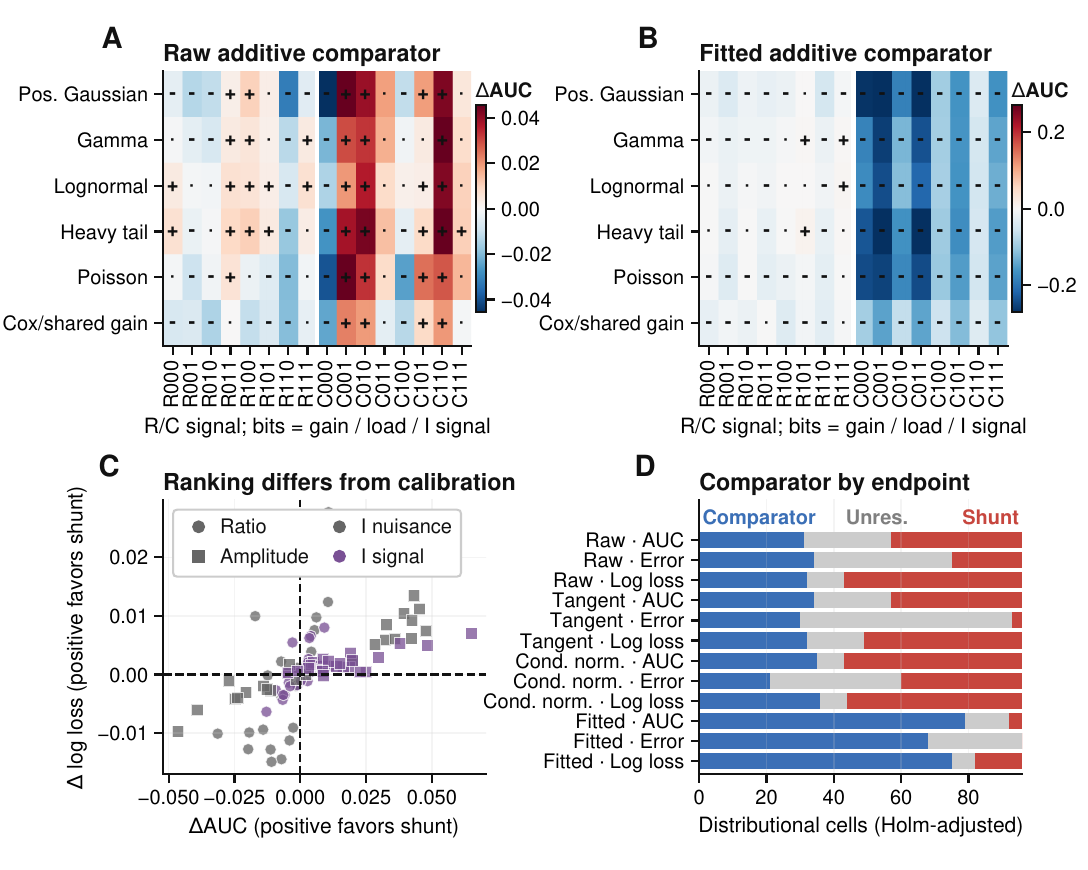}
\caption{\textbf{The finite-amplitude ordering depends on the input distribution
and comparator.} Six positive-data families (rectified Gaussian-like, gamma,
lognormal, heavy-tailed mixture, Poisson, and Cox/shared-gain) and two signal
structures (fixed-total ratio and fixed-ratio common amplitude) are crossed with
three binary nuisance factors (shared gain, independent I-stream background,
and an I-stream class signal), giving 96 conditions with 12 paired
seeds per condition. All fitting and calibration use training data, and all
panels report held-out performance. \textbf{(A--B)}~Shunting-minus-additive AUC
for raw $E-I$ and fitted sign-constrained $E-\rho I-b$ comparators; for held-out
AUC, shunting is below the fitted comparator in most conditions.
\textbf{(C)}~Ranking and probability
calibration can reverse independently. \textbf{(D)}~Holm-adjusted outcome
counts across four additive comparators and three endpoints. The additive
threshold mapped from the selected shunt agrees exactly with its binary
predictions. Distribution, nuisance structure, comparator, and endpoint jointly
determine the observed sign.}
\label{fig:nongaussian-supp}
\end{figure}

\begin{figure}[htbp]
\centering
\includegraphics[width=0.99\textwidth]{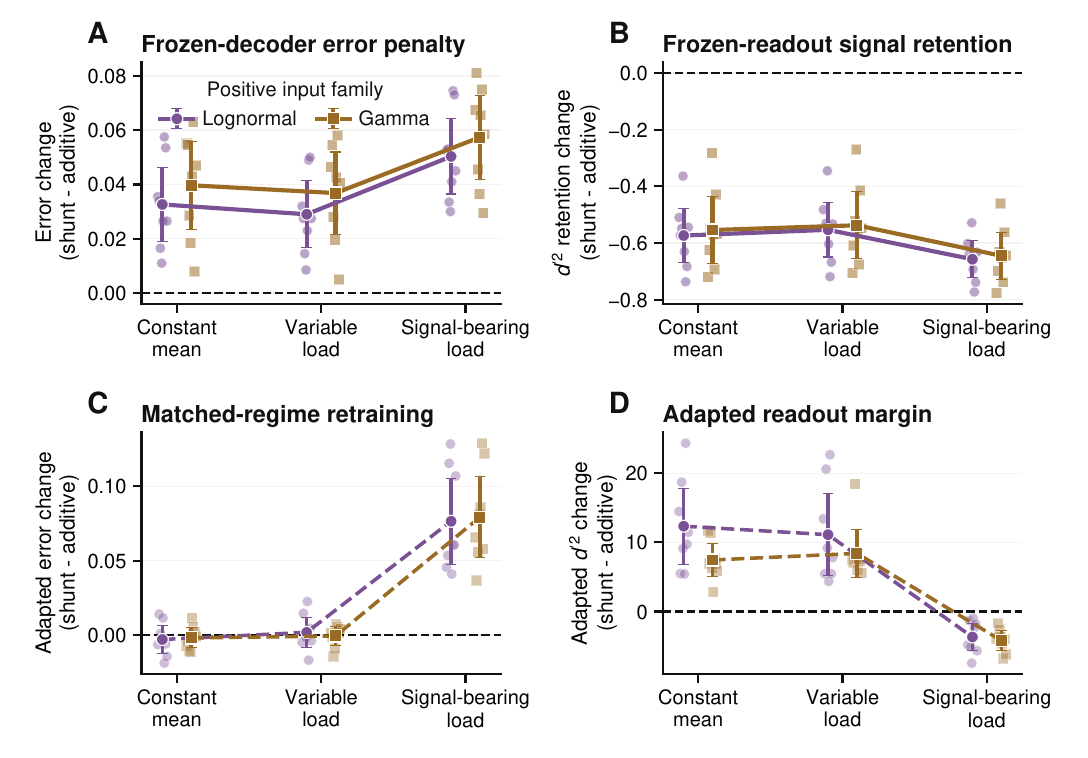}
\caption{\textbf{Adverse class-correlated denominator conductance has the largest
mean cost.} A passive $[8]$ DendriNet with post-voltage activation disabled receives an
external distal conductance that affects only the shunting denominator;
additive outputs and distal numerators are unchanged. Lognormal and gamma input
families each use eight paired seeds, and all nonzero regimes have mean
conductance 3. This controlled intervention uses the released DendriNet forward
model with a bespoke trainer and is not a standard-path fit. \textbf{(A--B)}~Without refitting, constant and independently
varying conductance have mean balanced-error increases of $3.62$ and $3.29$
percentage points, whereas the mean increase under adverse class-correlated conductance is
$5.38$ percentage points and is accompanied by the largest loss of signal.
\textbf{(C--D)}~Regime-specific retraining
largely recovers performance under constant or independent conductance but not
under adverse class-correlated conductance. The descriptive mean
shunting-minus-additive error difference is $+7.79$ percentage points, and the
margin-$\dprime^2$ difference is negative.
The two input families agree in direction; none of 60 sign-flip tests survives
Holm correction (minimum adjusted $p=0.469$).}
\label{fig:denominator-only-load-supp}
\end{figure}

\clearpage
\subsection{Tree structure and population access}
\label{supp:tree-population}

These analyses ask how local divisors, tree depth, and access to gain sensors
alter constrained population readout.

\begin{figure}[htbp]
\centering
\includegraphics[width=0.99\textwidth]{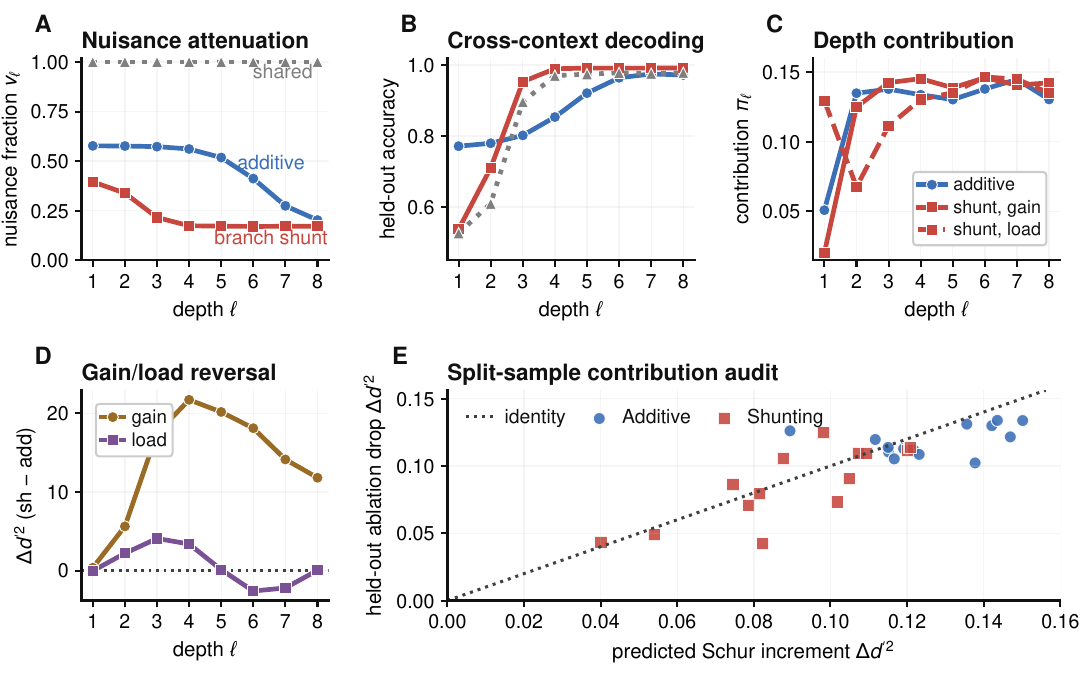}
\caption{\textbf{Local divisors can attenuate shared gain across depth in a
standalone controlled simulation.}
\textbf{(A--B)}~In a controlled gain-dominated regime, repeated
branch-local divisors attenuate the dominant nuisance mode and improve
cross-context decoding with depth; additive and shared-divisor controls do not
show the same hierarchical attenuation. \textbf{(C)}~Ordered
distal-to-proximal Schur contributions show how gain and load change the useful
increment of each stage. \textbf{(D)}~Depth is not intrinsically helpful:
gain-dominated inputs produce a positive shunting-minus-additive gap over
intermediate depths, whereas denominator-load-dominated inputs cross below
zero when load and fidelity costs outweigh further gain suppression.
\textbf{(E)}~Schur increments estimated on one sample predict ablation drops on
a disjoint sample. The additive and shunting correlations are $r=0.368$ and
$r=0.787$, respectively ($r=0.809$ when pooled). As in
Lemma~\ref{lem:decorrelation}, the interpretation requires sufficiently local
and complete divisors.}
\label{fig:tree-depth}
\end{figure}

\begin{figure}[htbp]
\centering
\includegraphics[width=0.98\textwidth]{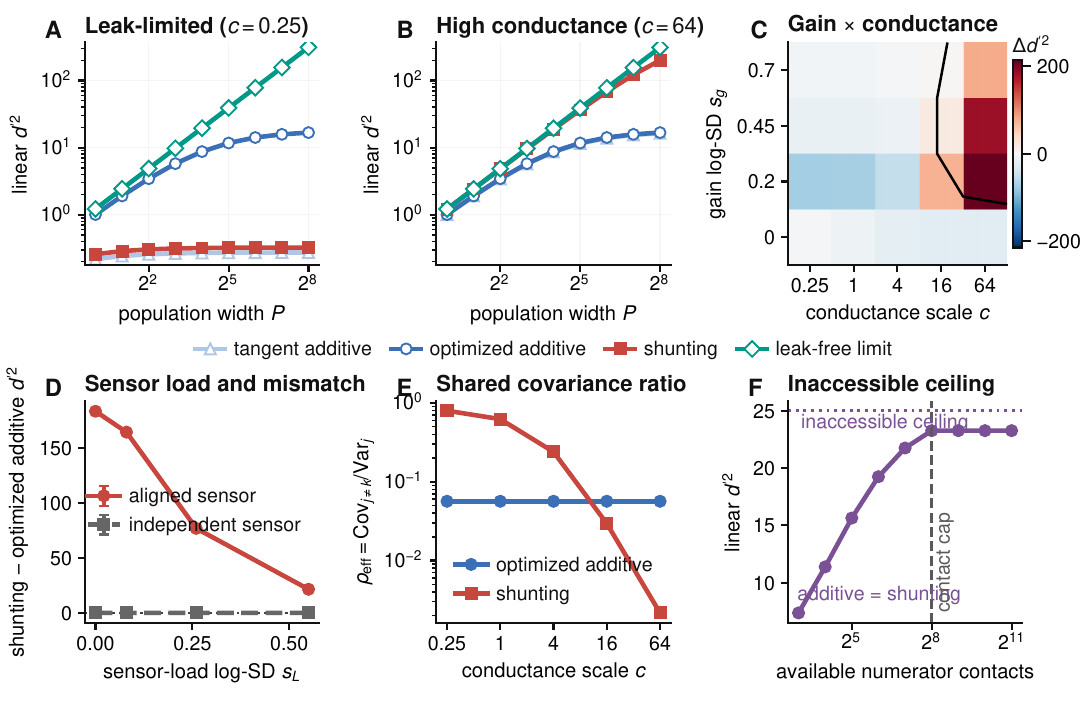}
\caption{\textbf{Population scaling improves when a high-conductance shunting
denominator tracks shared gain.} This code-faithful forward-equation analysis
does not train DendriNet. \textbf{(A--B)}~The displayed readouts are the
implemented shunt, tangent-additive and optimized positive-E/I additive
comparators, and the leak-free shunting limit. With aligned gain sensing at
gain log-SD $0.45$, the shunting-minus-optimized-positive-E/I gap at $P=256$
changes from $-16.40$ at conductance scale $c=0.25$ to $+183.20$ at $c=64$.
Across $P=1$--256, high-conductance $d'^2$ grows $164.97\times$ for shunting
and $16.62\times$ for the additive comparator. \textbf{(C)}~The full surface
shows that both nonzero shared gain and sufficient conductance are required.
\textbf{(D)}~Private sensor load reduces the aligned gap from $183.20$ to
$21.52$, and an independent gain sensor nearly eliminates it. \textbf{(E)}~As
conductance rises, the effective covariance ratio falls from $0.796$ to
$0.00216$ for shunting but remains $0.0565$ for additive, a 26.1-fold difference
at $c=64$. \textbf{(F)}~A separate identical-statistic control with inaccessible
differential noise saturates at $23.26$, near the exact ceiling $1/\nu=25$.
The endpoint is linear $d'^2$, and increasing $P$ adds an E/I observation pair
rather than preserving a fixed synaptic budget.}
\label{fig:identified-ceiling-supp}
\end{figure}

\clearpage
\subsection{Resource-matched morphology and locality}
\label{supp:morphology-locality}

These figures test whether the aligned morphology result survives changes in
operating point, resource redistribution, and inhibitory support.

\begin{figure}[htbp]
\centering
\includegraphics[width=0.99\textwidth]{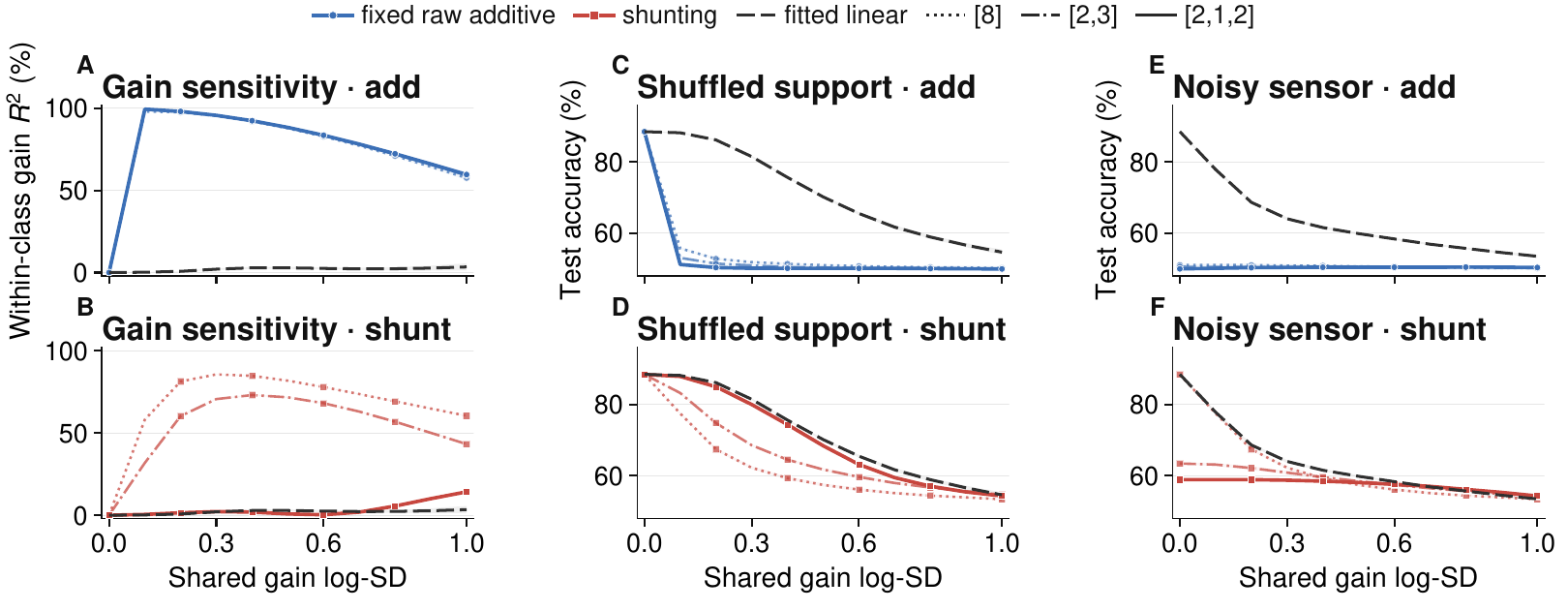}
\caption{\textbf{Aligned deep shunting reduces gain sensitivity, but the effect
depends on sensor support and fidelity.} \textbf{(A--B)}~Within-class gain
$R^2$ for additive and shunting models under aligned support. Fixed-raw
additive trees remain gain-sensitive, whereas deep shunting suppresses the
constructed nuisance mode. \textbf{(C--D)}~Shuffling sensor support removes the
aligned morphology advantage. \textbf{(E--F)}~Corrupting the gain sensor
reduces the shunting increment and changes its ordering relative to the fitted
linear reference.}
\label{fig:morphology-noise-diagnostics-supp}
\end{figure}

\begin{figure}[htbp]
\centering
\includegraphics[width=0.99\textwidth]{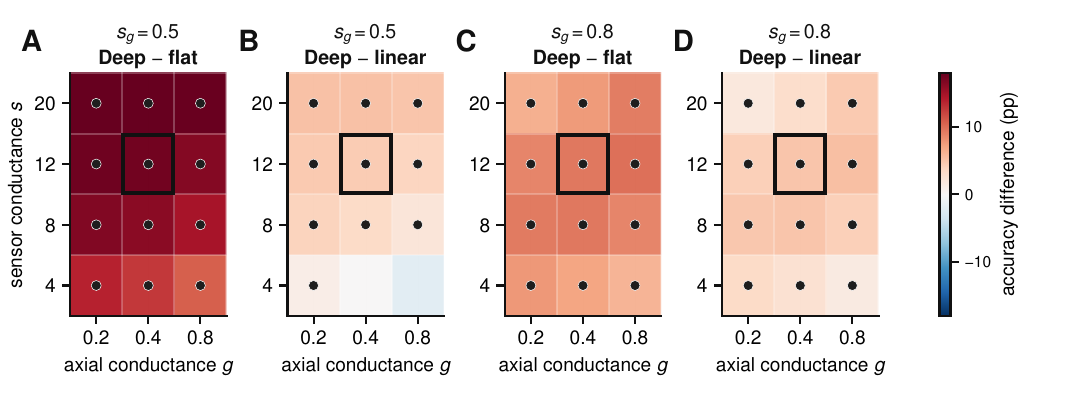}
\caption{\textbf{The aligned hierarchy persists across the tested passive
operating points but is not uniformly superior to additive readout.} \textbf{(A,C)}~Deep
minus flat shunting accuracy at gain log-SD $0.5$ and $0.8$.
\textbf{(B,D)}~Deep shunting relative to the topology-independent fitted linear
comparator. Rows vary sensor conductance and columns vary axial conductance; the
black outline marks the nominal point, and dots mark paired-seed 95\% $t$
intervals above zero. Across the 24 operating points, deep shunting exceeds
flat shunting in every cell and the fitted linear comparator in 23 cell means
(22 intervals exclude zero). This analysis tests operating-point sensitivity
of the designed aligned construction.}
\label{fig:exact-inventory-operating-sensitivity-supp}
\end{figure}

\begin{figure}[htbp]
\centering
\includegraphics[width=0.99\textwidth]{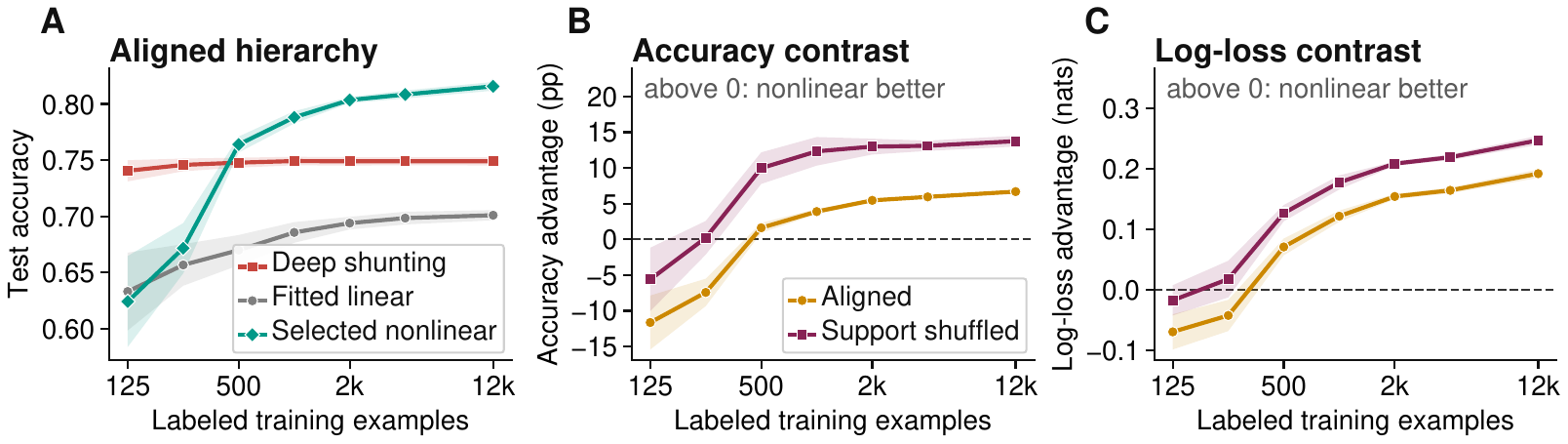}
\caption{\textbf{Passive capacity controls bound the designed hierarchy.}
\textbf{(A)}~In the aligned construction, the
structure-matched shunting readout has a small-sample accuracy advantage over
the tested validation-selected nonlinear predictors, which overtake it with
more labeled examples. \textbf{(B)}~The nonlinear-minus-shunting accuracy
contrast crosses zero under aligned support, whereas oracle support shuffling
restricts that advantage to the smallest tested sample. \textbf{(C)}~The
aligned log-loss contrast crosses in the same sample-size bracket, whereas
shuffling removes the low-sample log-loss advantage; positive values favor the
nonlinear predictor. Curves and intervals are means and paired 95\% $t$ intervals over
eight held-out computational seeds. All dendritic arms use passive branch
voltages with no post-voltage activation; the nonlinear predictors operate on
the raw observations and are capacity controls, not activated dendritic trees.}
\label{fig:capacity-reduction-supp}
\end{figure}

\begin{figure}[htbp]
\centering
\includegraphics[width=0.99\textwidth]{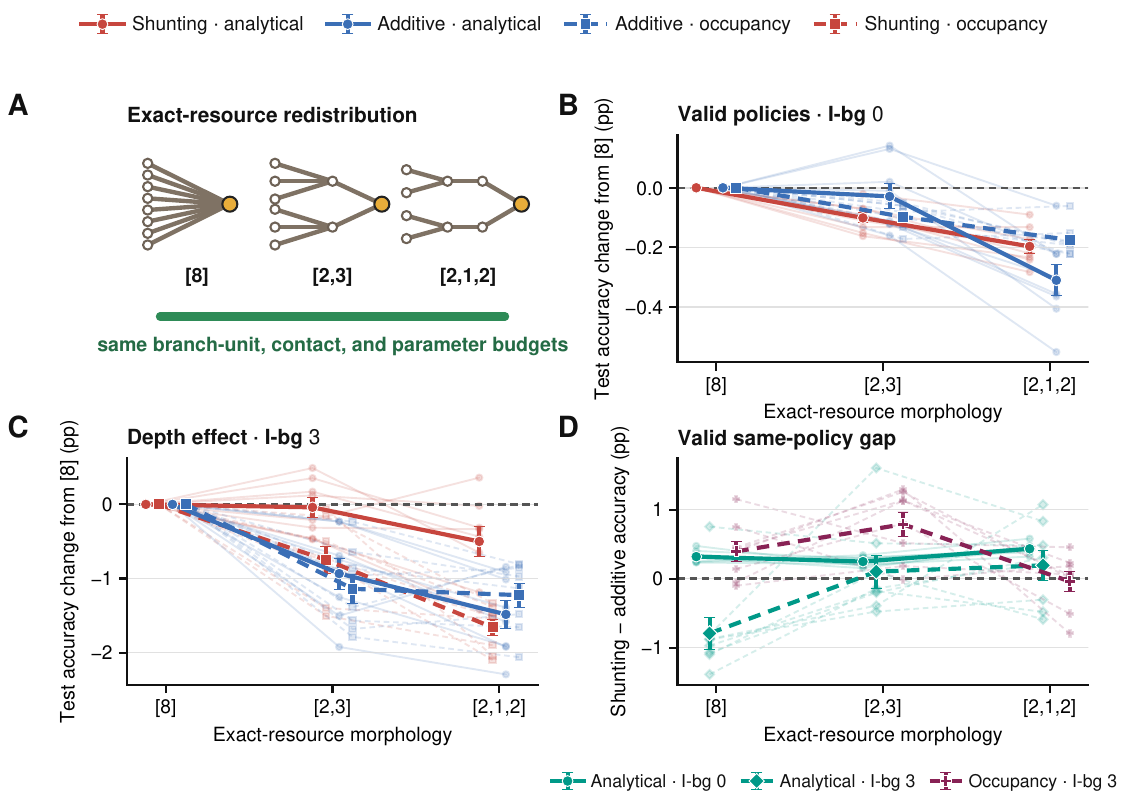}
\caption{\textbf{Redistributing a fixed branch budget does not produce a
generic depth benefit after standard-path DendriNet training.} All branches use
shifted-tanh post-voltage activation. \textbf{(A)}~The $[8]$, $[2,3]$, and
$[2,1,2]$ trees each contain eight branch units, 14,336 active contacts, 67,330
trainable parameters, and the same decoder; only path structure and terminal
allocation differ. \textbf{(B--C)}~Relative to $[8]$, all 14 valid
non-reference mean held-out test-accuracy contrasts are negative (eight paired
seeds). Among policy conditions realized for both mechanisms, all 12 held-out
test-accuracy contrasts are negative. Runs failing the prespecified gate-fit
validity check are excluded from cross-mechanism comparisons and audited in the
source data. An all-runs sensitivity retaining the 24 warning runs leaves all
16 non-reference cell means negative.
\textbf{(D)}~Within valid
same-policy cells, the shunting-minus-additive gap varies non-monotonically with
morphology and background.}
\label{fig:resource-matched-morph-supp}
\end{figure}

\begin{figure}[htbp]
\centering
\includegraphics[width=0.99\textwidth]{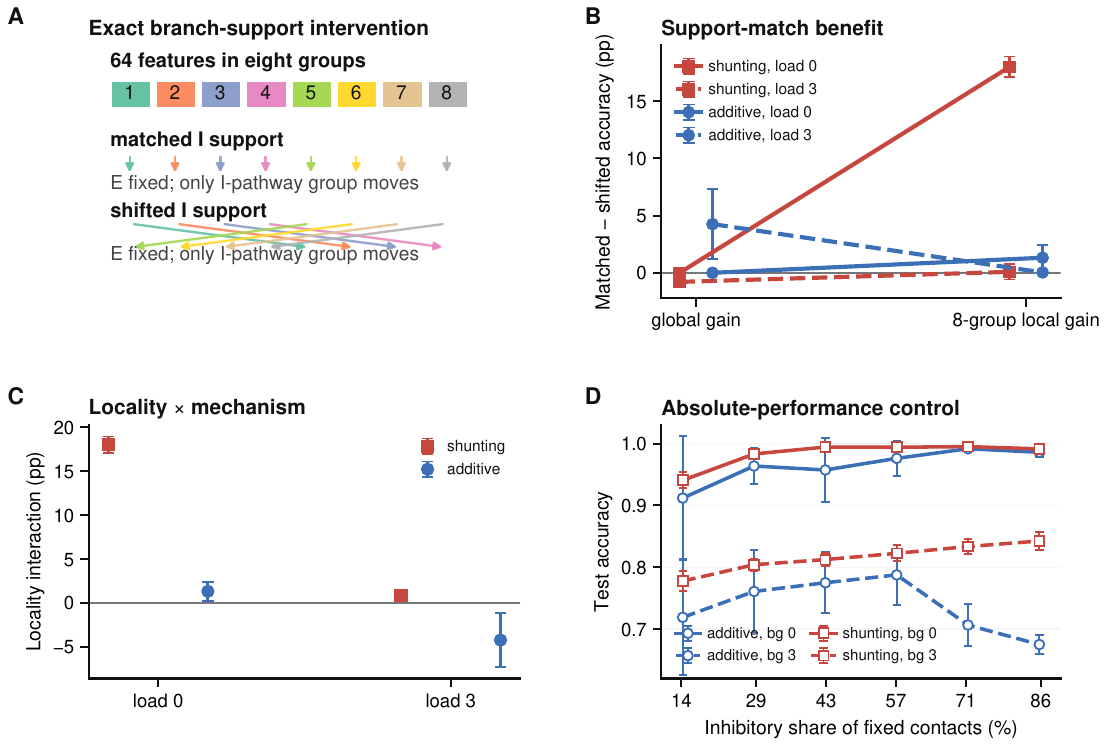}
\caption{\textbf{Standard-path trained, fixed-resource controls isolate inhibitory support and contact allocation.}
\textbf{(A)}~Excitatory support is fixed while eight inhibitory contacts are
shifted between matched and mismatched locations. \textbf{(B)}~Matching support
strongly improves shunting under eight local gain factors and no I-stream
background. \textbf{(C)}~At background 0, the mechanism difference in this
locality interaction is $+16.70$ percentage points; adding background changes it by
$-11.59$ percentage points. Additive has the higher mean in seven of eight
absolute-performance cells; the remaining shunting edge is $+0.09$ percentage
points with a 95\% interval spanning zero. Morphology, decoder width, parameters,
initialization, and active-contact counts are held fixed; post-voltage activation is
disabled. \textbf{(D)}~Absolute test accuracy for the separate fixed-contact
allocation control. The inhibitory share changes across six allocations while
$N_E+N_I=28$ contacts per branch, morphology, population width, and parameter
count remain fixed; solid and dashed lines denote I-stream background means 0
and 3. The value 28 is a computational budget, not a biological E/I-cell ratio.
Curves show means and 95\% $t$ intervals over eight matched seeds.}
\label{fig:locality-ei-fixed-supp}
\end{figure}

\clearpage
\subsection{Learned-network comparator matching and optimization}
\label{supp:learned-fairness}

These controls ask whether the trained comparison survives shared optimization
policies and stronger additive or normalization baselines.

\begin{figure}[htbp]
\centering
\includegraphics[width=0.99\textwidth]{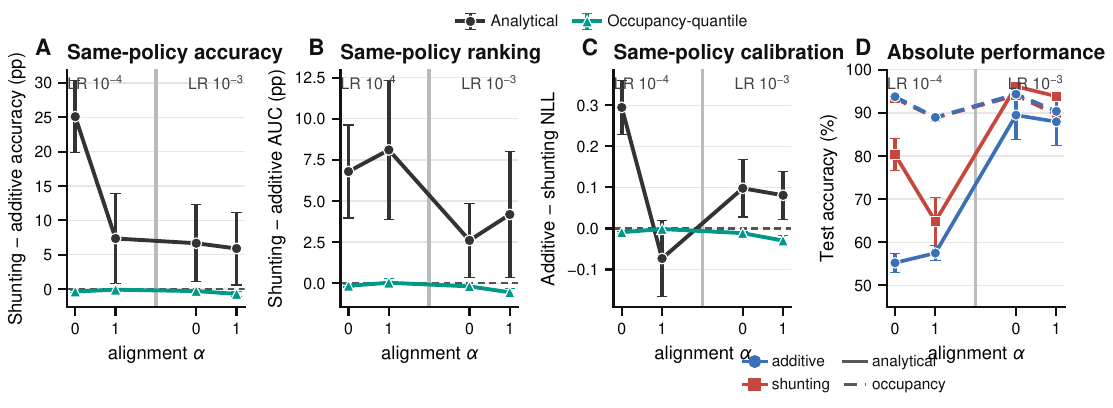}
\caption{\textbf{Supplementary setup-sensitivity control for the standard-path
trained DendriNet contextual comparison.} All models use shifted-tanh
activation after each branch voltage. \textbf{(A--C)}~Same-policy accuracy, AUC, and calibration
contrasts compare shunting with additive networks under analytical or
occupancy-quantile initialization; positive values favor shunting.
\textbf{(D)}~Absolute accuracy reveals undertraining at low learning rates.
Points show means $\pm$ paired-seed SEM ($n=8$). Mechanism comparisons use the
same named policy, avoiding cross-policy confounds; mechanism-specific equations
within each policy remain. Pooled validation selects shared occupancy at
$10^{-3}$, whose held-out shunting-minus-additive gap is $-0.50$ percentage
points.}
\label{fig:contextual-factorial-supp}
\end{figure}

\begin{figure}[htbp]
\centering
\includegraphics[width=0.99\textwidth]{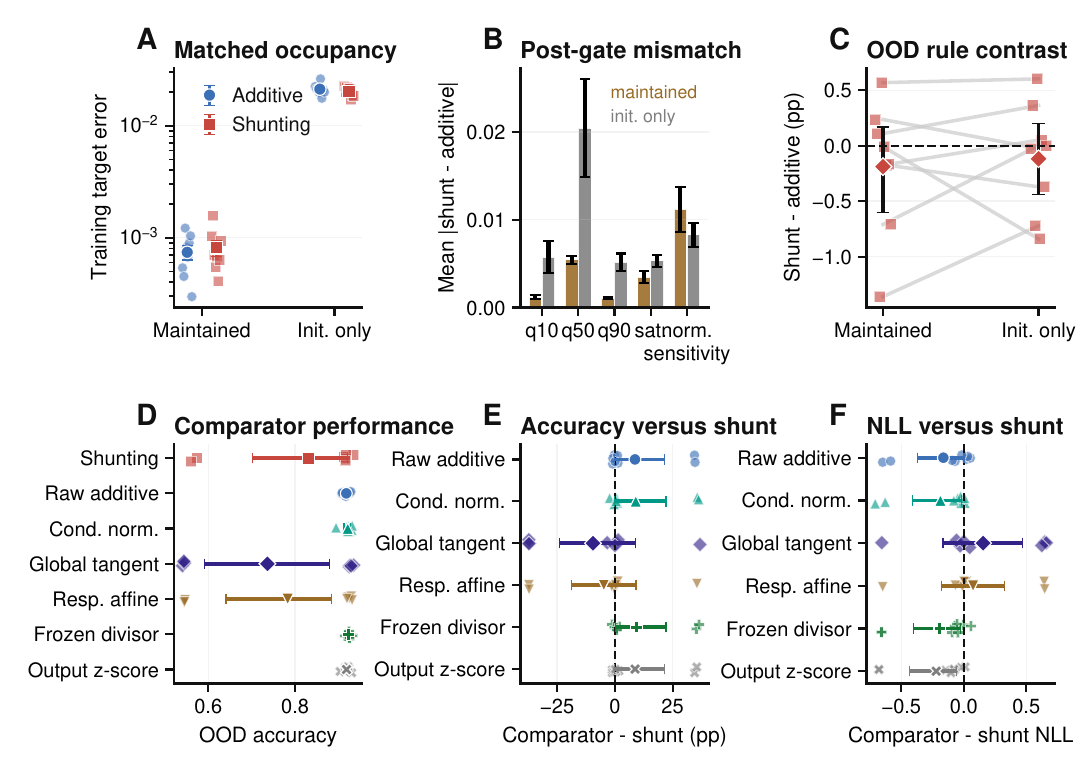}
\caption{\textbf{Matched Top-$K$ operating-point and comparator controls show no
robust rule advantage in standard-path trained DendriNet models.} Paired runs share initial parameters, Top-$K$ masks,
and $N_E/N_I=16/12$ contacts per branch. \textbf{(A--B)}~A label-free,
training-only adjustment maps pre-gate quantiles to common activation targets;
small post-gate differences in saturation and sensitivity remain.
\textbf{(C)}~Neither the adjusted nor initialization-only policy yields a
robust held-out shunting-minus-additive accuracy gap across eight paired seeds.
\textbf{(D--F)}~A seven-arm passive $[2,2]$ comparison holds the population,
decoder, data, optimizer, initial parameters, and masks fixed while varying the
integration rule or comparator. None of the six prespecified accuracy or log-loss
contrasts survives Holm correction ($p_{\rm Holm}\ge0.394$).}
\label{fig:topk-comparator-ladder-supp}
\end{figure}

\clearpage
\subsection{Additional V1 analyses}
\label{supp:v1-analyses}

These analyses test comparator and scale sensitivity, running-related
multiplicative variability, and topology-dependent noise robustness in V1.

\begin{figure}[htbp]
\centering
\includegraphics[width=0.99\textwidth]{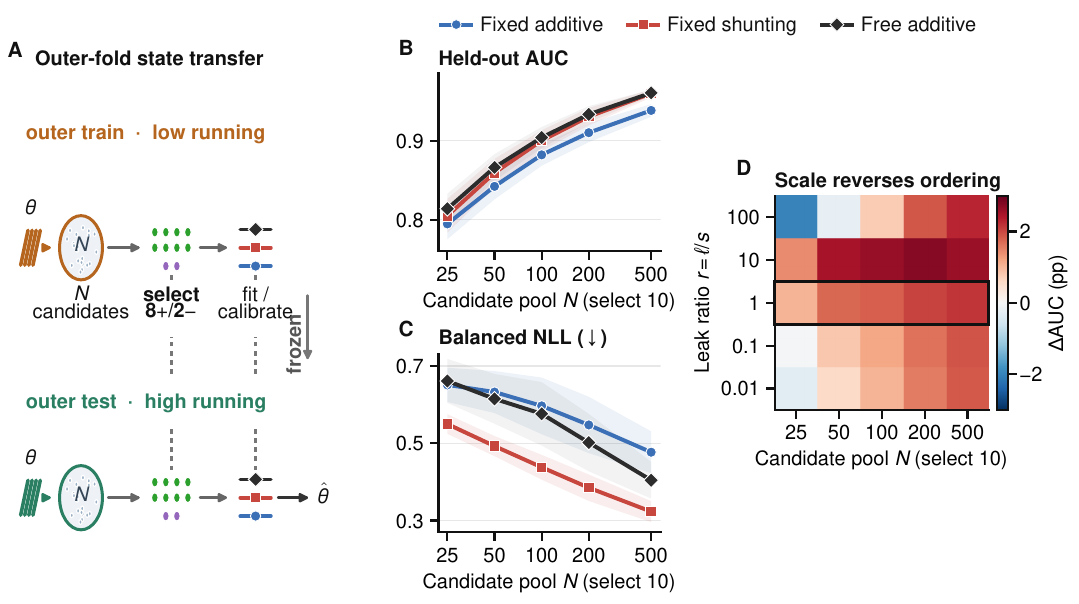}
\caption{\textbf{Fixed-template V1 transfer depends on comparator and response
scale.} \textbf{(A)}~Low-running outer-training trials select eight
positive-evidence and two negative-evidence responses and fit fixed $E-I$,
fixed $E/(1+E+I)$, and fitted $E-\rho I-b$ readouts;
held-out high-running trials test the fixed models. Increasing $N$ changes the
candidate pool rather than the ten-response readout budget.
\textbf{(B--C)}~Held-out AUC and balanced log loss produce different comparator
orderings (session mean $\pm$ SEM; $n=3$). \textbf{(D)}~The
shunting-minus-fixed-additive AUC contrast changes sign with the prespecified
leak/input-scale ratio and candidate access. This fixed-feature transfer test
does not identify biophysical shunting.}
\label{fig:v1-fixed-template-supp}
\end{figure}

\begin{figure}[htbp]
\centering
\includegraphics[width=0.99\textwidth]{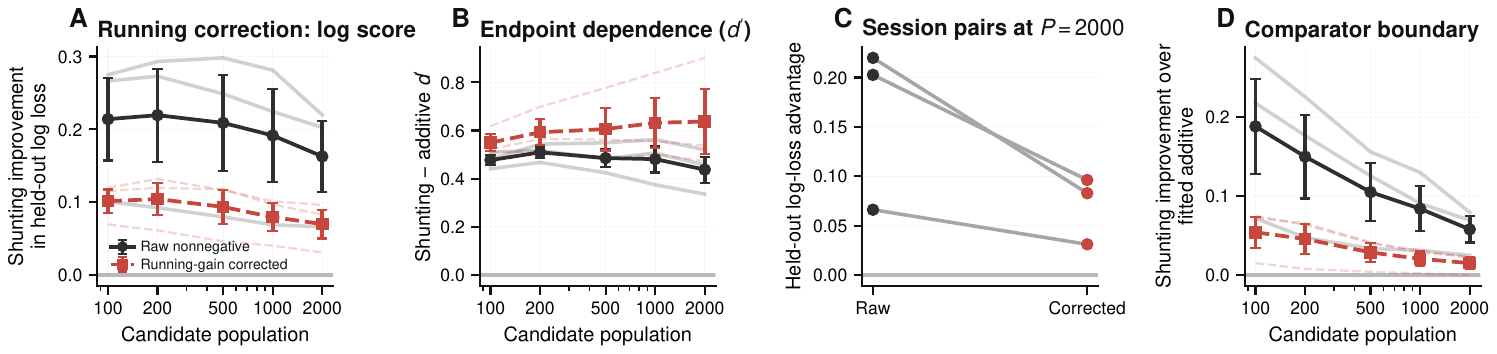}
\caption{\textbf{A training-estimated multiplicative correction removes much of
the running-linked V1 difference.} All response selection, nuisance
estimation, score orientation, thresholding, comparator fitting, and
calibration use outer-training data; high-running trials are held out.
Thin lines show recordings, and thick points show the unweighted recording mean
$\pm$ SEM ($n=3$). \textbf{(A)}~The shunting reduction in held-out log loss
relative to fixed additive falls
by 53--59\% after correction across five candidate-pool sizes, with the same
direction in all recordings. \textbf{(B)}~The shunting-minus-additive $d'$ moves in the
opposite direction, revealing an endpoint difference between standardized
separation and calibrated log loss. \textbf{(C)}~At candidate access 2000, the log-loss advantage
shrinks in all three recordings. \textbf{(D)}~The log-loss difference from a
training-fitted $E-\rho I-b$ comparator also approaches zero after correction.
This fixed-feature analysis tests sensitivity to running-linked multiplicative
variability rather than identifying the underlying circuit.}
\label{fig:v1-running-recovery-supp}
\end{figure}

\begin{figure}[htbp]
\centering
\includegraphics[width=0.98\textwidth]{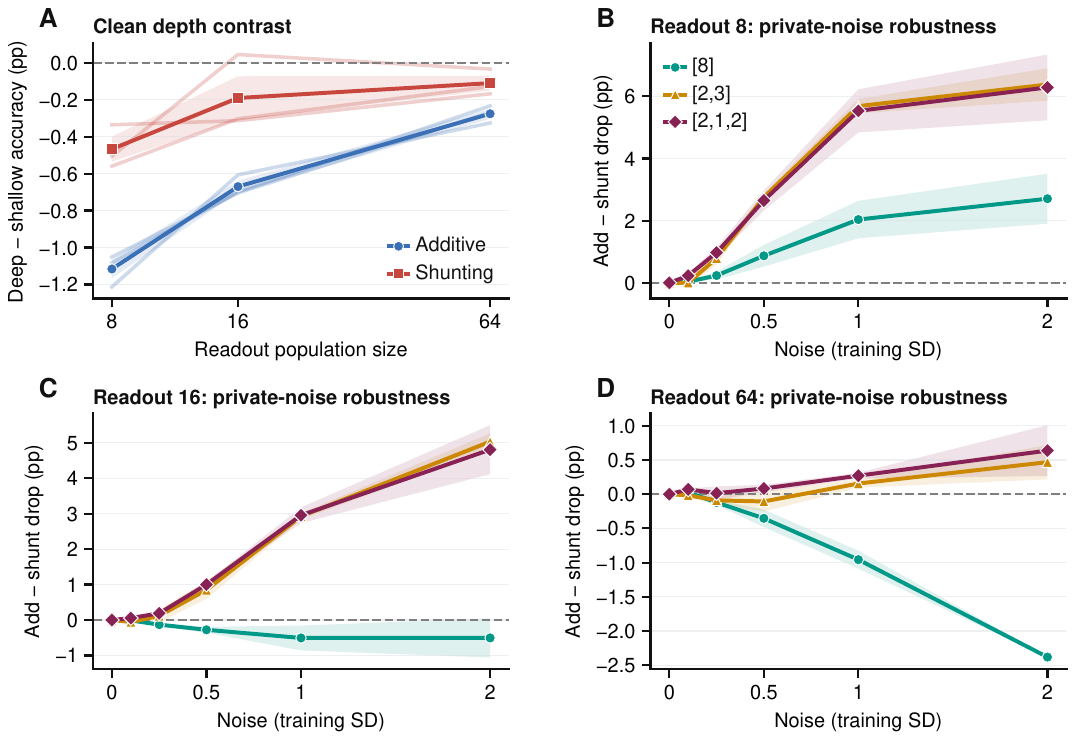}
\caption{\textbf{Topology changes relative robustness to private noise despite
a clean-accuracy cost.} These panels use standard-path
trained DendriNet fits with shifted-tanh activation after every branch voltage. The
$[8]$, $[2,3]$, and
$[2,1,2]$ trees each contain eight branch
units and match active contacts, candidate slots, and trainable parameters.
Bands show mean $\pm$ SEM across three recordings after averaging four seeds
within each recording. \textbf{(A)}~Deep-minus-shallow clean accuracy is
negative at every readout width for both mechanisms. \textbf{(B--D)}~Each curve
shows the additive-minus-shunting difference in accuracy degradation, where
degradation is clean minus noisy accuracy; positive values indicate less
degradation for shunting. At $P_E=64$ and two training
SDs, the shallow topology favors additive by $2.37$ percentage points, whereas the
intermediate and deep topologies favor shunting by $0.39$ and $0.54$ percentage points.
Topology can therefore alter robustness without improving clean
accuracy. The primary categorical-log-loss endpoint at the same width and
noise level favors
additive for every topology, so the accuracy crossover is endpoint-specific.}
\label{fig:v1-exact-resource-depth-noise-supp}
\end{figure}

\end{document}